\documentclass[11pt,a4paper]{article}

\usepackage[margin=1in]{geometry}
\usepackage{palatino}
\usepackage{mathpazo} 
\usepackage{partmac}
\usepackage{subcaption}

\usepackage{enumitem}

\usepackage[
  colorlinks,
  linkcolor = blue,
  citecolor = blue,
  urlcolor = blue]{hyperref}

\usepackage{amsmath,amssymb,amsthm}
\usepackage{mathtools}
\mathtoolsset{centercolon}

\usepackage[affil-it]{authblk}

\DeclarePairedDelimiter{\set}{\lbrace}{\rbrace}
\DeclarePairedDelimiter{\abs}{\lvert}{\rvert}

\newcommand{\be}{\begin{equation}}
\newcommand{\ee}{\end{equation}}

\renewcommand{\P}{\mathcal{P}}
\newcommand{\G}{\mathcal{G}}
\renewcommand{\H}{\mathbf{H}}
\newcommand{\K}{\mathbf{K}}
\renewcommand{\vec}[1]{\mathbf{#1}}

\DeclareMathOperator{\spn}{span}

\newcommand{\Thm}[1]{\hyperref[thm:#1]{Theorem~\ref*{thm:#1}}}
\newcommand{\Lem}[1]{\hyperref[lem:#1]{Lemma~\ref*{lem:#1}}}
\newcommand{\Cor}[1]{\hyperref[cor:#1]{Corollary~\ref*{cor:#1}}}
\newcommand{\Def}[1]{\hyperref[def:#1]{Definition~\ref*{def:#1}}}
\newcommand{\Obs}[1]{\hyperref[obs:#1]{Observation~\ref*{obs:#1}}}
\newcommand{\Prop}[1]{\hyperref[prop:#1]{Proposition~\ref*{prop:#1}}}
\newcommand{\Rem}[1]{\hyperref[rem:#1]{Remark~\ref*{rem:#1}}}
\newcommand{\Ex}[1]{\hyperref[ex:#1]{Example~\ref*{ex:#1}}}

\newcommand{\Sec}[1]{\hyperref[sec:#1]{Section~\ref*{sec:#1}}}

\newcommand{\Fig}[1]{\hyperref[fig:#1]{Figure~\ref*{fig:#1}}}
\newcommand{\Tab}[1]{\hyperref[tab:#1]{Table~\ref*{tab:#1}}}
\newcommand{\EqRef}[1]{\hyperref[eq:#1]{(\ref*{eq:#1})}}
\newcommand{\Eq}[1]{Equation~\hyperref[eq:#1]{(\ref*{eq:#1})}}

\makeatletter
\newtheorem*{rep@theorem}{\rep@title}
\newcommand{\newreptheorem}[2]{%
\newenvironment{rep#1}[1]{%
 \def\rep@title{#2 \ref{##1}}%
 \begin{rep@theorem}}%
 {\end{rep@theorem}}}
\makeatother

\newreptheorem{theorem}{Theorem}
\newreptheorem{lemma}{Lemma}

\newtheorem{theorem}{Theorem}[section]
\newtheorem*{theorem*}{Theorem}
\newtheorem{lemma}[theorem]{Lemma}
\newtheorem{cor}[theorem]{Corollary}

\theoremstyle{definition}
\newtheorem{definition}[theorem]{Definition}
\newtheorem{remark}[theorem]{Remark}
\newtheorem{example}[theorem]{Example}

\DeclareMathOperator{\aut}{Aut}
\DeclareMathOperator{\qut}{Qut}

\newcommand{\schur}{\bullet}

\DeclareMathOperator{\rel}{rel}

\title{Quantum isomorphism is equivalent to equality \\ of homomorphism counts from planar graphs}

\author[1]{Laura Man\v{c}inska}
\author[2]{David E.~Roberson} 

\affil[1]{QMATH, Department of Mathematical Sciences, University of Copenhagen, Universitetsparken 5, 2100 Copenhagen \O, Denmark}
\affil[2]{Department of Applied Mathematics and Computer Science, Technical University of Denmark, DK-2800 Lyngby, Denmark}

\begin{document}

\maketitle

\begin{abstract}
Over 50 years ago, Lov\'{a}sz proved that two graphs are isomorphic if and only if they admit the same number of homomorphisms from any graph [Acta Math.~Hungar.~18 (1967), pp.~321--328]. 
In this work we prove that two graphs are \emph{quantum isomorphic} (in the commuting operator framework) if and only if they admit the same number of homomorphisms from any planar graph. As there exist pairs of non-isomorphic graphs that are quantum isomorphic, this implies that homomorphism counts from planar graphs do not determine a graph up to isomorphism. 
Another immediate consequence is that determining whether there exists some planar graph that has a different number of homomorphisms to two given graphs is an undecidable problem, since quantum isomorphism is known to be undecidable. Our characterization of quantum isomorphism is proven via a combinatorial characterization of the intertwiner spaces of the quantum automorphism group of a graph based on counting homomorphisms from planar graphs. This result inspires the definition of \emph{graph categories} which are analogous to, and a generalization of, partition categories that are the basis of the definition of easy quantum groups. Thus we introduce a new class of \emph{graph-theoretic quantum groups} whose intertwiner spaces are spanned by maps associated to (bi-labeled) graphs. Finally, we use our result on quantum isomorphism to prove an interesting reformulation of the Four Color Theorem: that any planar graph is 4-colorable if and only if it has a homomorphism to a specific Cayley graph on the symmetric group $S_4$ which contains a complete subgraph on four vertices but is not 4-colorable.
\end{abstract}


%

\section{Introduction}

In this work we present a surprising connection between the theory of entanglement-assisted strategies for nonlocal games, quantum group theory, and homomorphism counts from combinatorics. The story begins with the \emph{graph isomorphism game}, a nonlocal game introduced in~\cite{qiso1} in which two players attempt to convince a referee that they know an isomorphism between two given graphs $G$ and $H$. Classically, the players can succeed with probability one if and only if the graphs are indeed isomorphic. This motivates the definition of \emph{quantum isomorphic} graphs: those pairs of graphs for which the game can be won perfectly when the players are able to make local quantum measurements on a shared entangled state. It is far from obvious that there exists any pair of quantum isomorphic graphs that are not isomorphic as well. However, in~\cite{qiso1} a reduction from the linear system games of Cleve, Liu, and Slofstra~\cite{cleveslofstra} to the isomorphism game was given which provided an infinite family of quantum isomorphic but non-isomorphic graphs.

Soon after their introduction, deep connections between the notion of quantum isomorphism and the theory of \emph{quantum automorphism groups} of graphs were established in~\cite{morita} and~\cite{qperms}. In particular, in the latter work an analog of a well-known classical result was proven: that connected graphs $G$ and $H$ are quantum isomorphic if and only if the quantum automorphism group of their disjoint union, $\qut(G \cup H)$, has an orbit which intersects both $V(G)$ and $V(H)$ nontrivially. 
We investigate an algebraic generalization of orbits known as \emph{intertwiners} of the quantum automorphism group of a graph. Our first main result is a combinatorial characterization of these intertwiners in terms of \emph{planar bi-labeled graphs} (see \Def{LGraph} and \Def{PBG}):
\begin{reptheorem}{thm:onecat}[Informal]
The intertwiners of the quantum automorphism of a graph $G$ are the span of matrices whose entries count homomorphisms from planar graphs to $G$, partitioned according to the images of certain labeled vertices of the planar graph.
\end{reptheorem}
\noindent As we will see in \Sec{bilabeledgraphs} (see \Rem{drawing}), the above provides a pictorial way of representing the intertwiners of the quantum automorphism group of a graph, analogously to the pictures of partitions used for easy quantum groups. Whether such a description was possible has been an implicit open question in the quantum group community. We remark that some of the ideas used here, though arrived at independently, appear to be similar to ideas from the theory of graph limits~\cite{lovasz2012large}.

\Thm{onecat} is additionally the basis of our combinatorial characterization of quantum isomorphism. We show that the measurement operators from a winning quantum strategy for the isomorphism game can be used to construct a linear map between intertwiners of $\qut(G)$ and $\qut(H)$. Moreover, this linear map takes an intertwiner of $\qut(G)$ associated to a given planar bi-labeled graph to the intertwiner of $\qut(H)$ associated to the same planar bi-labeled graph (\Lem{weakiso}). Combining this linear map with \Thm{onecat} and making use of the aforementioned result of~\cite{qperms} regarding orbits of $\qut(G \cup H)$, we are able to prove the following remarkable result:
\begin{reptheorem}{thm:main}[Abbreviated]
Graphs $G$ and $H$ are quantum isomorphic if and only if they admit the same number of homomorphisms from any planar graph.
\end{reptheorem}
The above supplies a completely combinatorial description of quantum isomorphism, and thus of the advantage provided by entanglement in the setting of the isomorphism game. We hope that Theorems~\ref{thm:onecat} and~\ref{thm:main}, as well as the other results in this work, will help to attract researchers from discrete mathematics to the growing intersection of quantum information, quantum groups, and combinatorics.

%

\subsection{Context and consequences}

Our characterization of quantum isomorphism can be viewed as a quantum analog of a classic theorem of Lov\'{a}sz~\cite{lovasz1967operations}: graphs $G$ and $H$ are isomorphic if and only if they have the same number of homomorphisms from any graph. By restricting your homomorphism counts from all graphs to specific families, one can obtain relaxations of isomorphism. Trivial examples include counting homomorphisms from just the single vertex graph or the two vertex graph with a single edge, which simply test whether $G$ and $H$ have the same number of vertices or edges respectively. Less trivially, counting homomorphisms from all star graphs or from all cycles determines a graph's degree sequence or spectrum respectively, the latter being a classical result of algebraic graph theory. Very recently, a surprising result of this form was proven by Dell, Grohe, and Rattan~\cite{DGR}: graphs $G$ and $H$ are not distinguished by the $k$-dimensional Weisfeiler-Leman algorithm if and only if they admit the same number of homomorphisms from all graphs with treewidth at most $k$. Interestingly, in the same work they asked whether homomorphism counts from planar graphs determine a graph up to isomorphism. Though we were not motivated by this question, having only learned of it later, our results answer it in the negative. As far as we are aware, this is the first example of using tools from quantum groups to solve a problem in graph theory.

We remark that in all of the above examples (excluding our result) there is an at worst quasi-polynomial time algorithm for determining if the two graphs $G$ and $H$ have the same number of homomorphisms from any graph in the given family. For all but Lov\'{a}sz' result there is even a polynomial time algorithm (and of course there may possibly be a polynomial time algorithm for graph isomorphism). This may be unexpected in the case where the families are infinite, as {\it{a priori}} one must check an infinite number of conditions. In stark contrast, the known undecidability of quantum isomorphism~\cite{qiso1} along with our result implies that determining if there exists some planar graph having a different number of homomorphisms to two given graphs is an undecideable problem. This also implies that there is no computable function of two graphs $G$ and $H$ which gives an upper bound on the size of planar graphs that must be checked to determine whether $G$ and $H$ are quantum isomorphic. In the classical case, it always suffices to count homomorphisms from graphs on $|V(G)|$ or fewer vertices.

As a consequence of our characterization of quantum isomorphism, we show that a planar graph has a 4-coloring if and only if 
it has a homomorphism to a particular Cayley graph $H$ for the symmetric group $S_4$. For an arbitrary graph $G$, it is strictly easier to have a homomorphism to $H$ than to have a 4-coloring, i.e., every graph with a 4-coloring has a homomorphism to $H$ but the converse fails for some (in fact infinitely many) graphs. Thus it seems that this gives a nontrivial reformulation of the famous Four Color Theorem.\\


Quantum automorphism groups of graphs belong to the class of \emph{compact matrix quantum groups}, introduced by Woronowicz~\cite{woronowicz}. By a version of Tannaka-Krein duality proven by Woronowicz~\cite{woronowicz1988}, these quantum groups are completely determined by their intertwiner spaces. This helps to underline the importance of our characterization of the intertwiners of $\qut(G)$. Our \Thm{onecat} echos a transformative result of Banica and Speicher from 2009~\cite{noncross} in which they gave combinatorial characterizations of the intertwiners of several groups and their quantum analogs. Though the specific characterizations were important results, the lasting impact of the work is due to the introduction of the so-called \emph{easy quantum groups} (see \Sec{qautogroups} for more details), whose intertwiner spaces are given by maps associated to partitions. The ability to exploit the combinatorial structure of these partitions has led to easy quantum groups being extensively studied (e.g.~\cite{banica2010classification,Banica2011,freslon2016representation,raum2013easy,raum2014combinatorics,WEBER2013500}), with their complete classification (in the orthogonal case) given in~\cite{Raum2016}.

Motivated by our \Thm{onecat}, we introduce the notion of \emph{graph-theoretic quantum group}: those quantum groups whose intertwiners are given by some family of bi-labeled graphs (closed under the appropriate operations). These are not merely similar to the notion of easy quantum groups, we show in \Sec{graphqgroups} that they are a vast generalization: easy quantum groups are the graph-theoretic quantum groups corresponding to families of \emph{edgeless} bi-labeled graphs. We therefore expect graph-theoretic quantum groups to be a much richer class than easy quantum groups, while still retaining the underlying combinatorial structure that made the latter such fruitful ground for research.\\



In quantum information, nonlocal games\footnote{More commonly known as Bell inequalities within the physics community.}, like the isomorphism game, are used to study the capabilities and limitations of entanglement, as well as to elucidate the differences between the various mathematical models of joint measurement scenarios. Recently, incredible progress has been made by Slofstra, who investigated (binary) linear system games. He was able to resolve two longstanding open questions, namely he solved Tsirelson's problem~\cite{slofstra16} and showed that the set of quantum correlations is not closed~\cite{nonclosureslofstra}. The key to this work was the discovery that quantum strategies for linear system games could be characterized in group theoretic language~\cite{cleveslofstra}, thus allowing for the use of powerful tools from combinatorial group theory. Analogously, we make use of the quantum group theoretic characterization of quantum isomorphism from~\cite{qperms} to establish our combinatorial characterization from \Thm{main}. Interestingly, it was shown in~\cite{qiso1} that linear system games can be reduced to isomorphism games. Combining this with our characterization of quantum isomorphism, it follows that the group-theoretic condition from~\cite{cleveslofstra} has a combinatorial reformulation in terms of homomorphism counts from planar graphs. More specifically, the central element $J_\Gamma$ of a solution group $\Gamma$ is equal to the identity if and only if there exists a planar graph which has a different number of homomorphisms to $G_\Gamma$ and $H_\Gamma$ -- the two graphs arising from the reduction of a linear system game to an isomorphism game.

%
%
%
%
%
%
%

\subsection{Outline and preliminaries}\label{sec:outline}

Given a $C^*$-algebra $A$, we denote by $M_n(A)$ the $n \times n$ matrices that have entries from $A$. To distinguish from the usual case of complex valued matrices, we will typically denote elements of $M_n(A)$ by script letters such as $\mathcal{U}$, and we will write $\mathcal{U} = (u_{ij})$ to denote that the $ij$-entry of $\mathcal{U}$ is $u_{ij} \in A$. We will often want to multiply an algebra-valued matrix $\mathcal{U} = (u_{ij})$ by a complex-valued matrix $M = (m_{\ell k})$. As long as the dimensions match, this is a well-defined product given by
\[(M\mathcal{U})_{ij} = \sum_k m_{ik}u_{kj},\]
and similarly for multiplying in the other order. It is straightforward to check that this product satisfies $M(N\mathcal{U}) = (MN)\mathcal{U}$. If $\mathcal{U},\mathcal{V} \in M_n(A)$, then we may also define their product entrywise as $(\mathcal{U}\mathcal{V})_{ij} = \sum_k u_{ik}v_{kj}$. If $\mathcal{U} \in M_n(A),\mathcal{V} \in M_k(A)$, then we let $\mathcal{U}\otimes \mathcal{V}$ be the matrix in $M_{nk}(A)$ whose $i_1i_2,j_1j_2$-entry is $u_{i_1j_1}v_{i_2j_2}$. It is again straightforward to check that $(M \otimes N)(\mathcal{U} \otimes \mathcal{V}) = (M\mathcal{U}) \otimes (N\mathcal{V})$ whenever these products are defined. The transpose $\mathcal{U}^T$ and conjugate transpose $\mathcal{U}^*$ of a matrix $\mathcal{U} = (u_{ij}) \in M_n(A)$ are the matrices whose $ij$-entries are $u_{ji}$ and $u_{ji}^*$ respectively.

We use $[n]$ to denote the set $\{1,2,\ldots,n\}$. We will denote vectors/tuples by boldface characters, such as $\vec{a} = (a_1, \ldots, a_k)$. One other place we will utilize boldface is in denoting the identity element of a unital $C^*$-algebra as $\vec{1}$.\\

The rest of the paper is organized as follows. In \Sec{background} we introduce the relevant background on graph theory, quantum groups, and quantum information needed for our results. In \Sec{bilabeledgraphs} we define two of the central concepts of this work: bi-labeled graphs and homomorphism matrices, and we give several pertinent examples. We also define the operations of composition, tensor product, transpose, and Schur product of bi-labeled graphs in \Sec{blgops}, and we prove that these correspond to the analogous operations on homomorphism matrices in \Sec{correspondence}. In \Sec{planar} we present some of the basic theory of planar graphs that we will need for the proofs in the following sections. \Sec{planarblg} introduces the class of planar bi-labeled graphs and gives several examples. We prove that the class of planar bi-labeled graphs is closed under composition, tensor product, and transposition in \Sec{closureprops}, thus showing that the corresponding homomorphism matrices form a tensor category with duals. Our characterization of the intertwiners of the quantum automorphism group of a graph, \Thm{onecat}, is proven in \Sec{tale}, with a few additional results given in \Sec{further}. In \Sec{qiso}, we prove our characterization of quantum isomorphism in \Thm{main}. We finish up with \Sec{discuss} where we discuss some consequences of our results, including a reformulation of the Four Color Theorem, and possible future directions.

\section{Background}\label{sec:background}

\subsection{Graph theory basics}

The basic structure dealt with in this work is that of a graph. For us, a graph $G$ consists of a vertex set $V(G)$ and edge set $E(G)$. We do not allow multiple edges but we do allow loops (at most one per vertex), thus the elements of $E(G)$ are unordered pairs and singletons. We write $u \sim v$ if $u$ and $v$ are adjacent, which includes the case where $u = v$ and $u$ has a loop. If we need to specify adjacency in a particular graph then we write $u \sim_G v$, for instance. As is common practice we will refer to the edge between \emph{distinct} vertices $u$ and $v$ as $uv$, instead of $\{u,v\}$. Given a graph $G$, we define its \emph{complement}, denoted $\overline{G}$, as the graph on $V(G)$ which has the same loops as $G$ but in which distinct vertices are adjacent if and only if they are not adjacent in $G$. Thus the complement of a graph without loops is the same for us as the usual notion of complement for graphs without loops. We define the \emph{full complement} of $G$, denoted $\overline{\overline{G}}$, to be the graph where two vertices are adjacent if and only if they are not adjacent in $G$ (thus looped vertices become non-looped and vice versa). We will use the term \emph{multigraph} when we allow multiple edges (and multiple loops). In this case the elements of $E(G)$ have names such as $e_1, \ldots, e_m$ and there is a function $\phi_G\colon E(G) \to \binom{V(G)}{2} \cup \binom{V(G)}{1}$ indicating which vertices an edge is incident to. Here $\binom{V(G)}{k}$ is the set of subsets of $V(G)$ of size $k$. However it will often be more convenient to simply refer to an edge $e$ between vertices $u$ and $v$, without referring to any such function $\phi_G$.

A \emph{path of length $k$} in a multigraph $G$ is a sequence of $k$ edges $e_1, \ldots, e_k$ such that there are \emph{distinct} vertices $u_1, \ldots, u_{k+1}$ where $\phi_G(e_i) = \{u_i,u_{i+1}\}$ for all $i \in [k]$. A \emph{cycle of length $k > 1$} is a sequence of $k$ \emph{distinct} edges $e_1, \ldots, e_k$ such that there are $k$ \emph{distinct} vertices where $\phi_G(e_i) = \{u_i,u_{i+1}\}$ for $i \in [k-1]$ and $\phi_G(e_k) = \{u_k,u_1\}$. Note that this allows for cycles of length 2 which consist of two vertices and two distinct edges between them. We also consider a vertex with a loop as a cycle of length 1. In a graph, the edges between consecutive vertices in a path/cycle are uniquely determined, and so in this setting we will usually refer to paths of length $k-1$ and cycles of length $k$ simply by a sequence of $k$ distinct vertices $u_1, \ldots, u_k$. The \emph{complete graph} on $n$ vertices, denoted $K_n$, is the graph with vertex set $[n]$ (or any $n$-element set), having an edge between every pair of distinct vertices, but no loops. An \emph{empty/edgeless} graph is a graph whose edge set is the empty set.

Given a graph $G$, its \emph{adjacency matrix}, denoted $A_G$ is a symmetric $01$-matrix whose $uv$-entry is 1 if $u \sim v$ and is 0 otherwise. Note that this means there are 1's in the diagonal entries corresponding to loops. The adjacency matrices of $\overline{G}$ and $\overline{\overline{G}}$ are given by $A_{\overline{G}} = J - I - A_G$ and $A_{\overline{\overline{G}}} = J - A_G$ where $J$ is the all ones matrix.

Given a multigraph $G$, a subgraph (we will not use the term ``submultigraph") $H$ is a multigraph with $V(H) \subseteq V(G)$,  and $E(H) \subseteq E(G)$ such that $E(H)$ only contains edges that are not incident to any vertices of $V(G) \setminus V(H)$. If $U \subseteq V(G)$, then the subgraph $H$ of $G$ \emph{induced} by $U$ has vertex set $U$ and $E(H)$ is the sub(multi)set of $E(G)$ consisting of all edges that are incident only to vertices in $U$. Given two graphs $G$ and $H$, their \emph{disjoint union} is the graph with vertex set $V(G) \cup V(H)$ and edge set $E(G) \cup E(H)$. Here we are implicitly assuming that $V(G)$ and $V(H)$ are disjoint.

We will make frequent use of edge deletions and contractions in our work. Given a multigraph $G$ and a subset $S \subseteq E(G)$, the multigraph $G\setminus S$ resulting from \emph{deleting} $S$ has vertex set $V(G)$ and edge set $E(G) \setminus S$. Edge contraction is more complicated to explain but we use the usual notion. The multigraph $G/S$ resulting from \emph{contracting} $S$ has as its vertex set the connected components of the multigraph $H$ with vertex set $V(G)$ and edge set $S$. The edge set of $G/S$ can be identified with $E(G)\setminus S$, but if $e \in E(G) \setminus S$ was incident to vertices $u$ and $v$ in $G$, then in $G/S$ the edge $e$ is incident to the connected components of $H$ containing $u$ and $v$ respectively (if these are the same then $e$ becomes a loop). If $e$ was a loop incident to $u$ in $G$, then in $G/S$ it is a loop incident to the connected component of $H$ containing $u$. Note that though we have defined the vertices of $G/S$ as the connected components of the multigraph $H$, in practice we will often simply give the vertices corresponding to components that are not single vertices new names, so that referring to them is simpler. We will also use the phrase ``the vertex of $G/S$ that $u$ became after contraction" to refer to the vertex of $G/S$ corresponding to the connected component of $H$ containing $u$. It is important to point out that even if we begin with a graph $G$, edge contractions can result in a multigraph. To return to the class of graphs, we then perform \emph{simplification}: we replace any multiple edges with single edges, but we keep loops even if they were formed by the contractions (though we keep at most one loop per vertex). We will use $G\setminus e$ and $G/e$ when the set $S$ we are deleting/contracting has only one element.

One of the most important notions from graph theory that we will make use of is that of graph homomorphisms:

\begin{definition}
A \emph{homomorphism} from a graph $H$ to a graph $G$ is a function $\varphi\colon V(H) \to V(G)$ such that $u \sim_H v$ implies that $\varphi(u) \sim_G \varphi(v)$.
\end{definition}

Note that this means that a homomorphism can map vertices with loops only to vertices with loops, and that two distinct adjacent vertices can be mapped to the same vertex only if it has a loop. We will write $\varphi\colon H \to G$ to denote that $\varphi$ is a homomorphism from $H$ to $G$. We can also define homomorphisms of multigraphs. The condition that adjacency is preserved is the same, but if $u$ and $v$ are adjacent in $H$ then for each edge $e$ between them we have a choice of which edge between $\varphi(u)$ and $\varphi(v)$ to map $e$ to. This choice makes a difference when counting homomorphisms, however we will only consider the case when $G$ is a graph, thus all such choices are determined.

\begin{definition}
An \emph{isomorphism} from a graph $H$ to a graph $G$ is a bijection $\varphi\colon V(H) \to V(G)$ such that $u \sim_H v$ if and only if $\varphi(u) \sim_G \varphi(v)$.
\end{definition}

Note that this means that an isomorphism must map loops to loops and non-loops to non-loops. Whenever there exists and isomorphism from $G$ to $H$, we say that they are \emph{isomorphic} and write $G \cong H$.

\begin{definition}
An \emph{automorphism} of a graph $G$ is an isomorphism from $G$ to itself and these form a group under composition known as the \emph{automorphism group of $G$}, denoted $\aut(G)$.
\end{definition}

\subsection{Quantum automorphism groups}\label{sec:qautogroups}

Here we will introduce the necessary background on quantum groups we will need for our results. Our main focus is on quantum automorphism groups of graphs, but these fit into the more general framework of compact matrix quantum groups (CMQGs):

\begin{definition}\label{def:CMQG}
A \emph{compact matrix quantum group} $Q$ is specified by a pair $(A,\mathcal{U})$ where $A$ is a $C^*$-algebra and $\mathcal{U} = (u_{ij}) \in M_n(A)$ is a matrix whose entries generate $A$. Moreover, we require that there exists a unital $*$-homomorphism $\Delta\colon A \to A \otimes A$ satisfying $\Delta(u_{ij}) = \sum_k u_{ik} \otimes u_{kj}$, and that both $u$ and its transpose $\mathcal{U}^T$ are invertible. The $*$-homomorphism $\Delta$ is referred to as the \emph{comultiplication}. The matrix $\mathcal{U}$ is known as the \emph{fundamental representation} of $Q$.
\end{definition}

\begin{remark}\label{rem:qgroupnonsense}
The notion of a compact matrix quantum group is an abstract generalization of a compact matrix group. In the latter case, the algebra $A$ is the algebra of continuous $\mathbb{C}$-valued functions over the group $\Gamma$, denoted $C(\Gamma)$, and $u_{ij}\colon \Gamma \to \mathbb{C}$ is the coordinate function mapping a matrix $(g_{\ell k}) \in \Gamma$ to the entry $g_{ij}$. In analogy, even in the quantum case the algebra $A$ is \emph{referred to} as the algebra of continuous functions over the quantum group $Q$, and is denoted $C(Q)$. Note however, that $A$ is not an algebra of functions in general as any such algebra would be commutative, which we do not require of $A$. We further remark that some authors say that the quantum group $Q$ \emph{is} the pair $(A,\mathcal{U})$~\cite{qsymscalgs}, whereas others say that the quantum group $Q$ is not a concrete object that actually exists yet still refer to it as an actual object in analogy to the group case~\cite{qpermsurvey}.
\end{remark}

\begin{example}\label{ex:qortho}
In~\cite{WangOrtho}, Wang defined the orthogonal quantum group as the compact matrix quantum group $O_n^+$ with where $C(O_n^+)$ is the universal $C^*$-algebra generated by the entries of $\mathcal{U} = (u_{ij})$ which satisfy
\begin{equation}
u_{ij}^* = u_{ij}, \quad \text{\&} \quad \sum_{\ell = 1}^n u_{i\ell}u_{j\ell} = \delta_{ij} \vec{1} = \sum_{k=1}^n u_{ki}u_{kj} \text{ for all } i,j \in [n].
\end{equation}
The first condition above says that the entries of $\mathcal{U}$ are self-adjoint. Furthermore, given the first condition, the second condition can be written as $\mathcal{U}\mathcal{U}^T = I \otimes \vec{1} = \mathcal{U}^T\mathcal{U}$. Thus we see the relation to the usual orthogonal group.
\end{example}

\begin{example}\label{ex:qsym}
Another example of a compact matrix quantum group introduced by Wang~\cite{wang} that will be particularly important for us is the \emph{quantum symmetric group}, denoted $S_n^+$. Here, $C(S_n^+)$ is the universal $C^*$-algebra generated by the entries of $\mathcal{U} = (u_{ij})$ which satisfy the conditions
\begin{equation}\label{eq:qperm}
u_{ij}^2 = u_{ij}^* = u_{ij}, \quad \text{\&} \quad
\sum_{\ell = 1}^n u_{i\ell} = \vec{1} = \sum_{k = 1}^n u_{kj} \text{ for all } i,j \in [n] 
\end{equation}

A matrix $\mathcal{U} = (u_{ij})$ whose entries are from a unital $C^*$-algebra satisfying \Eq{qperm} is known as a \emph{magic unitary} or \emph{quantum permutation matrix}. The latter term is motivated by the fact that if the entries are from $\mathbb{C}$, then the conditions of \Eq{qperm} define permutation matrices. Note that the two conditions from \Eq{qperm} imply that $u_{i\ell}u_{ik} = \delta_{\ell k}u_{i\ell}$ and $u_{\ell j}u_{kj} = \delta_{\ell k}u_{kj}$, i.e., the elements in a row/column are mutually orthogonal. Furthermore, we remark that such a matrix is indeed unitary as
\[(\mathcal{U}\mathcal{U}^*)_{ij} = \sum_k u_{ik}u_{jk}^* = \sum_k u_{ik}u_{jk} = \delta_{ij}\vec{1},\]
and similarly for $\mathcal{U}^*\mathcal{U}$.
\end{example}

Shortly after Wang's introduction of the quantum symmetric group, Banica~\cite{banicahomogeneous} introduced the following definition of the quantum automorphism group of a graph\footnote{It is worth noting that Bichon~\cite{bichon2003quantum} had previously defined a related but different notion of quantum automorphism group of a graph. This is the same as Banica's but with the additional condition that $u_{i\ell}u_{jk} = u_{jk}u_{i\ell}$ for $i \sim j$. However, Banica's version is the one that is related to quantum isomorphism.}:

\begin{definition}\label{def:qaut}
Given a graph $G$ with vertex set $[n]$, its \emph{quantum automorphism group}, denoted $\qut(G)$, is the compact matrix quantum group given by $(C(\qut(G)), \mathcal{U})$ where $C(\qut(G))$ is the universal $C^*$-algebra with generators $u_{ij}, 1 \le i,j \le n$ satisfying the following relations for all $i,j \in [n]$:
\begin{align}
&u_{ij}^* = u_{ij}^2 = u_{ij} \label{eq:qaut1}\\
&\sum_{\ell = 1}^n u_{i\ell} = \vec{1} = \sum_{k = 1}^n u_{kj} \label{eq:qaut2}\\
&\sum_{k \sim j} u_{i k} = \sum_{\ell \sim i} u_{\ell j} \label{eq:qaut3}
\end{align}
\end{definition}

\begin{remark}\label{rem:qautdef}
Note that Relations \ref{eq:qaut1} and \ref{eq:qaut2} are simply the conditions requiring that $\mathcal{U}$ is a quantum permutation matrix, i.e., precisely the relations defining the quantum symmetric group. Further, Relation~\ref{eq:qaut3} can be written more compactly as $\mathcal{U}A_G = A_G\mathcal{U}$. 
Lastly, under the assumption that $\mathcal{U}$ is a quantum permutation matrix, Condition~\ref{eq:qaut3} is equivalent to the orthogonality conditions
\begin{equation}\label{eq:qaut3'}
u_{i\ell}u_{jk} = 0 \text{ if } \rel(i,j) \ne \rel(\ell,k),
\end{equation}
where $\rel$ is a function distinguishing the four possible cases for how two vertices $i$ and $j$ can be \emph{related}: $i = j$ with no loop, $i = j$ with a loop, $i \ne j$ \& $i \sim j$, and $i \ne j$ \& $i \not\sim j$.
\end{remark}

The above definition may seem far removed from the notion of the automorphism group of a graph. However, if we were to add the condition that $u_{ij}u_{\ell k} = u_{\ell k}u_{ij}$ for all $i,j,\ell,k$ (i.e., that the entries of $\mathcal{U}$ commute and thus generate a commutative $C^*$-algebra), then the resulting universal $C^*$-algebra would be isomorphic to the algebra of $\mathbb{C}$-valued functions on $\aut(G)$. Under such an isomorphism the operator $u_{ij}$ is mapped to the characteristic function of the automorphisms of $G$ that map $i$ to $j$. This informs our intuition in the quantum case: we think of $u_{ij}$ as an operator which somehow corresponds to mapping $i$ to $j$.

Sometimes, the commutativity of the entries of $\mathcal{U}$ in the definition of $\qut(G)$ is implied by the other conditions. This happens for instance for the complete graph $K_n$ for $n \le 3$, and for $G$ being any cycle graph except the 4-cycle~\cite{qpermsurvey}. In these cases we say that $G$ \emph{has no quantum symmetry} and that $\qut(G) = \aut(G)$.

Before moving on, let us point out that, like the classical case, the quantum automorphism group of a graph is not changed by taking the complement: $\qut(G) = \qut(\overline{G})$. This is easy to see if Condition~\ref{eq:qaut3'} is used in place of Condition~\ref{eq:qaut3}. However, as noted in \Rem{qautdef}, we can write the latter as $\mathcal{U}A_G = A_G\mathcal{U}$. It then follows that $\mathcal{U}A_{\overline{G}} = A_{\overline{G}}\mathcal{U}$ since $A_{\overline{G}} = J - I - A_G$ and any quantum permutation matrix commutes with both $I$ and $J$. Similarly $\qut(\overline{\overline{G}}) = \qut(G)$.

\paragraph{Quantum subgroups.}

In this work we will mostly deal with quantum automorphism groups of graphs. But we will see that the ideas we present can be applied to study more general quantum groups. However, we will always stay within the framework of compact matrix quantum groups, and more specifically quantum subgroups of the orthogonal quantum group $O_n^+$ introduced in \Ex{qortho}. Thus we should define the notion of a quantum subgroup of a compact matrix quantum group.

Given two quantum groups $Q$ and $Q'$ with fundamental representations $\mathcal{U} = (u_{ij})_{i,j = 1, \ldots, n}$ and $\mathcal{U}' = (u'_{ij})_{i,j = 1, \ldots, n}$ respectively, we say that $Q$ is a \emph{quantum subgroup} of $Q'$, denoted $Q \subseteq Q'$, if there exists a $*$-homomorphism $\psi \colon C(Q') \to C(Q)$ such that $\psi(u'_{ij}) = u_{ij}$. Note that this is a strict definition of quantum subgroup in that the fundamental representations must have the same dimension. Thus we do not have that $S_n^+ \subseteq S_{n+1}^+$ under this definition. Of course it is possible to define a compact quantum group $Q$ which is isomorphic to $S_{n}^+$ and satisfies $Q \subseteq S_{n+1}^+$. 

It is not difficult to see that for a graph $G$ with $n$ vertices, we have that $\qut(G) \subseteq S_n^+ \subseteq O_n^+$. The first inclusion follows from the fact that the relations imposed on the entries of the fundamental representation of $\qut(G)$ imply the relations imposed on $S_n^+$ (in fact the former are a superset of the latter). The second inclusion follows for the same reason, though here one needs a short argument (since one set of relations is not a subset of the other as they are written).

\paragraph{Orbits and intertwiners.}

In~\cite{qperms}, and independently in~\cite{banica2018modeling}, a notion of orbits of $\qut(G)$ on $V(G)$ was defined. Furthermore, in~\cite{qperms} they also defined orbits of $\qut(G)$ on $V(G) \times V(G)$, which are referred to as \emph{orbitals} to distinguish them from the orbits on $V(G)$. Since $\qut(G)$ is not actually a group consisting of maps acting on $V(G)$, these  notions cannot be defined as they usually are for $\aut(G)$. Instead, in~\cite{qperms} two relations are defined on $V(G)$ and $V(G) \times V(G)$ respectively. These relations are then proven to be equivalence relations and the orbits and orbitals of $\qut(G)$ are defined to be the equivalence classes of the relations. We repeat these definitions below:

\begin{definition}\label{def:qorbits}
Let $\mathcal{U} = (u_{ij})$ be the fundamental representation of $\qut(G)$. Define the relations $\sim_1$ and $\sim_2$ on $V(G)$ and $V(G) \times V(G)$ respectively as $i \sim_1 j$ if $u_{ij} \ne 0$ and $(i,\ell) \sim_2 (j,k)$ if $u_{ij}u_{\ell k} \ne 0$. Then both $\sim_1$ and $\sim_2$ are equivalence relations~\cite[Lemmas 3.2 \& 3.4]{qperms}, and the \emph{orbits} and \emph{orbitals} of $\qut(G)$ are defined to be the equivalence classes of these relations respectively.
\end{definition}

Note that in the classical case this corresponds with the usual definition of orbits of $\aut(G)$ on $V(G)$ and $V(G) \times V(G)$. Indeed, in the classical case (i.e., when the entries of $\mathcal{U}$ are required to commute) $u_{ij} \ne 0$ implies that there is some automorphism mapping $i$ to $j$ and thus they must be in the same orbit. Similarly, $u_{ij}u_{\ell k} \ne 0$ means that there is some automorphism that both maps $i$ to $j$ and $\ell$ to $k$, and thus $(i,\ell)$ and $(j,k)$ are in the same orbital.

We say that a vector $w \in \mathbb{C}^{V(G)}$ is \emph{constant on the orbits of} $\qut(G)$ if $w_{i} = w_j$ whenever $i$ and $j$ are in the same orbit. Similarly, a matrix $T \in \mathbb{C}^{V(G) \times V(G)}$ is \emph{constant on the orbitals of} $\qut(G)$ if $T_{i\ell} = T_{jk}$ whenever $(i,\ell)$ and $(j,k)$ are in the same orbital. Using these terms another characterization of the orbits and orbitals of $\qut(G)$ was given in~\cite{qperms}.

\begin{lemma}\label{lem:orbitsasintertwiners}
Let $G$ be a graph and let $\mathcal{U}$ be the fundamental representation of $\qut(G)$. Then $\mathcal{U}w = w (\vec{1})$ if and only if $w$ is constant on the orbits of $\qut(G)$, where $(\vec{1})$ is the $1 \times 1$ matrix with entry $\vec{1}$. Similarly, $\mathcal{U}T = T\mathcal{U}$ if and only if $T$ is constant on the orbitals of $\qut(G)$.
\end{lemma}

As usual, the above is analogous to the classical case, where $Pw = w$ for all $P \in \aut(G)$ (where we think of $\aut(G)$ as being represented by permutation matrices) if and only if $w$ is constant on the orbits of $\aut(G)$, and $PT = TP$ for all $P \in \aut(G)$ if and only if $T$ is constant on the orbitals of $\aut(G)$.

The above description of orbits and orbitals of $\qut(G)$ is in fact a special case of the more general notion of intertwiners which will be the focus of much of this paper. These are defined as follows:

\begin{definition}\label{def:intertwiners}
Consider a compact matrix quantum group $Q \subseteq O_n^+$ with fundamental representation $\mathcal{U} = (u_{ij})_{i,j = 1, \ldots, n}$. For $\ell \ge 0$, define $\mathcal{U}^{\otimes \ell}$ to be the $n^\ell \times n^\ell$ matrix whose $(i_1 \ldots i_\ell,j_1 \ldots j_\ell)$-entry is the product $u_{i_1j_1} \ldots u_{i_\ell j_\ell}$, thus $\mathcal{U}^{\otimes 0}$ is defined to be the matrix $(\vec{1})$. Then for $\ell,k \ge 0$, an $(\ell,k)$-intertwiner of $Q$ is a $n^\ell \times n^k$ $\mathbb{C}$-valued matrix $T$ such that $\mathcal{U}^{\otimes \ell}T = T\mathcal{U}^{\otimes k}$.
\end{definition}

For a given CMQG $Q \subseteq O_n^+$, we will use $C_Q(\ell,k)$ to denote its $(\ell,k)$-intertwiners, and let $C_Q := \bigcup_{\ell,k = 0}^\infty C_Q(\ell,k)$. It is well known~\cite{noncross} that the intertwiners of a compact matrix quantum group $Q \subseteq O_n^+$ form a \emph{tensor category with duals}, meaning that they have the following properties:
\begin{enumerate}
\item if $T,T' \in C_Q(\ell,k)$ and $\alpha,\beta \in \mathbb{C}$, then $\alpha T + \beta T' \in C_Q(\ell,k)$, i.e., $C_Q(\ell,k)$ is a vector space;
\item if $T \in C_Q(\ell,k)$ and $T' \in C_Q(r,s)$, then $T \otimes T' \in C_q(\ell + r,k+s)$;
\item if $T \in C_Q(\ell,k)$ and $T' \in C_Q(k,r)$, then $TT' \in C_Q(\ell,r)$;
\item if $T \in C_Q(\ell,k)$, then $T^* \in C_Q(k,\ell)$;
\item $I \in C_Q(1,1)$, where $I$ is the $n \times n$ identity matrix;
\item $\xi = \sum_{i=1}^n e_i \otimes e_i \in C_Q(2,0)$, where $e_i \in \mathbb{C}^n$ is the $i^\text{th}$ standard basis vector.
\end{enumerate}

Most of these are straightforward. For the last, note that it is only required to show that $\xi \in C_{O_n^+}(2,0)$ since if $Q \subseteq O_n^+$, then any relations satisfied by the entries of the fundamental representation of $O_n^+$ must also be satisfied by those of $Q$. For $O_n^+$, we have that
\[(\mathcal{U}^{\otimes 2} \xi)_{i_1i_2} = \sum_{j_1,j_2 = 1}^n u_{i_1 j_1}u_{i_2 j_2}\xi_{j_1 j_2} = \sum_{j=1}^n u_{i_1j}u_{i_2j} = \delta_{i_1 i_2}\vec{1} = (\xi\mathcal{U}^{\otimes 0})_{i_1 i_2},\]
by definition.

The importance of intertwiners is reflected in a Tannaka-Krein type result by Woronowicz~\cite{woronowicz1988}, which says that there is a one-to-one correspondence between tensor categories and compact matrix quantum groups. This is given below where we use $C_n(\ell,k)$ to denote all $n^\ell \times n^k$ $\mathbb{C}$-valued matrices and $C_n = \bigcup_{\ell,k = 0}^\infty C_n(\ell,k)$.

\begin{theorem}\label{thm:tannakakrein}
The construction $Q \mapsto C_Q$ induces a one-to-one correspondence between compact matrix quantum groups $Q \subseteq O_n^+$ and tensor categories with duals $C_Q \subseteq C_n$.
\end{theorem}

The above theorem means that finding new CMQGs $Q \subseteq O_n^+$ is equivalent to finding new tensor categories with duals.

\begin{remark}\label{rem:orbitsasintertwiners}
In the language of intertwiners, \Lem{orbitsasintertwiners} says that $i,j \in V(G)$ are in the same orbit of $\qut(G)$ if and only if $T_i = T_j$ for all $T \in C_{\qut(G)}(1,0)$. Similarly, $(i,\ell)$ and $(j,k)$ are in the same orbital of $\qut(G)$ if and only if $T_{i,\ell} = T_{j,k}$ for all $T \in C_{\qut(G)}(1,1)$.
\end{remark}

\paragraph{Intertwiners of the quantum automorphism group of a graph.}

The first main result of this paper is a combinatorial characterization of the intertwiners of $\qut(G)$ for an arbitrary graph $G$. To establish this characterization, we will make use of a result of Chassaniol which states that $C_{\qut(G)}$ is generated by three intertwiners using the operations of matrix product, tensor product, conjugate transposition, and linear combinations. Two of these three intertwiners are known as the \emph{multiplication} and \emph{unit} maps, denoted $M \in C_n(1,2)$ and $U \in C_n(1,0)$ respectively. These are defined as the linear maps satisfying the following:
\begin{equation}\label{eq:multunit}
M(e_i \otimes e_j) = \delta_{ij} e_i \quad \text{\&} \quad U(1) = \sum_{i = 1}^n e_i.
\end{equation}
As matrices $U$ is simply the all ones column vector and $M$ is an $n \times n^2$ matrix whose $i,j_1j_2$-entry is 1 if $i = j_1 = j_2$ and is 0 otherwise. Note that strictly speaking we should indicate the dimension by writing $M_n$ or $U_n$, but it will almost always be clear what $n$ is from context (usually the number of vertices of a graph), and so we will routinely omit this.

\begin{remark}\label{rem:Snintertwiners}
It is well known that the maps $M$ and $U$ are intertwiners for $S_n^+$. Indeed, if $\mathcal{U}$ is any quantum permutation matrix, then it is a straightforward computation to show that $\mathcal{U}M = M\mathcal{U}^{\otimes 2}$ and $\mathcal{U}U = U\mathcal{U}^{\otimes 0}$. It is further known~\cite{banicaintegration} that $M$ and $U$ actually generate all of the intertwiners of $S_n^+$, i.e., that $C_{S_n^+} = \langle U,M\rangle_{+,\circ, \otimes, *}$ where the righthand side denotes the set of matrices that can be obtained from $M$, $U$, and the operations of matrix product, tensor product, conjugate transposition, and linear combinations (we use `$\circ$' to denote matrix product here even though we usually write it as juxtaposition, and we use `$+$' to denote linear combinations). Equivalently, $\langle U,M\rangle_{+,\circ, \otimes, *}$ is the intersection of all sets of matrices closed under these operations that contain $M$ and $U$.
\end{remark}

It will be useful for us to consider, and name, an infinite family of maps/matrices that can be seen as a generalization of the multiplication and unit maps. We denote by $M^{\ell,k} \in C_n(\ell,k)$ the $(\ell,k)$-\emph{generalized multiplication map/matrix} which for $\ell + k > 0$ is defined entrywise as
\begin{align}\label{eq:genmultmap}
M^{\ell,k}_{i_1\ldots i_\ell,j_1\ldots j_k} &= \begin{cases} 1 & \text{if } i_1 = \ldots = i_\ell = j_1 = \ldots = j_k \\ 0 & \text{o.w.}
\end{cases}
\end{align}
We also define $M^{0,0}$ to be the $1 \times 1$ matrix whose only entry is $n$. For $\ell,k > 0$, the map $M^{\ell,k}$ takes $e_{j_1} \otimes \ldots \otimes e_{j_k}$ to $e_{i}^{\otimes \ell}$ (the $\ell$-fold tensor product of $e_i$) if $j_1 = \ldots = j_k = i$, and maps it to the zero vector otherwise. On the other hand for $\ell > 0$ the map $M^{\ell,0}$ takes $(1)$ to $\sum_{i} e_i^{\otimes \ell}$. Note that $M^{k,\ell}$ is the (conjugate) transpose/adjoint of $M^{\ell,k}$. It is known, and we will see later, that $M^{\ell,k}$ is an intertwiner for $S_n^+$ for all $\ell,k \ge 0$, i.e., we can construct any $M^{\ell,k}$ from $M$ and $U$.

\begin{remark}\label{rem:Mnotation}
It is straightforward to see that $M$ and $U$ are equal to $M^{1,2}$ and $M^{1,0}$ respectively. From here on we will mainly use the latter notation. We also remark that $M^{1,1} = I$ and $M^{2,0} = \xi$.
\end{remark}

Another important map that we need is the \emph{swap map} $S \in C_n(2,2)$. This is defined as $S(e_i \otimes e_j) = e_j \otimes e_i$. Given a matrix $\mathcal{U}$, it is easy to see that its entries commute if and only if $\mathcal{U}S = S\mathcal{U}$. Thus a CMQG $Q \subseteq O_n^+$ is a classical group (i.e., $C(Q)$ is commutative) if and only if $S \in C_Q(2,2)$.

Recall that the relations imposed on the entries of the fundamental representation of $S_n^+$ are also imposed for $\qut(G)$ for $n = |V(G)|$. Additionally, for $\qut(G)$, Condition~\ref{eq:qaut3} is imposed which can be written as $\mathcal{U}A_G = A_G\mathcal{U}$ (recall \Rem{qautdef}). In other words, $\qut(G)$ is obtained from $S_n^+$ by adding the extra condition that $A_G$ is a $(1,1)$-intertwiner. This is formally captured in the result of Chassaniol below. Here and henceforth we will use $C^G$ and $C_q^G$ to denote $C_{\aut(G)}$ and $C_{\qut(G)}$ respectively (and we similarly use $C^G(\ell,k)$ and $C_q^G(\ell,k)$).

\begin{theorem}[Chassaniol~\cite{chassaniol2019study}]\label{thm:chassaniol}
Let $G$ be a graph with adjacency matrix $A_G$. Then we have that $C_q^G = \langle M^{1,0},M^{1,2},A_G\rangle_{+,\circ, \otimes, *}$, and $C^G = \langle M^{1,0},M^{1,2},A_G, S\rangle_{+,\circ, \otimes, *}$.
\end{theorem}

\begin{remark}\label{rem:chassaniol}
Note that if $T \in \langle T_1, \ldots, T_r \rangle_{+,\circ, \otimes, *}$, then any use of linear combinations in the expression for $T$ can be ``moved to the end". In other words,
\[\langle T_1, \ldots, T_r\rangle_{+,\circ, \otimes, *} = \bigcup_{\ell,k = 0}^\infty \spn\{T \in C_n(\ell,k) : T \in \langle T_1, \ldots, T_r\rangle_{\circ, \otimes, *}\}.\]
We will actually abuse notation somewhat and write the righthand side above as simply $\spn\{T : T \in \langle T_1, \ldots, T_r\rangle_{\circ, \otimes, *}\}$. Thus we have that $C_q^G = \spn\{T : T \in \langle M^{1,0}, M^{1,2},A_G \rangle_{\circ, \otimes, *}\}$.
\end{remark}

\paragraph{Easy quantum groups and partition categories.}

We saw earlier that the intertwiners of $S_n^+$ are generated by the multiplication and unit maps $M^{1,2}$ and $M^{1,0}$. However, there is a much more complete description of $C_{S_n^+}$ given by Banica and Speicher in~\cite{noncross}. This description is based on the combinatorics of partitions. The idea is to associate to any partition a linear map in such a way that mutliplication, tensor product, and (conjugate) transposition of these maps correspond to natural operations on the underlying partitions.

Fix $n \ge 1$ and $\ell,k \ge 0$ and consider a partition $\mathbb{P} = \{P_1, \ldots, P_r\}$ of the set $\{1_L, \ldots, \ell_L, 1_U, \ldots, k_U\}$, where we allow $\mathbb{P}$ to have empty parts. We think of $1_L, \ldots, \ell_L$ as the \emph{lower points} of the set and of $1_U, \ldots, k_U$ of the \emph{upper points}. We can then draw such a partition graphically by joining points that are in the same part, as shown in \Fig{drawpartitions}. Given two tuples $\vec{i} = (i_1, \ldots, i_\ell) \in [n]^\ell$ and $\vec{j} = (j_1, \ldots, j_k) \in [n]^k$, we define $\delta_{\mathbb{P}}(\vec{i},\vec{j})$ to be 1 if when putting the indices of $\vec{i}$ and $\vec{j}$ on the points of $\mathbb{P}$ in the obvious way, the partition $\mathbb{P}$ only joins indices corresponding to equal entries. Otherwise $\delta_{\mathbb{P}}(\vec{i},\vec{j})$ is defined to be 0. More formally, $\delta_{\mathbb{P}}(\vec{i},\vec{j}) = 0$ if there exists $a,b,t$ such that $i_a \ne i_b$ and $a_L,b_L \in P_t$, or $j_a \ne j_b$ and $a_U,b_U \in P_t$, or $i_a \ne j_b$ and $a_L,b_U \in P_t$. The matrix $T_{\mathbb{P}}$ is then defined entrywise as $(T_{\mathbb{P}})_{\vec{i},\vec{j}} = n^{e(\mathbb{P})}\delta_{\mathbb{P}}(\vec{i},\vec{j})$, where $e(\mathbb{P})$ is the number of empty parts of $\mathbb{P}$\footnote{It is more standard to not allow empty parts in $\mathbb{P}$ and thus to not have any factor of $n^{e(\mathbb{P})}$ here. However, for us it is more natural to allow empty parts and so we must include this factor.}. Note that we are leaving off the $n$ from the notation $T_{\mathbb{P}}$ since it will usually be implicit.

\begin{figure}[h!]
\begin{subfigure}{.5\textwidth}
  \centering
  \[\mathbb{P}_1=
  \BigPartition{
  \Pblock 0 to 0.25:2,3
  \Pblock 1 to 0.75:1,2,3
  \Psingletons 0 to 0.25:1,4
  \Pline (2.5,0.25) (2.5,0.75)
  }\]
  \caption{\phantom{a} $\mathbb{P}_1 = \{\{1_U, 2_U, 3_U,2_L,3_L\},\{1_L\},\{4_L\}\}$}
  \label{fig:part1}
\end{subfigure}
\begin{subfigure}{.5\textwidth}
  \centering
  \[\mathbb{P}_2=
  \BigPartition{
  \Psingletons 0 to 0.25:3
  \Psingletons 1 to 0.75:3
  \Pline (2,0) (1,1)
  \Pline (3,0) (3,1)
  \Pline (1,0) (2,1)
  }\]
  \caption{\phantom{a} $\mathbb{P}_2 = \{\{1_U, 2_L\}, \{1_L,2_U\},\{3_L,3_U\}\}$}
  \label{fig:part2}
\end{subfigure}
\caption{How to draw partitions}\label{fig:drawpartitions}
\end{figure}

We denote by $\mathbb{P}(\ell,k)$ the set of partitions of $\{1_L, \ldots, \ell_L, 1_U, \ldots, k_U\}$. Given $\mathbb{P} \in \mathbb{P}(\ell,k)$ and $\mathbb{P}' \in \mathbb{P}(\ell',k')$ the tensor product $\mathbb{P} \otimes \mathbb{P}'$ is equal to the partition in $\mathbb{P}(\ell +\ell',k + k')$ obtained by drawing the partitions $\mathbb{P}$ and $\mathbb{P}'$ horizontally next to each other as in \Fig{partops1}. If $k = \ell'$ then the composition $\mathbb{P} \circ \mathbb{P}' \in \mathbb{P}(\ell,k')$ is obtained by drawing $\mathbb{P}$ below $\mathbb{P}'$ and joining the points $1_U, \ldots, k_U$ with the points $1_L, \ldots, \ell'_L$ (see \Fig{partops2}). Note that this may create some closed blocks, each of which is taken to be an empty part of the resulting partition\footnote{This also differs from the standard of removing these closed blocks. This change is also reflected in the formula for composition in \Eq{partitioncorr} which contains no $n^{e(\mathbb{P} \circ \mathbb{P}')}$ factor.}. Finally, the transposition $\mathbb{P}^* \in \mathbb{P}(k,\ell)$ is obtained by reflecting $\mathbb{P}$ over the horizontal axis, as shown in \Fig{partops3}. In~\cite{noncross} Banica and Speicher showed that these operations on partitions coincide with the analogous operations on matrices, i.e., that
\begin{equation}\label{eq:partitioncorr}
T_{\mathbb{P}} \otimes T_{\mathbb{P}'} = T_{\mathbb{P} \otimes \mathbb{P}'}, \quad T_{\mathbb{P}}T_{\mathbb{P}'} = T_{\mathbb{P} \circ \mathbb{P}'}, \quad T_{\mathbb{P}}^* = T_{\mathbb{P}^*}.
\end{equation}

\begin{figure}[h!]
\begin{subfigure}{.345\textwidth}
  \centering
  \[
  \BigPartition{
  \Pblock 0 to 0.25:1,2
  \Psingletons 1 to 0.75:1
  \Pline (1,0) (1,1)
  } \otimes 
  \BigPartition{
  \Pblock 1 to 0.75:1,2
  \Psingletons 0 to 0.25:1
  \Pline (1,0) (1,1)
  } = 
  \BigPartition{
  \Pblock 0 to 0.25:1,2
  \Psingletons 1 to 0.75:1
  \Pblock 1 to 0.75:2,3
  \Psingletons 0 to 0.25:3
  \Pline (1,0) (1,1)
  \Pline (3,0) (3,1)
  }
  \]
  \vspace{-.1in}
    \caption[]{
  \begin{minipage}{\linewidth}
     \begin{align*}
        &\{\{1_U,1_L,2_L\}\} \otimes \{\{1_U,2_U,1_L\}\} \\
        &= \{\{1_U,1_L,2_L\},\{2_U,3_U,3_L\}\}
     \end{align*}
  \end{minipage}
}
  \label{fig:partops1}
\end{subfigure}
\hspace{.1\textwidth}
\begin{subfigure}{.52\textwidth}
  \centering
  \[
  \BigPartition{
  \Pblock 1 to 0.75:2,3
  \Psingletons 0 to 0.25:1
  \Psingletons 1 to 0.25:1
  \Pline (2,0) (2,0.75)
  } \circ 
  \BigPartition{
  \Pblock 0 to 0.25:1,2
  \Pblock 1 to 0.75:1,2
  \Psingletons 0 to 0.25:3
  \Pline (1.5,0.25) (1.5,0.75) }= 
  \BigPartition{
  \Pblock 0 to 0.25:1,2
  \Pblock 1 to 0.75:1,2
  \Pline (1.5,0.25) (1.5,0.75)
  }
  \]
  \vspace{-.1in}
  \caption[]{
  \begin{minipage}{\linewidth}
     \begin{align*}
        &\{\{1_U,1_L\},\{2_U, 3_U,2_L\}\}  \circ  \{\{1_U,2_U,1_L,2_L\},\{3_L\}\} \\
        &= \{\{1_U,2_U,1_L,2_L\}\}
     \end{align*}
  \end{minipage}
}
  \label{fig:partops2}
\end{subfigure}
\begin{center}
\begin{subfigure}{.8\textwidth}
  \centering
  \[
  \left(\BigPartition{
  \Pblock 0 to 0.25:1,2
  \Pblock 1 to .75:2,3
  \Psingletons 1 to 0.75:1
  \Pline (1,0) (1,1)
  }\right)^*  = 
  \BigPartition{
  \Pblock 1 to 0.75:1,2
  \Pblock 0 to .25:2,3
  \Psingletons 0 to 0.25:1
  \Pline (1,0) (1,1)
  }
  \]
  \caption{\phantom{a} $\{\{1_U,1_L,2_L\},\{2_U,3_U\}\}^* = \{\{1_L,1_U,2_U\},\{2_L,3_L\}\}$}
  \label{fig:partops3}
\end{subfigure}
\end{center}
\caption{Operations on partitions}\label{fig:partitionops}
\end{figure}

Banica and Speicher showed that the maps $T_\mathbb{P}$ for $\mathbb{P} \in \mathbb{P}(\ell,k)$ span the intertwiner space $C_{S_n}(\ell,k)$. In the quantum case, i.e., for $S_n^+$, they considered \emph{non-crossing partitions}:

\begin{definition}\label{def:noncrossing}
A partition $\mathbb{P} = \{P_1, \ldots, P_r\}$ of a totally ordered set $V$ is said to be \emph{non-crossing} if whenever $a < b < c < d$, and $a,c$ are in the same part and $b,d$ are in the same part, then the two parts coincide.
\end{definition}

For example, the partition $\{\{1,2\},\{3,4\}\}$ is non-crossing, but $\{\{1,3\},\{2,4\}\}$ is not. Note that empty parts do not have any effect on whether a partition is non-crossing. For the set $\{1_L, \ldots, \ell_L, 1_U, \ldots, k_U\}$ used in the definition of $\mathbb{P}(\ell,k)$, we impose the total order $1_L < 2_L < \ldots < \ell_L < k_U < \ldots < 1_U$ in order to define when an element of $\mathbb{P}(\ell,k)$ is non-crossing. As one might expect, a partition $\mathbb{P} \in \mathbb{P}(\ell,k)$ is non-crossing if it can be drawn in the manner of \Fig{drawpartitions} without any crossings. For example, the partition in \Fig{part1} is non-crossing, but the partition in \Fig{part2} is not. We use $NC(\ell,k)$ to denote the non-crossing partitions of $\mathbb{P}(\ell,k)$. Banica and Speicher showed that $C_{S_n^+}(\ell,k)$ is spanned by the maps $T_{\mathbb{P}}$ for $\mathbb{P} \in NC(\ell,k)$ (note that the maps $M^{\ell,k}$ defined above are precisely the maps $T_\mathbb{P} \in C_n(\ell,k)$ where $\mathbb{P}$ has a single part). Similarly, for $O_n$ and $O_n^+$, they showed that the intertwiner spaces are spanned by the pairings (partitions with every part having size two) and non-crossing pairings. More generally, Banica and Speicher define a \emph{partition category} as any set of partitions closed under the operations of composition, tensor product, transposition, and containing the partitions $\{\{1_L,1_U\}\}$ and $\{\{1_L,2_L\}\}$ (whose associated maps are the identity and $\xi$). Given any partition category $\mathcal{PC}$ and integer $n$, the span of the maps $T_\mathbb{P} \in C_n$ for $\mathbb{P} \in \mathcal{PC}$ is a tensor category with duals, and thus is the intertwiner space of some compact matrix quantum group. The CMQGs arising in this way are known as \emph{easy quantum groups} (or sometimes \emph{partition quantum groups}), and these have been extensively studied in the literature, with their full classification being achieved only recently~\cite{Raum2016}.

The spirit of the first part of this paper is similar to that of Banica and Speicher's work. We show that the intertwiners spaces of $\qut(G)$ are spanned by maps associated to bi-labeled graphs (see \Def{LGraph}), and we introduce the notion of \emph{graph categories} (see \Sec{graphqgroups}) which are classes of bi-labeled graphs closed under operations that correspond to product, tensor product, and conjugate transposition. For any fixed graph $G$, any graph category corresponds to a tensor category with duals and thus to a compact matrix quantum group. Like the easy quantum groups, these \emph{graph-theoretic quantum groups} have an underlying combinatorial structure that can be exploited for their study. In fact, we show that partition categories are precisely the graph categories consisting of bi-labeled graphs having no edges. Thus graph-theoretic quantum groups promise to be a much richer class than easy quantum groups, analogously to the difference between the class of graphs and the class of sets. Furthermore, a single parititon category gives rise to a quantum group for every positive integer $n$, whereas a single graph category gives rise to a quantum group for every graph $G$.



%


\subsection{Quantum isomorphism}

In~\cite{qiso1}, a nonlocal game was introduced which captures the notion of graph isomorphism. In turn, by allowing entangled strategies, this allows one to define a type of quantum isomorphism in a natural way. We will give a brief description of this game and its classical/quantum strategies, but for a more thorough explanation we refer the reader to~\cite{qiso1}.

Given graphs $G$ and $H$, the $(G,H)$-isomorphism game is played as follows: a referee/verifier sends each of two players (Alice and Bob) a vertex of $G$ or $H$ (not necessarily the same vertex to both). 
Each of Alice and Bob must respond to the referee with a vertex of $G$ or $H$. Alice and Bob win if they meet two conditions, the first of which is:

\begin{enumerate}
\item If a player (Alice or Bob) receives a vertex from $G$, they must respond with a vertex from $H$ and vice versa\footnote{It is implicitly assumed that the vertex sets of $G$ and $H$ are disjoint, and thus the players know which graph the vertex they receive is from. Alternatively, we can simply require that the referee tells them which graph the vertex is from.}.
\end{enumerate}

Assuming this condition is met, Alice either receives or responds with a vertex of $G$, which we will call $g_A$, and either responds with or receives a vertex of $H$, which we will call $h_A$. We can similarly define $g_B$ and $h_B$ for Bob. The second condition they must meet in order to win is then given by

\begin{enumerate}
\item[{2.}] $\rel(g_A, g_B) = \rel(h_A, h_B)$,
\end{enumerate}

\begin{remark}
In~\cite{qiso1}, graphs with loops were not considered, and thus the function $\rel$ took only three values, indicating whether the vertices were equal, adjacent, or distinct non-adjacent. When loops are allowed in the graphs $G$ and $H$, we must refine the function $\rel$ in order to distinguish the vertices with loops from those without, thus obtaining the definition presented in \Sec{qautogroups}.  This more general setting does not fundamentally change any of the previous results on quantum isomorphisms, and it is completely straightforward how to adapt these results to the case where loops are allowed.
\end{remark}

The players know the graphs $G$ and $H$ beforehand and can agree on any strategy they like, but they are not allowed to communicate during the game. For simplicity we may assume that the referee sends the vertices to the players uniformly at random. We only require that Alice and Bob play one round of the game (each receive and respond with a single vertex), but we require that their strategy guarantees that they win with probability 1. We say that such a strategy is a \emph{perfect} or \emph{winning} strategy.

If $\varphi: V(G) \to V(H)$ is an isomorphism, then it is not difficult to see that responding with $\varphi(g)$ for $g \in V(G)$ and $\varphi^{-1}(h)$ for $h \in V(H)$ is a perfect strategy for the $(G,H)$-isomorphism game. Conversely, any perfect deterministic classical strategy can be shown to have this form. This is what was shown in~\cite{qiso1} and it follows that there exists a perfect classical strategy for the $(G,H)$-isomorphism game if and only if $G \cong H$. In general, classical players could use shared randomness but it is not hard to see that this would not allow them to win if they were not already able to succeed perfectly using a deterministic strategy. This motivates the definition of \emph{quantum isomorphic} graphs $G$ and $H$: those for which the $(G,H)$-isomorphism game can be won with a ``quantum strategy"

In a quantum strategy, Alice and Bob have access to a shared entangled state which they are allowed to perform local quantum measurements on. This does not allow them to communicate, but may allow them to correlate their actions/responses in ways not possible for classical players. In~\cite{qiso1}, two different models for performing joint measurements on a shared state were considered: the tensor product framework and the commuting operator framework. In this work we will only consider the commuting operator framework and thus we will not go into detail about the tensor product framework. In the commuting operator framework, Alice and Bob share a Hilbert space in which their shared state lives. They are each allowed to perform measurements on the shared state but all of Alice's measurement operators must commute with all of Bob's. The Hilbert space, and thus the measurement operators, are in general allowed to be infinite-dimensional, and it is known that there are graphs $G$ and $H$ such that the $(G,H)$-isomorphism game can be won by infinite-dimensional strategies but not finite dimensional ones.

The precise mathematical model of a quantum strategy for the $(G,H)$-isomorphism game is as follows: the players share a quantum system modelled by a Hilbert space $\mathcal{H}$, which is in a quantum state modelled by a unit vector $\psi \in \mathcal{H}$. Upon receiving input $x \in V(G) \cup V(H)$ from the referee, Alice performs a quantum measurement, modeled as a positive-operator valued measurement (POVM) $\mathcal{E}_x = \{E_{xy} \in B(\mathcal{H}): y \in V(G) \cup V(H)\}$, whose outcomes are indexed by her possible responses. Here $B(\mathcal{H})$ denotes the bounded linear operators on $\mathcal{H}$, and a POVM is a collection of positive operators\footnote{An operator $M \in B(\mathcal{H})$ is \emph{positive} if $\phi^*M\phi \ge 0$ for all $\phi \in \mathcal{H}$.} whose sum is the identity. Similarly, Bob has POVMs $\mathcal{F}_x = \{F_{xy} \in B(\mathcal{H}) : y \in V(G) \cup V(H)\}$ for each $x \in V(G) \cup V(H)$. Lastly, it is required that $E_{xy}F_{x'y'} = F_{x'y'}E_{xy}$ for all $x,x',y,y' \in V(G) \cup V(H)$. Any such quantum strategy results in a \emph{correlation}: a joint conditional probability distribution indicating the probability of the players responding with a given pair of vertices conditioned on them having received a given pair. These probabilities are given by the following formula:
\[p(y,y'|x,x') = \psi^*E_{xy}F_{x'y'}\psi,\]
where $p(y,y'|x,x')$ is the probability of Alice and Bob responding with vertices $y,y'$ assuming they received $x,x'$ respectively. Recall that a perfect/winning strategy is one which allows the players to win with probability 1, or equivalently, causes them to lose with probability 0. The latter is equivalent to the requirement that $\psi^*E_{xy}F_{x'y'}\psi = 0$ whenever responding with $y,y'$ upon receiving $x,x'$ would result in the players not meeting Conditions (1) and (2) above. As defined in~\cite{qiso1}, we say that graphs $G$ and $H$ are \emph{quantum isomorphic}, and write $G \cong_{qc} H$, if there exists a perfect quantum strategy for the $(G,H)$-isomorphism game. Note that in~\cite{qiso1} and in most other sources ``quantum isomorphism" refers the existence of a perfect quantum strategy for the isomorphism game \emph{in the tensor product framework}, and ``quantum commuting isomorphism" is used to for the commuting operator framework which we consider here. As we will only discuss the latter henceforth, we will simply refer to is as quantum isomorphism. However, we keep the notation $G \cong_{qc} H$ which is the standard for this notion.

\begin{remark}\label{rem:complements}
It is well known and obvious that two graphs are isomorphic if and only if their complements are isomorphic (and similarly for their full complements). The same is true for quantum isomorphism. Indeed, the Conditions (1) and (2) required of the players are not affected by taking complements (or full complements). Thus any strategy, quantum or otherwise, that wins the $(G,H)$-isomorphism game also wins the $(\overline{G},\overline{H})$-isomorphism game and $(\overline{\overline{G}},\overline{\overline{H}})$-isomorphism game.
\end{remark}

One of the difficulties in analyzing quantum strategies is that they have many components: one must choose a shared state and measurement operators for both Alice and Bob. In the case of the isomorphism game, some of this difficulty has been alleviated by the following result proven in~\cite{qperms}:

\begin{theorem}\label{thm:qisoMU}
Let $G$ and $H$ be graphs with adjacency matrices $A_G$ and $A_H$ respectively. Then $G \cong_{qc} H$ if and only if there exists a quantum permutation matrix $\mathcal{U} = (u_{gh})_{g \in V(G), h \in V(H)}$ such that $A_G\mathcal{U} = \mathcal{U}A_H$.
\end{theorem}

The entries of the quantum permutation matrix $\mathcal{U}$ above correspond to the measurement operators of Alice (or Bob) used in winning strategy for the $(G,H)$-isomorphism game. The above theorem is an analog of the fact that graphs $G$ and $H$ are isomorphic if and only if there exists a permutation matrix $P$ such that $A_GP = PA_H$, though this is usually written as $P^TA_GP = A_H$. \Thm{qisoMU} illustrates that there is a close connection between quantum isomorphisms and quantum automorphism groups of graphs. The following theorem of~\cite{qperms} make this connection even more concrete:

\begin{theorem}\label{thm:qisoorbits}
Let $G$ and $H$ be connected graphs. Then $G \cong_{qc} H$ if and only if there exists $g \in V(G)$ and $h \in V(H)$ that are in the same orbit of $\qut(G \cup H)$.
\end{theorem}

We remark that the connectedness condition above is not really a restriction as either a graph or its complement must be connected, and a connected graph cannot be quantum isomorphic to a disconnected one.

\Thm{qisoorbits} will allow us to use our characterization of the intertwiners of $\qut(G)$ (\Thm{onecat}) in order to prove that two graphs are quantum isomorphic if and only if they admit the same number of homomorphisms from any planar graph (\Thm{main}).

\section{Bi-labeled graphs and homomorphism matrices}\label{sec:bilabeledgraphs}


Here we introduce the two central notions of our combinatorial characterization of the intertwiner space of $\qut(G)$: bi-labeled graphs and their corresponding homomorphism matrices. The name ``bi-labeled" graph comes from the work of Lov\'{a}sz on graph limits~\cite{lovasz2012large}. Our notion is equivalent, but we formulate it slightly differently:

\begin{definition}[Bi-labeled graphs]
A \emph{bi-labeled graph} $\vec{K}$ is a triple $(K,\vec{a},\vec{b})$ where $K$ is a graph and $\vec{a} = (a_1, \ldots, a_\ell) \in V(K)^\ell$, $\vec{b} = (b_1, \ldots, b_k) \in V(K)^k$ are tuples/vectors of vertices of $K$ with $\ell,k \ge 0$. We refer to $\vec{a}$ as the \emph{output tuple/vector}, and the $a_i$'s as the \emph{output vertices}. Similarly, $\vec{b}$ is the \emph{input tuple/vector} and the $b_j$'s are the \emph{input vertices}. We refer to $K$ as the \emph{underlying graph} of the bi-labeled graph $\K$. Finally, we use $\G(\ell,k)$ to denote the set of all bi-labeled graphs with $\ell$ output vertices and $k$ input vertices, and use $\G = \cup_{\ell=0,k=0}^\infty \G(\ell,k)$.
\label{def:LGraph}
\end{definition}

\begin{remark}\label{rem:isoclasses}
Strictly speaking we should let $\G(\ell,k)$ be the set of all \emph{isomorphism classes} of bi-labeled graphs with $\ell$ outputs and $k$ inputs since otherwise $\G(\ell,k)$ would not even be a set. Here an isomorphism of bi-labeled graphs $\vec{K} = (K,\vec{a},\vec{b}) \in \G(\ell,k)$ and $\vec{K'} = (K',\vec{a'},\vec{b'}) \in \G(\ell,k)$ is an isomorphism $\varphi$ from $K$ to $K'$ such that $\varphi(\vec{a}) = \vec{a'}$ and $\varphi(\vec{b}) = \vec{b'}$, where $\varphi(\vec{a}) = \vec{a'}$ denotes $\varphi(a_i) = a'_i$ for all $i = 1, \ldots, \ell$, and similarly for $\varphi(\vec{b}) = \vec{b'}$. However we will be informal and write $\vec{K} \in \G(\ell,k)$ to denote that the isomorphism class of $\vec{K}$ belongs to $\G(\ell,k)$.
\end{remark}

Note that we will use $\varnothing$ to denote the \emph{empty tuple/vector}, i.e, the tuple/vector of length zero. In~\cite{lovasz2012large}, Lov\'{a}sz described bi-labeled graphs as graphs in which the left labels $1, \ldots, \ell$ are assigned to some vertices, and the right labels $1, \ldots, k$ are assigned to some vertices. The left and right label $i$ are distinguished, and vertices can be assigned more than one left and/or right label. Intuitively, one thinks of the graph being drawn with the left labeled vertices on the left and the right labeled vertices on the right, though of course there are problems when a vertex has both a left and right label. The correspondence to our notion is straightforward: $a_i$ is equal to the vertex with left label $i$ and $b_j$ is equal to the vertex with right label $j$. We will always refer to bi-labeled graphs by a capital boldface letter, and unless otherwise specified its underlying graph will be denoted by the same capital letter but without boldface.

\begin{remark}[Drawing a bi-labeled graph]\label{rem:drawing}
It is well known that graphs are often thought of and represented diagrammatically. The vertices are represented as points and the edges as curves between the appropriate pair of points. A bi-labeled graph is a graph plus some additional information, and thus to represent them we will need to add something to the usual picture of a graph. This is done by the addition of input/output ``wires" that are attached to the input and output vertices of the bi-labeled graph. Specifically, to draw a bi-labeled $\K = (K,\vec{a},\vec{b})$ we draw the underlying graph $K$, and we attach the $i^\text{th}$ output wire to $a_i$ and the $j^\text{th}$ input wire to $b_j$. Note that this means that vertices can have multiple input/output wires attached to them. The input and output wires extend to the far right and far left of the picture respectively. Finally, in order to indicate which input/output wire is which, we draw them so that they occur in numerical order (first at the top) at the edges of the picture. We have given examples of how to draw bi-labeled graphs in \Fig{draw}. In \Fig{drawabstract} we illustrate an example whose underlying graph is ``generic" and thus represented by a gray blob. In \Fig{drawconcrete} we give a concrete example whose underlying graph is $K_4$. Note that we do not (yet) impose any sort of planarity or non-crossing condition on wires or edges. The wires differ from the edges in that they only have a vertex at one end. To help distinguish them we will draw the wires thinner and lighter than the edges. Note that these drawings differ from drawings of graphs in a key aspect: a homeomorphism of the plane can change the bi-labeled graph represented by a drawing. This is because such a homeomorphism can change the order in which the input and output wires are drawn which corresponds to changing the order of the input and output vectors. Of course, these drawings are mainly meant to help visualize and build intuition for bi-labeled graphs. Later, when we define operations such as composition and tensor product for bi-labeled graphs, we will see how useful this is.
\end{remark}

\begin{figure}[h!]
\begin{subfigure}{.5\textwidth}
  \centering
  \includegraphics[scale=1]{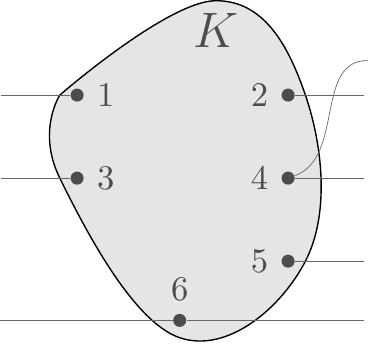}
  \caption{\phantom{a} $\K = (K,(1,3,6),(4,2,4,5,6))$.}
  \label{fig:drawabstract}
\end{subfigure}
\begin{subfigure}{.5\textwidth}
  \centering
  \vspace{.1in}
  \includegraphics[scale=1.2]{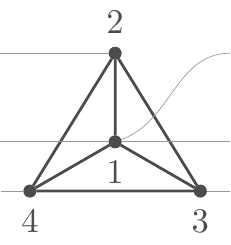}
  \vspace{.1in}
  \caption{\phantom{a} $\K = (K_4,(2,1,4),(1,1,3))$.}
  \label{fig:drawconcrete}
\end{subfigure}
\caption{How to draw bi-labeled graphs.}\label{fig:draw}
\end{figure}

We will not keep the reader in suspense with regards to how the combinatorial notion of bi-labeled graphs relates to the algebraic notion of intertwiners. The connection is provided by homomorphism matrices, whose entries count the number of homomorphisms from a bi-labeled graph to some fixed graph, partitioned according to where the input/output vertices are mapped:

\begin{definition}[$G$-homomorphism matrix]
Given a graph $G$ and a bi-labeled graph $\K = \left(K, \vec{a}, \vec{b} \right) \in \G(\ell,k)$, the \emph{$G$-homomorphism matrix of $\K$}, denoted $T^{\K \to G}$, is the $V(G)^\ell \times V(G)^k$ matrix defined entrywise as
\[\left(T^{\vec{K} \to G}\right)_{\vec{u}\vec{v}} = \# \text{ homomorphisms } \varphi\colon K \to G \text{ s.t. } \varphi(\vec{a}) = \vec{u} \ \& \ \varphi(\vec{b}) = \vec{v}.\]

\label{def:HomMatrix}
\end{definition}

\begin{remark}\label{rem:lovaszknows}
Based on remarks in~\cite{lovasz2012large} and other works, we suspect that the notion of $G$-homomorphism matrices is known to Lov\'{a}sz and others in the field of graph limits. However, we were unable to find an explicit definition these objects in the literature, despite our efforts~\cite{mathoverflow}.
\end{remark}

To gain some intuition, let us discuss some simple examples of the above definitions.

\begin{example}\label{ex:00}
If $\K\in \G(0,0)$, then $T^{\K \to G}$ is simply a $1 \times 1$ matrix (which we may interpret as a scalar at will) whose single entry is the number of homomorphisms $\varphi \colon K \to G$.
\end{example}

\begin{example}\label{ex:01}
If $\K = (K,\varnothing,(b))\in \G(0,1)$, i.e, $\K$ has a single input vertex $b \in V(K)$ and no outputs, then $T^{\K \to G}$ is a row vector whose $u$-entry counts the number of homomorphisms $\varphi\colon K\to G$ with $\varphi(b) = u$. For example, if $\K = (K_3,\emptyset,(1))$ then $T^{\K \to K_3} = (2,2,2)$ as there are two ways to map a $K_3$ to itself when the image of one vertex is fixed. Analogously, if $\K = (K,(b),\varnothing)\in \G(1,0)$, then $T^{\K \to G}$ is a column vector. In fact it is the transpose of the row vector described in the previous case.
\end{example}

\begin{example}\label{ex:A}
If $\K = (K,(a),(b)) \in \G(1,1)$, then $T^{\K\to G}$ is a $V(G) \times V(G)$ matrix. An important case arises when $\vec{K} = (K_2,(1),(2))$. Then it is easy to see that
\[
\left(T^{\K \to G}\right)_{uv} = \begin{cases} 1 & \text{if } u \sim v \\ 0 & \text{o.w.}
\end{cases}
\]
In other words, $T^{\K \to G}$ is equal to $A_G$, the adjacency matrix of $G$. Due to this, we will refer to this bi-labeled graph as $\vec{A}$ (see \Fig{gengraphs}).
\end{example}

\begin{figure}[h!]
  \centering
  \includegraphics[scale=1.5]{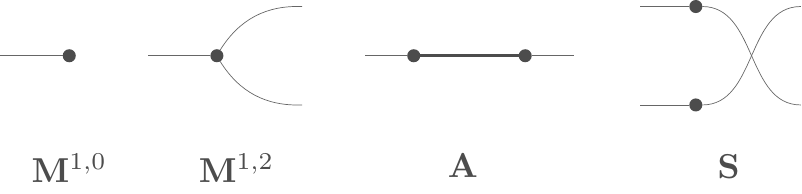}
\caption{Some important bi-labeled graphs.}\label{fig:gengraphs}
\end{figure}

In the above example, note that every vertex of the underlying graph $K$ appeared as an input or output. In this case, every entry of $T^{\K \to G}$ will be either $0$ or $1$. This is because the $\vec{u},\vec{v}$ entry is counting the homomorphisms $\varphi\colon K \to G$ such that $\varphi(\vec{a}) = \vec{u}$ and $\varphi(\vec{b}) = \vec{v}$ which completely determines the function $\varphi$. There are two reasons the corresponding entry could be zero: 1) No such function $\varphi$ exists, i.e., there is some $a_i = a_j$ (or $b_i = b_j$, or $a_i = b_j$) but $u_i \ne u_j$ (or $v_i \ne v_j$, or $u_i \ne v_j$), or 2) the function $\varphi$ exists but is not a homomorphism from $K$ to $G$. If neither $(1)$ or $(2)$ hold, then the corresponding entry will be 1. In the previous example, the 0's in $T^{\K \to G}$ were due to $(2)$. In the following example, they are due to $(1)$. We will make use of the notation $(a^k)$ to denote a tuple of length $k$ in which every entry is $a$. We will also sometimes use this notation in the middle of a tuple, e.g., $(a_1, a_2, a^k, a_3)$ is a tuple of length $k+3$ in which entries $3$ through $k+2$ are $a$.

\begin{example}\label{ex:M}
Let $K$ be the graph with a single vertex $a$ and no edges, and define $\vec{M}^{\ell,k} = (K,(a^\ell),(a^k))$. 
Note that since $K$ has no edges, the matrix $T^{\vec{M}^{\ell,k} \to G}$ depends only on $|V(G)|$. In particular, if $\ell + k > 0$, then
\[\left(T^{\vec{M}^{\ell,k} \to G}\right)_{\vec{u},\vec{v}} = \begin{cases} 1 & \text{if } u_1 = \ldots = u_\ell = v_1 = \ldots = v_k \\ 0 & \text{o.w.}
\end{cases}
\]
On the other hand, $T^{\vec{M}^{0,0} \to G} = (|V(G)|)$. In either case, $T^{\vec{M}^{\ell,k} \to G} = M^{\ell,k}$, the generalized multiplication map from \Sec{qautogroups}. Recall that $M^{1,1} = I$, and thus we will sometimes use $\vec{I}$ to refer to $\vec{M}^{1,1}$. We illustrate the bi-labeled graphs $\vec{M}^{1,0}$ and $\vec{M}^{1,2}$ in \Fig{gengraphs}.
\end{example}

\begin{example}\label{ex:Mloop}
Let $K$ be the graph with a single vertex $a$ with a loop, and define $\mathring{\vec{M}}^{\ell,k} = (K,(a^\ell),(a^k))$. In this case, $T^{\mathring{\vec{M}}^{\ell,k} \to G}$ depends on the graph $G$, specifically it depends on which vertices have loops. If $\ell + k \ge 0$, then 
\[\left(T^{\mathring{\vec{M}}^{\ell,k} \to G}\right)_{\vec{u},\vec{v}} = \begin{cases} 1 & \text{if } u_1 = \ldots = u_\ell = v_1 = \ldots = v_k \text{ and this vertex has a loop}\\ 0 & \text{o.w.}
\end{cases}
\]
The matrix $T^{\mathring{\vec{M}}^{0,0} \to G}$ is the $1 \times 1$ matrix whose entry is equal to the number of loops of $G$.
\end{example}

\begin{example}\label{ex:swap}
Let $K$ be the edgeless graph on vertex set $\{a,b\}$, and let $\vec{S} = (K,(a,b),(b,a))$ (see \Fig{gengraphs}). Then for any graph $G$,
\[\left(T^{\vec{S} \to G}\right)_{(u_1,u_2),(v_1,v_2)} = \begin{cases} 1 & \text{if } u_1 = v_2 \ \& \ u_2 = v_1 \\ 0 & \text{o.w.}
\end{cases}
\]
Thus $T^{\vec{S} \to G}$ is the swap map from \Sec{qautogroups} that is contained in $C_q^G$ if and only if the entries of the fundamental representation of $\qut(G)$ commute, i.e.~if $\qut(G) = \aut(G)$.
\end{example}

\begin{remark}\label{rem:bilabeledmultigraphs}
We have chosen to not allow multiple edges in our bi-labeled graphs. We could consider this more general case, but not doing so simplifies the presentation somewhat (see \Sec{colorededges}). Furthermore, allowing multiple edges in our bi-labeled graphs would not affect the possible $G$-homomorphism matrices because we will only ever consider the case where $G$ is a graph. On the other hand, loops will arise naturally through the operations we apply to bi-labeled graphs and ignoring these would result in different (incorrect) homomorphism matrices, even if we restrict to $G$-homomorphism matrices for $G$ not having loops.
\end{remark}


\subsection{Operations on bi-labeled graphs}\label{sec:blgops}

We have seen how to associate a matrix to a bi-labeled graph. It is thus natural to ask if/how typical matrix operations correspond to operations on bi-labeled graphs. In this section we introduce the graph operations corresponding to matrix product, tensor product, entrywise product, and transposition. Some of these operations are described by Lov\'{a}sz in~\cite{lovasz2012large} and, as mentioned in \Rem{lovaszknows}, we suspect that their correspondence to matrix operations may be previously known, though we could not find this in the literature. On the other hand, the operations of multiplication, tensor product, and transposition have been considered for the \emph{matchgates} introduced by Valiant~\cite{valiant}, and the correspondence to matrix operations (for matrices that count matchings rather than homomorphisms) is known in that case~\cite{xia}.\\


The first operation we introduce corresponds to matrix product/composition, and thus we refer to it as the \emph{composition} of bi-labeled graphs. In~\cite{lovasz2012large} it is referred to as ``concatenation".

\begin{definition}[Composition of bi-labeled graphs]\label{def:comp}
Let $\H_1 = (H_1, \vec{a},\vec{b}) \in \G(\ell,t)$ and $\H_2= (H_2, \vec{a'},\vec{b'}) \in \G(t,k)$ be bi-labeled graphs. 
We define their composition, denoted by $\H_1 \circ \H_2$, to be the bi-labeled graph $\H = (H, \vec{c},\vec{d}) \in \G(\ell,k)$, where $H$ is obtained from the disjoint union $H_1 \cup H_2$ by first adding the edges $e_i = b_i a'_i$ for all $i\in[t]$, then contracting all of the edges $e_i$, and finally removing any resulting multiple edges but keeping the loops. We let $c_i$ be the vertex of $H$ that $a_i$ became after contraction, and similarly let $d_i$ be the vertex $b'_i$ became after contraction.
\end{definition}

Note that in the above definition we could have phrased the construction of $\H_1 \circ \H_2$ as taking the disjoint union of $H_1$ and $H_1$ and then identifying the vertices $b_i$ and $a'_i$ for all $i$. This is equivalent to the above, but it is sometimes useful to have the formalism of edge contractions to work with. Also note that the composition of bi-labeled graphs is easy to visualize in terms of the drawings of bi-labeled graphs described in \Rem{drawing}. We simply draw $\H_1$ on the left and $\H_2$ on the right so that the input wires of the former align with the output wires of the latter. We then join the $i^\text{th}$ input wire of $\H_1$ to the $i^\text{th}$ output wire of $\H_2$ which creates the edge $e_i$ described in \Def{comp}, which we then contract. 


\begin{figure}
\begin{subfigure}{\textwidth}
  \centering
  \includegraphics[scale=1]{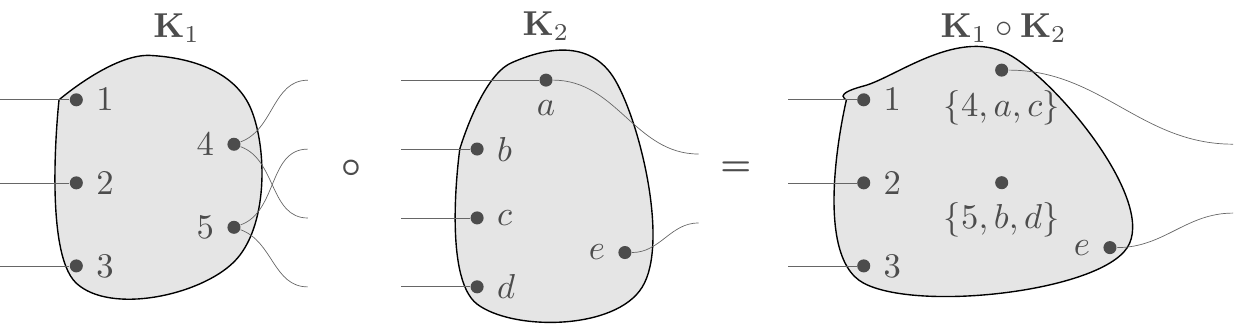}
  \caption{\phantom{a} Generic example of composition.}
  \label{fig:compex1}
\end{subfigure}

\vspace{.2in}

\begin{subfigure}{\textwidth}
  \centering
  \includegraphics[scale=1]{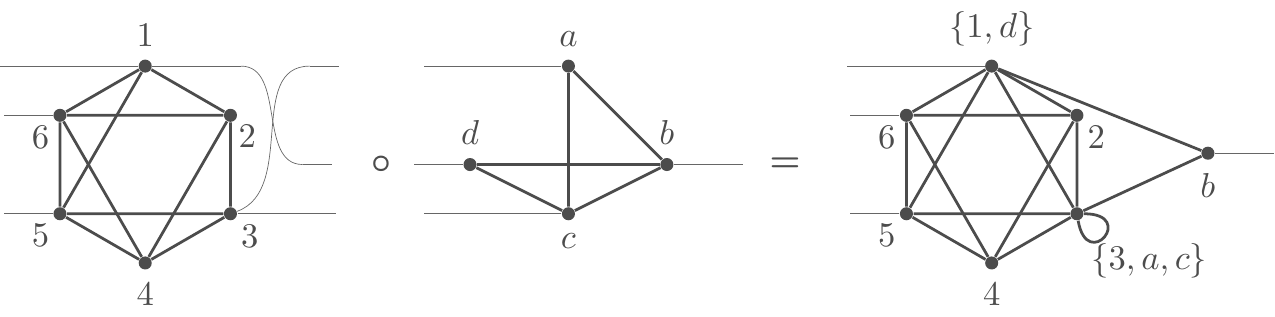}
  \caption{\phantom{a} Concrete example of composition.}
  \label{fig:compex2}
\end{subfigure}
\caption{Illustration of the composition of bi-labeled graphs.}\label{fig:compex}
\end{figure}

We illustrate the above definition with an example shown in \Fig{compex}. In \Fig{compex1} we show the composition of two bi-labeled graphs with generic underlying graphs, and in \Fig{compex2} we give a concrete example. In the latter we are considering the composition of two bi-labeled graphs $\H_1 = (H_1,(1,6,5),(3,1,3))$ and $\H_2 = (H_2,(a,d,c),(b))$. The result is a bi-labeled graph $\H = (H,(\{1,d\},6,5),(b))$ where we have named the vertices of $H$ by the set of vertices from $H_1$ and $H_2$ that were merged to create it (unless it is a singleton in which case it retains its original name). We see that it is possible to create a loop via composition even if the factors do not contain a loop.


Next we define the tensor product and transpose of bi-labeled graphs, which of course will correspond to the analogous operations on homomorphism matrices.

\begin{definition}(Tensor product and transpose of bi-labeled graphs).
Let $\H_1 = (H_1, \vec{a},\vec{b}) \in \G(\ell,k)$, $\H_2= (H_2, \vec{a'},\vec{b'}) \in \G(\ell',k')$ be bi-labeled graphs. We define the \emph{tensor product}, $\H_1 \otimes \H_2$, to be the bi-labeled graph $\H = (H_1 \cup H_2, \vec{aa'},\vec{bb'}) \in \G(\ell+\ell',k+k')$, where $\vec{cc'}$ is the tuple obtained by concatenating $\vec{c}$ with $\vec{c'}$. We define the \emph{transpose} of $\H_1$, to be the bi-labeled graph $\H_1^* = (H_1,\vec{b},\vec{a}) \in \G(k,\ell)$.
\end{definition}
 
Note that we use $^*$ to denote the conjugate transpose of a matrix, but we refer to $\vec{K}^*$ as simply the transpose. This is mainly just to keep the terminology shorter, but also note that $T^{\vec{K} \to G}$ is a real matrix for any graph $G$, and thus transpose and conjugate transpose are the same.

In our drawings of bi-labeled graphs, the tensor product $\H_1 \otimes \H_2$ is drawn by simply drawing $\H_1$ above $\H_2$. Similarly, a drawing of $\H^*$ is obtained from a drawing of $\H$ by simply reflecting about the vertical axis.

Finally, we define the Schur product of two bi-labeled graphs, which will correspond to the Schur, or entrywise, product of matrices.

\begin{definition}(Schur product of bi-labeled graphs)
Let $\vec{H_1} = (H,\vec{a},\vec{b}), \vec{H_2} = (H,\vec{a'},\vec{b'}) \in \G(\ell,k)$ be bi-labeled graphs. We define their Schur product, denoted by $\H_1 \schur \H_2$, to be a bi-labeled graph $\H = (H, \vec{c},\vec{d}) \in \G(\ell,k)$, where $H$ is obtained from the disjoint union $H_1 \cup H_2$ by first, adding the edges $e_i = a_i a'_i$ for all $i\in[\ell]$ and $f_j = b_jb'_j$ for all $j \in [k]$, then contracting all of the edges $e_i$ and $f_j$, and finally removing any resulting multiple edges but keeping the loops. We let $c_i$ be the vertex of $H$ that $a_i$ (and $a'_i$) became after contraction, and similarly let $d_j$ be the vertex $b_j$ (and $b'_j$) became after contraction.
\end{definition}

We remark that it is easy to see that the Schur product of bi-labeled graphs is commutative, and that the composition is not. However, the composition of bi-labeled graphs is associative, and $(\H_1\circ \H_2)^* = \H_2^* \circ \H_1^*$.


It turns out that we can construct the Schur product of $\vec{H}_1,\vec{H}_2 \in \G(\ell,k)$ for certain small values of $\ell, k$ using composition, tensor product, transposition, and $\vec{M}^{1,2}$:

\begin{lemma}\label{lem:buildingschur}
If $\vec{H}_1, \vec{H}_2 \in \G(1,1)$, then $\vec{H}_1 \schur \vec{H}_2 = \vec{M}^{1,2} \circ (\vec{H}_1 \otimes \vec{H}_2)\circ (\vec{M}^{1,2})^*$. If $\vec{H}_1, \vec{H}_2 \in \G(1,0)$, then $\vec{H}_1 \schur \vec{H}_2 = \vec{M}^{1,2} \circ (\vec{H}_1 \otimes \vec{H}_2)$, and similarly for $\G(0,1)$ by taking the transpose. Lastly, if $\vec{H}_1, \vec{H}_2 \in \G(0,2)$, then
\[\H_1 \schur \H_2 = \left(\vec{M}^{0,2} \otimes \vec{M}^{0,2}\right) \circ \left(\vec{I} \otimes \H_1^* \otimes \H_2 \otimes \vec{I}\right) \circ \left(\vec{M}^{2,1} \otimes \vec{M}^{2,1}\right),\]
and similarly for $\G(2,0)$ by taking transpose.
\end{lemma}
\begin{proof}
Let $\vec{H}_1 = (H_1,(a),(b)), \vec{H}_2 = (H_2,(a'),(b')) \in \G(1,1)$. Then $\vec{H}_1 \otimes \vec{H}_2 = (H_1 \cup H_2, (a,a'),(b,b'))$. Let $K$ be the underlying graph of $\vec{M}^{1,2}$ which consists of a single vertex $u$, so $\vec{M}^{1,2} = (K,(u),(u,u))$. Then the underlying graph of $\vec{M}^{1,2} \circ (\vec{H}_1 \otimes \vec{H}_2)$ is formed from $K \cup H_1 \cup H_2$ by adding the edges $ua$ and $ua'$ and contracting them.  Of course this is equivalent to adding the edge $aa'$ to $H_1 \cup H_2$ and contracting it. Moreover, the new output vector will consist of the vertex resulting from this contraction. Similarly, multiplying on the right by $(\vec{M}^{1,2})^*$ will have the same effect as adding the edge $bb'$ and contracting it, with the new input vector consisting of the vertex which results from this contraction. This is precisely the construction of $\vec{H}_1 \schur \vec{H}_2$, and thus we are done. The proof for $\G(1,0)$ and $\G(0,1)$ is similar but easier. The proof for $\G(0,2)$ and $\G(2,0)$ is also similar but more difficult.
\end{proof}

Before we move on to the correspondence between the above bi-labeled graph operations and matrix operations, we will show that the bi-labeled graphs $\vec{M}^{1,2}$ and $\vec{M}^{1,0}$ generate all of the bi-labeled graphs $\vec{M}^{\ell,k}$. First we prove the following:

\begin{lemma}\label{lem:buildmultgraphs}
If $k \ge 1$, then
\begin{align}
\vec{M}^{\ell,k+1} &= \vec{M}^{\ell,k} \circ (\vec{M}^{1,2} \otimes \vec{I}^{\otimes k-1}) \label{eq:Mup}\\
\vec{M}^{\ell,k} &= \vec{M}^{\ell,k+1} \circ (\vec{M}^{2,1} \otimes \vec{I}^{\otimes k-1}) \label{eq:Mdown}
\end{align}
\end{lemma}
\begin{proof}
We will prove the first equation, the second is similar. Note that $\vec{M}^{1,2} \otimes \vec{I}^{\otimes k-1} \in \G(k,k+1)$ and so the composition on the right hand side above is defined. Each of the bi-labeled graphs in this tensor product consist of a single vertex. Let us refer to the vertex corresponding to the $i^\text{th}$ tensor as $u_i$. Then the underlying graph of $\vec{M}^{1,2} \otimes \vec{I}^{\otimes k-1} \in \G(k,k+1)$ is an edgeless graph on vertex set $\{u_1, \ldots, u_k\}$ and its output vector is $(u_1, \ldots, u_k)$. Let $u$ be the single vertex of the underlying graph of $\vec{M}^{\ell,k}$, which has input vector $(u^k)$. The underlying graph of $\vec{M}^{\ell,k} \circ (\vec{M}^{1,2} \otimes \vec{I}^{\otimes k-1})$ is obtained from the graph on vertex set $\{u,u_1. \ldots, u_k\}$ with edges $uu_i$ for $i \in [k]$ by contracting all of these edges. It is easy to see that this results in a graph $K$ consisting of a single vertex, call it $v$, with no loops. Since the product is an element of $\G(\ell,k+1)$ and has only one vertex $v$ with no loop, it must be equal to the bi-labeled graph $(K,(v^\ell),(v^{k+1})) = \vec{M}^{\ell,k+1}$.
\end{proof}

We will write $\langle \vec{H_1}, \ldots, \vec{H}_m\rangle_{\circ,\otimes,*}$ to denote the bi-labeled graphs that can be constructed from $\H_1, \ldots, \H_m$ using finitely many applications of the operations of composition, tensor product and transpose. We can use the above to show $\vec{M}^{\ell,k} \in \langle \vec{M}^{1,0}, \vec{M}^{1,2} \rangle_{\circ,\otimes,*}$ for all $\ell,k \ge 0$.

\begin{lemma}\label{lem:Mgraphs}
For all $\ell,k \ge 0$, we have that $\vec{M}^{\ell,k} \in \langle \vec{M}^{1,0}, \vec{M}^{1,2} \rangle_{\circ,\otimes,*}$.
\end{lemma}
\begin{proof}
It is straightforward to check that $(\vec{M}^{\ell,k})^* = \vec{M}^{k,\ell}$ for all $\ell,k \ge 0$. Also, $\vec{M}^{1,2} \circ \vec{M}^{2,1} = \vec{M}^{1,1} = \vec{I}$ and so we are able to make use of \Lem{buildmultgraphs}. Starting from $\vec{M}^{1,1}$, we can use \Eq{Mup} to construct $\vec{M}^{1,\ell}$ for all $\ell \ge 1$. Using transpose we obtain $\vec{M}^{\ell,1}$ from which we can obtain $\vec{M}^{\ell,k}$ for all $k \ge 1$ by again using \Eq{Mup}. Similarly starting from $\vec{M}^{0,1}$ and applying \Eq{Mup} we can obtain $\vec{M}^{0,\ell}$ for all $\ell \ge 1$ and by transpose we obtain $\vec{M}^{\ell,0}$. Finally, it is easy to check that $\vec{M}^{0,1} \circ \vec{M}^{1,0} = \vec{M}^{0,0}$.
\end{proof}

\begin{lemma}\label{lem:Mloopgraphs}
For all $\ell,k \ge 0$, we have that $\mathring{\vec{M}}^{\ell,k} \in \langle \vec{M}^{1,0}, \vec{M}^{1,2}, \vec{A} \rangle_{\circ,\otimes,*}$.
\end{lemma}
\begin{proof}
Consider $\vec{M}^{1,2} \circ \left(\vec{A} \otimes \vec{I}\right) \circ \vec{M}^{2,1}$. By \Lem{buildingschur}, this is equal to the Schur product $\vec{A} \schur \vec{I}$. In this Schur product the two adjacent vertices of $\vec{A}$ both get identified with the single vertex of $\vec{I}$, thus resulting in a vertex with a loop. Therefore $\vec{A} \schur \vec{I}$ is an element of $\G(1,1)$ whose underlying graph consists of a single vertex with a loop. This must be $\mathring{\vec{M}}^{1,1}$. It is straightforward to see that $\mathring{\vec{M}}^{\ell,k} = \vec{M}^{\ell,1} \circ \mathring{\vec{M}}^{1,1} \circ \vec{M}^{1,k}$ for all $\ell,k \ge 0$. Thus, using \Lem{Mgraphs}, we are done.
\end{proof}

\begin{remark}\label{quantumgraphs}
We note that Lov\'{a}sz considered taking formal linear combinations of (bi-labeled) graphs, referring to these as \emph{quantum graphs}. We will not need to do this, since it will always suffice for us to take linear combinations of homomorphism matrices instead. But it is interesting that we run into yet another use of the word ``quantum" after ``quantum groups" and ``quantum strategies". Note however that Lov\'{a}sz' quantum graphs are not related to other notions of quantum/non-commutative graphs that have been studied in the quantum isomorphism literature~\cite{qfuncs,morita,bigalois}.
\end{remark}

\subsection{Correspondence to matrix operations}\label{sec:correspondence}

Now that we have introduced the operations on bi-labeled graphs that correspond to matrix operations, we will prove this correspondence. We begin with composition.


\begin{lemma}\label{lem:compcorr}
Let $\H_1 = (H_1, \vec{a},\vec{b}) \in \G(\ell,t)$, $\H_2= (H_2, \vec{a'},\vec{b'}) \in \G(t,k)$ be bi-labeled graphs and $G$ be a graph. Then 
\be
	T^{\H_1 \to G} T^{\H_2 \to G} = T^{\H_1 \circ \H_2 \to G}.
\label{eq:Mult}
\ee
\end{lemma}
\begin{proof}
Let $\H_1,\H_2$, and $G$ be as in the lemma statement and let $\H = (H,\vec{c},\vec{d}) = \H_1 \circ \H_2$. Consider the $\vec{u}\vec{v}$-entry of the left-hand side of \Eq{Mult}: 
\begin{align*}
	&\left( T^{\H_1 \to G} T^{\H_2 \to G} \right)_{\vec{u}\vec{v}} = 
	\sum_{\vec{w} \in V(G)^t} \left( T^{\H_1 \to G} \right)_{\vec{u}\vec{w}}  \left( T^{\H_2 \to G} \right)_{\vec{w}\vec{v}} \\
	=& \sum_{\vec{w} \in V(G)^t} \abs[\Big]{\set[\big]{\varphi_1\colon H_1 \to G \mid \varphi_1(\vec{a}) = \vec{u}, \varphi_1(\vec{b}) = \vec{w} }} \cdot 
	\abs[\Big]{\set[\big]{\varphi_2\colon H_2 \to G \mid \varphi_2(\vec{a'}) = \vec{w}, \varphi_2(\vec{b'}) = \vec{v} }} 
\end{align*}
Now fix $\vec{u},\vec{w}$, and $\vec{v}$. If there exists $i,j$ such that $a_i = a_j$ and $u_i \ne u_j$, or $b_i = b_j$ and $w_i \ne w_j$ or $a'_i = a'_j$ and $w_i \ne w_j$, or $b'_i = b'_j$ and $v_i \ne v_j$, then the corresponding term in the above summation is necessarily zero. Otherwise let $\varphi_1\colon H_1 \to G$ and $\varphi_2\colon H_2 \to G$ be a pair of homomorphisms such that $\varphi_1(\vec{a}) = \vec{u}$, $\varphi_1(\vec{b}) = \vec{w}$, $\varphi_2(\vec{a'}) = \vec{w}$, and $\varphi_2(\vec{b'}) = \vec{v}$. We can construct a homomorphism $\varphi\colon H \to G$ such that $\varphi(\vec{c}) = \vec{u}$ and $\varphi(\vec{d}) = \vec{v}$ from $\varphi_1$ and $\varphi_2$ as follows. Let $e_i = b_ia'_i$ for $i \in [t]$, and recall that $H$ is constructed from $H_1 \cup H_2$ by adding these edges and then contracting them. Thus, each vertex $x \in V(H)$, corresponds to a connected component $H_x$ in the (multi)graph on vertex set $V(H_1)\cup V(H_2)$ with edges $e_i$ for $i \in [t]$. If $H_x$ consists of a single isolated vertex, say $x'$, then $x'$ does not appear in either $\vec{b}$ or $\vec{a'}$. In this case, we let $\varphi(x)$ be equal to either $\varphi_1(x')$ or $\varphi_2(x')$ depending on whether $x' \in V(H_1)$ or $x' \in V(H_2)$. Otherwise, $H_x$ consists of edges $e_i = b_ia'_i$ for $i \in S$ for some $S \subseteq [t]$. As $\varphi_1(b_i) = w_i = \varphi_2(a'_i)$ for all $i$, and since $H_x$ is connected, either there is some common value $w \in V(G)$ such that $w_i = w$ for all $i \in S$, or there exists $i,j \in S$ such that $a_i = a_j$ or $b_i = b_j$ but $w_i \ne w_j$. The latter case is a contradiction to our assumption, thus we must be in the former case and we can let $\varphi(x) = w$. Note that whether $H_x$ is an isolated vertex or contains edges, we have that every vertex $x'$ of $H_x$ is mapped to the same vertex $y$ of $G$ by $\varphi_1$ or $\varphi_2$ depending on whether $x' \in V(H_1)$ or $x' \in V(H_2)$, and we define $\varphi(x)$ to be equal to $y$. Consider now $c_i \in V(H)$. By definition of $\H_1 \circ \H_2$, the component $H_{c_i}$ must contain $a_i$, and therefore $\varphi(c_i) = \varphi_1(a_i) = u_i$. Similarly we can show that $\varphi(d_j) = v_j$ as desired. Finally, suppose that $x \sim y$ in $H$. Then by the definition of contraction there exist $x' \in V(H_x)$ and $y' \in V(H_y)$ such that $x'$ and $y'$ are adjacent in the graph obtained from the disjoint union of $H_1$ and $H_2$ by adding the edges $e_i$. Moreover, $x'$ and $y'$ must be adjacent via an edge that was not contracted, i.e., not one of the $e_i$ edges. Therefore, we have that $x'$ and $y'$ are adjacent vertices in $H_1$ or $H_2$. In the former case, it follows that $\varphi(x) = \varphi_1(x') \sim \varphi_1(y') = \varphi(y)$, and similarly in the latter case. Thus $\varphi$ is indeed a homomorphism from $H$ to $G$ that is counted in the $\vec{u}\vec{v}$-entry of $T^{\H \to G}$.


We claim that the above construction of $\varphi$ from $\varphi_1$ and $\varphi_2$ is a bijection between the pairs of homomorphisms being counted in the $\vec{u}\vec{v}$-entries of $T^{\H \to G}$ and $T^{\H_1 \to G} T^{\H_2 \to G}$ respectively. To see that it is injective, suppose that $(\varphi_1,\varphi_2)$ and $(\varphi'_1,\varphi'_2)$ are two distinct pairs of homomorphisms contributing to $(T^{\H_1 \to G} T^{\H_2 \to G})_{\vec{u}\vec{v}}$, and let $\varphi$ and $\varphi'$ be the homomorphisms from $H$ to $G$ arising from the construction above. Without loss of generality, assume that $\varphi_1 \ne \varphi'_1$, i.e., that there exists $x' \in V(H_1)$ such that $\varphi_1(x') \ne \varphi'_1(x')$. Let $x \in V(H)$ be such that $x' \in V(H_x)$. Then $\varphi(x) = \varphi_1(x') \ne \varphi'_1(x') = \varphi'(x)$. For surjectivity, suppose that $\varphi$ is a homomorphism from $H$ to $G$ such that $\varphi(\vec{c}) = \vec{u}$ and $\varphi(\vec{d}) = \vec{v}$. For each $x' \in V(H_1)$, define $\varphi_1(x')$ as $\varphi(x)$ for $x \in V(H)$ such that $x' \in V(H_x)$ (note that this exists and is unique). 
Similarly define $\varphi_2$. Since $b_i$ and $a'_i$ get contracted to the same vertex of $H$, we will have that $\varphi_1(b_i) = \varphi_2(a'_i)$. Thus we can define $w_i = \varphi_1(b_i)$ for each $i \in [t]$. It is easy to see that $(\varphi_1,\varphi_2)$ is a pair of homomorphisms counted by the $\vec{u},\vec{v}$-entry of $T^{\H_1 \to G} T^{\H_2 \to G}$.

%
%
%
%
\end{proof}

The proof for Schur product is similar, and so we will omit it.

\begin{lemma}\label{lem:schurcorr}
Let $\H_1 = (H_1, \vec{a},\vec{b}), \H_2= (H_2, \vec{a'},\vec{b'}) \in \G(\ell,k)$ be bi-labeled graphs and $G$ be a graph. Then 
\[T^{\H_1 \to G} \schur T^{\H_2 \to G} = T^{\H_1\schur\H_2 \to G}.\]
\end{lemma}

For tensor product we have the following:

\begin{lemma}\label{lem:tensorcorr}
Let $\H_1 = (H_1, \vec{a},\vec{b}) \in \G(\ell,k)$, $\H_2= (H_2, \vec{a'},\vec{b'}) \in \G(\ell',k')$ be bi-labeled graphs and $G$ be a graph. Then 
\be
	T^{\H_1 \to G} \otimes T^{\H_2 \to G} = T^{\H_1\otimes\H_2 \to G}
\label{eq:Tensor}
\ee
\end{lemma}
\begin{proof}
Let $\H_1$, $\H_2$, and $G$ be as in the lemma statement and let $\H = (H,\vec{c},\vec{d}) = \H_1 \otimes \H_2$. Recall that $H = H_1 \cup H_2$, $\vec{c} = \vec{a}\vec{a'}$, and $\vec{d} = \vec{b}\vec{b'}$. The rows of the matrices in \Eq{Tensor} are indexed by elements of $V(G)^{\ell + \ell'}$ which can be written as $\vec{u}\vec{u'}$ for $\vec{u} \in V(G)^{\ell}$, $\vec{u'} \in V(G)^{\ell'}$. Similarly, the columns can be indexed by $\vec{v}\vec{v'}$ for $\vec{v} \in V(G)^{k}$, $\vec{v'} \in V(G)^{k'}$. The $\vec{u}\vec{u'},\vec{v}\vec{v'}$-entry of the left-hand side of \Eq{Tensor} is
\begin{align*}
&\left(T^{\H_1 \to G} \otimes T^{\H_2 \to G}\right)_{\vec{u}\vec{u'},\vec{v}\vec{v'}} = \left(T^{\H_1 \to G}\right)_{\vec{u},\vec{v}} \left(T^{\H_2 \to G}\right)_{\vec{u'},\vec{v'}} \\
=& \abs[\Big]{\set[\big]{\varphi_1\colon H_1 \to G \mid \varphi_1(\vec{a}) = \vec{u}, \varphi_1(\vec{b}) = \vec{v} }} \cdot 
	\abs[\Big]{\set[\big]{\varphi_2\colon H_2 \to G \mid \varphi_2(\vec{a'}) = \vec{u'}, \varphi_2(\vec{b'}) = \vec{v'} }} \\
	=& \abs[\Big]{\set[\big]{\varphi \colon H_1 \cup H_2 \to G \mid \varphi(\vec{a}) = \vec{u}, \varphi(\vec{a'}) = \vec{u'}, \varphi(\vec{b}) = \vec{v}, \varphi(\vec{b'}) = \vec{v'} }} \\
	=& \left(T^{\H_1 \otimes \H_2 \to G}\right)_{\vec{u}\vec{u'},\vec{v}\vec{v'}}
\end{align*}
where the third equality comes from the fact that any two such homomorphisms $\varphi_1$ and $\varphi_2$ can be used to construct a homomorphism $\varphi$ from $H_1 \cup H_2$ to $G$ by simply defining $\varphi$ according to $\varphi_1$ on $V(H_1)$ and according to $\varphi_2$ on $V(H_2)$, and moreover this construction is clearly bijective.
\end{proof}

The proof for the transpose operation is straightforward and so we omit it:

\begin{lemma}\label{lem:transposecorr}
Let $\H = (H, \vec{a},\vec{b}) \in \G(\ell,k)$ be a bi-labeled graph and $G$ be a graph. Then 
\[\left(T^{\H \to G}\right)^* = T^{\H^* \to G}.\]\qed
\end{lemma}

\noindent Recall from \Thm{chassaniol} and \Rem{chassaniol} that for a graph $G$,
\[C_q^G = \langle M^{1,0}, M^{1,2}, A_G\rangle_{+,\circ, \otimes, *} = \spn\left\{T : T \in \langle M^{1,0}, M^{1,2}, A_G\rangle_{\circ, \otimes, *}\right\}.\]
Also, from \Ex{A} and \Ex{M} we have that $T^{\vec{A} \to G} = A_G$ and $T^{\vec{M}^{\ell,k} \to G} = M^{\ell,k}$. Thus $C_q^G = \spn\{T : T \in \langle T^{\vec{M}^{1,0} \to G}, T^{\vec{M}^{1,2} \to G}, T^{\vec{A} \to G}\rangle_{\circ,\otimes,*}\}$. Finally, from the correspondence between matrix and bi-labeled graph operations proven above, we have the following:

\begin{theorem}\label{thm:reformulation}
For any graph $G$,
\[C_q^G = \spn\left\{T^{\vec{K} \to G} : \vec{K} \in \langle\vec{M}^{1,0},\vec{M}^{1,2},\vec{A}\rangle_{\circ,\otimes,*}\right\}.\]
\end{theorem}

\section{Planar Graphs}\label{sec:planar}


In the following section we will present the class of planar bi-labelled graphs, which we will show is closed under composition, tensor product, and transposition. To do this, we will need some tools for dealing with and manipulating planar graphs, which we will present here.

\begin{definition}\label{def:planarity}
Informally, a \emph{planar embedding} of a (multi)graph $G$ is a drawing of $G$ in the plane $\mathbb{R}^2$ such that no edges cross. Formally, a planar embedding consist of an \emph{injective} function $f\colon V(G) \to \mathbb{R}^2$ and \emph{continuous} functions $h_{e}\colon [0,1] \to \mathbb{R}^2$ for each edge 
$e \in E(G)$ between vertices $u,v$ such that $h_{e}(0) = f(u)$, $h_{e}(1) = f(v)$, and for any two distinct edges $e, e'$, we have $h_{e}(s) \ne h_{e'}(t)$ unless $s,t \in \{0,1\}$. Moreover we require that $h_e(s) \ne h_e(t)$ for distinct $s,t \in [0,1]$ unless $e$ is a loop and $\{s,t\} = \{0,1\}$. A (multi)graph $G$ is \emph{planar} if there exists some planar embedding of $G$.

\begin{remark}\label{rem:planarmultigraphs}
It is well known that a multigraph is planar if and only if the graph obtained by removing all loops and replacing all multiple edges with a single edge is planar. Thus for the most part it suffices to discuss planarity of graphs rather than the more general multigraphs.
\end{remark}

\begin{remark}\label{rem:fary}
A result of Fary~\cite{fary} states that given any planar embedding of a graph, there is a homeomorphism of the plane that maps the embedding to a straight-line embedding, i.e., a planar embedding in which each edge is embedded as a single line segment. Thus we can alway assume that our planar embeddings are straight-line embeddings if we desire. Actually it suffices to assume that they are polygonal embeddings, i.e., each edge is embedded as finitely many straight line segments (this also allows for the embedding of multiple edges between two vertices, whereas a straight line embedding does not). We will implicitly assume this throughout the rest of the paper.
\end{remark}

There are some basic operations of graphs that are known to preserve planarity. The most well-known are of course edge contraction and vertex/edge deletion. Any graph $H$ obtained from a graph $G$ via these operations is known as a \emph{minor} of $G$. In this language planar graphs are said to be a \emph{minor-closed} family of graphs. Another well-known planarity preserving operation is \emph{subdivision}. Given a graph $G$ with an edge $e = uv$, the operation of replacing $e$ with a path $u,w,v$ is known as \emph{subdividing} $e$. We can also ``subdivide an edge $k$ times" by replacing the edge with a path of length $k+1$. Note that when subdividing an edge we will often specify the names of the vertices on the resulting path. A graph $H$ is said to be a \emph{subdivision} of $G$ if $H$ can be obtained from $G$ by repeatedly subdividing edges. Further, a graph $G$ is said to be a \emph{topological minor} of $H$ (or that $H$ contains $G$ as a subdivision) if $H$ has a subgraph isomorphic to a subdivision of $G$. The famous theorems of Kuratowski~\cite{kuratowski} and Wagner~\cite{wagner} characterize planar graphs in terms of subdivisions and minors:

\begin{theorem}[Kuratowski]
A graph $G$ is planar if and only if it does not contain $K_{3,3}$ or $K_5$ as a topological minor.
\end{theorem}

\begin{theorem}[Wagner]
A graph $G$ is planar if and only if it does not contain $K_{3,3}$ or $K_5$ as a minor.
\end{theorem}

We will also sometimes make use of the inverse operation of subdividing: \emph{unsubdividing}. Given a path $P = u_1, \ldots, u_m$ in a graph $G$ in which all of the interior vertices of $P$ (i.e., $u_2, \ldots, u_{m-1}$) have degree two, we refer to the operation of replacing $P$ with a single edge between $u_1$ and $u_m$ as unsubdividing $P$. It is easy to see that unsubdividing also preserves planarity (it is equivalent to contracting all but one edge of the path $P$).


The \emph{faces} of a planar embedding are the connected regions of $\mathbb{R}^2 \setminus \left(\cup_{e \in E(G)} \mathrm{Im}(h_{e}) \right)$, where $\mathrm{Im}(h_{e})$ is the image of $h_{e}$. We say that a vertex/edge is incident to a face if the embedding of the vertex/edge is contained in the closure of that face. The (graph-theoretical) \emph{boundary} of a face is the subgraph of $G$ consisting of the vertices and edges incident to that face.
\end{definition}

We remark that it is well-known that in any planar embedding of a graph $G$, precisely one of the faces will be unbounded. We refer to this as the \emph{outer face}. We will frequently need to make use of the notion of facial cycles:

\begin{definition}
For a planar graph $G$ with cycle $C$, we say that $C$ is a \emph{facial cycle}\footnote{It is more standard for facial cycles to only be defined for a fixed planar embedding of a graph, but the given definition is better suited for our needs.} of $G$ if there exists a planar embedding of $G$ such that $C$ is the boundary of a face.
\end{definition}

\begin{remark}\label{rem:facial}
It is well-known that if $C$ is a facial cycle of a planar graph $G$, then there is a planar embedding of $G$ in which $C$ is the boundary of the outer face.
\end{remark}

We will also need the notion of cyclic order of edges incident to a vertex for our results:

\begin{definition}\label{def:cyclicorder}
Let $G$ be a planar graph and fix a planar embedding of $G$. For any vertex $v$ of $G$, there exists a radius $r > 0$ such that the circle $C_r$ of radius $r$ centered at $v$ intersects each edge incident to $v$ exactly once~\footnote{Here is one place where our assumption that our planar embeddings are polygonal is needed.}. We say that the edges incident to $v$ occur in cyclic order $e_1, \ldots, e_k$ in this embedding if, starting from the intersection of $C_r$ and $e_1$ and moving around $C_r$ (in some direction), we intersect the edges $e_1, \ldots, e_k, e_1$ in that order.
\end{definition}

We now present several lemmas that we will make use of in \Sec{planarblg} and \Sec{tale}. They are basic results about planar (multi)graphs whose proofs are straightforward, thus we present them without proof. The first lemma offers a useful characterization of facial cycles. 

\begin{lemma}\label{lem:facialcycle}
Suppose $G$ is a planar multigraph with cycle $C = u_1, \ldots, u_k$. Let $G'$ be the multigraph obtained from $G$ by adding a new vertex $u$ adjacent to each $u_i$. Then $C$ is a facial cycle of $G$ if and only if $G'$ is planar. Moreover, any planar embedding of $G$ in which $C$ is bounding some face $F$ can be extended to a planar embedding of $G'$ where $u$ is embedded in $F$ and then connected to the vertices $u_1, \ldots, u_k$.
\end{lemma}

Somewhat similarly we have the following:

\begin{lemma}\label{lem:commonface}
Let $G$ be a planar multigraph and let $S \subseteq V(G)$. Then $G$ has a planar embedding which has a face incident to every vertex of $S$ if and only if the multigraph $G'$ obtained from $G$ by adding a new vertex $s$ adjacent to every element of $S$ is planar. Moreover, any planar embedding of $G$ in which the vertices of $S$ are incident to a common face $F$ can be extended to a planar embedding of $G'$ where $s$ is embedded in $F$ and then connected to every vertex of $S$.
\end{lemma}

The next lemma allows us to add cycles around a vertex:

\begin{lemma}\label{lem:addcycle}
Let $G$ be a planar multigraph and $v \in V(G)$. Suppose that there is a planar embedding of $G$ in which the edges incident to $v$ occur in cyclic order $e_1, \ldots, e_k$. Let $G'$ be the multigraph obtained from $G$ by subdividing each of $e_1, \ldots, e_k$ with vertices $v_1,\ldots, v_k$ and adding edges $v_iv_{i+1}$ for $i=1,\ldots,k$ where indices are taken modulo $k$. Then $G'$ is planar.
\end{lemma}

Finally, the following lemma allows us to join up two planar graphs with given facial cycles in a controlled manner:

\begin{lemma}\label{lem:cyclejoin}
Suppose that $G$ and $H$ are planar graphs with facial cycles $C = u_1, \ldots, u_\ell$ and $D = v_1, \ldots, v_k$ respectively. Let $X$ be the graph obtained from $G \cup H$ by adding edges $u_1v_1, \ldots, u_m v_m$ for some $m \le \min\{\ell,k\}$. Then $X$ is planar and $C' = u_m,u_{m+1}, \ldots, u_\ell, u_1, v_1,v_k,v_{k-1},\ldots, v_m$ is a facial cycle.
\end{lemma}

\section{Planar bi-labeled graphs}\label{sec:planarblg}

At the end of \Sec{correspondence}, we saw in \Thm{reformulation} that for any graph $G$, we have that $C_q^G = \spn\{T^{\vec{K}\to G} : \vec{K} \in \langle\vec{M}^{1,0},\vec{M}^{1,2},\vec{A}\rangle_{\circ,\otimes,*}\}$. This perhaps falls short of what one might deem a characterization of $C_q^G$, since it is not a priori obvious what bi-labeled graphs can be generated by $\vec{M}^{1,2},\vec{M}^{1,0},\vec{A}$ and the operations of composition, tensor product, and transposition. Rather, \Thm{reformulation} is a translation of the algebraic problem of determining $C_q^G$ (or, equivalently $\langle M^{1,0}, M^{1,2}, A_G\rangle_{\circ,\otimes,*}$) into the combinatorial problem of characterizing $\langle\vec{M}^{1,0},\vec{M}^{1,2},\vec{A}\rangle_{\circ,\otimes,*}$. In this section, we will introduce a family of \emph{planar bi-labeled graphs} $\P$, which we will prove is equal to $\langle\vec{M}^{1,0},\vec{M}^{1,2},\vec{A}\rangle_{\circ,\otimes,*}$ in Section~\ref{sec:tale}.

To define $\P$, we first need the following:

\begin{definition}
Given a bi-labeled graph $\vec{K} = (K,\vec{a},\vec{b}) \in \G(\ell,k)$, define the graph $K^\circ := K^\circ(\vec{a},\vec{b})$ as the graph obtained from $K$ by adding the cycle $C = \alpha_1, \ldots, \alpha_\ell, \beta_k, \ldots, \beta_1$ of new vertices, and edges $a_i\alpha_i$, $b_j \beta_j$ for all $i \in [\ell], j \in [k]$. We refer to the cycle $C$ as the \emph{enveloping cycle} of $K^\circ$. We further define $K^\odot := K^\odot(\vec{a},\vec{b})$ as the graph obtained from $K^\circ$ by adding an additional vertex adjacent to every vertex of the enveloping cycle.
\end{definition}

Note that although the two graphs defined above depend on the input and output vectors of $\vec{K} = (K,\vec{a},\vec{b})$, we will typically refer to them as simply $K^\circ$ and $K^\odot$ when there should be no confusion. 

\begin{remark}
In the case $\ell = k = 0$, we simply let $K^\circ = K^\odot = K$. If $\ell + k = 1$, then the enveloping cycle of $K^\circ$ consists of a single vertex with a loop. If $\ell+k = 2$, then the enveloping cycle is two vertices with two edges between them. Thus in this case we should actually refer to $K^\circ$ and $K^\odot$ as multigraphs, but we will mostly overlook this. If we really wanted to define $K^\circ$ and $K^\odot$ as graphs in this case we could simply replace the double edge with a single edge, and this would not affect \Def{PBG} of $\P$. However it is convenient for us to have the double edge so we can treat things in a uniform manner.
\end{remark}

For intuitions sake, we think of drawing $K^\circ$ with the enveloping cycle surrounding the graph $K$, and with the vertices $\alpha_1, \ldots, \alpha_\ell$ on the left-hand side with $\alpha_1$ at the top left, and $\beta_1, \ldots, \beta_k$ on the right-hand side with $\beta_1$ at the top right. The idea is to take a drawing of a bi-labeled graph $\K$, as described in \Rem{drawing}, and create $K^\circ$ by drawing a border around the drawing so that the vertex-less ends of the wires just touch the border and then adding vertices at these intersections.

Now we are able to define our class $\P$ of planar bi-labeled graphs:

\begin{definition}[Planar bi-labeled graphs]\label{def:PBG}
For any $\ell,k \in \mathbb{N}$,
\[\P(\ell,k) = \{\vec{K} \in \G(\ell,k) : K^\odot \text{ is planar.}\},\]
and we let $\P = \cup_{\ell,k = 0}^\infty \P(\ell,k)$.
\end{definition}

By \Lem{facialcycle}, we also have the following equivalent definition:

\begin{definition}[Planar bi-labeled graphs]\label{def:PBG2}
For any $\ell,k \in \mathbb{N}$,
\[\P(\ell,k) = \{\vec{K} \in \G(\ell,k) : K^\circ \text{ is planar and its enveloping cycle is facial.}\},\]
and we let $\P = \cup_{\ell,k = 0}^\infty \P(\ell,k)$.
\end{definition}

Again we can gain intuition by thinking in terms of drawings of bi-labeled graphs. Membership in $\P$ corresponds to being able to draw a bi-labeled graph, as described in \Rem{drawing}, in a planar manner, i.e, without any wires or edges crossing.

In the case of small values of $\ell + k$, the condition of membership in $\P(\ell,k)$ can be rephrased somewhat:

\begin{lemma}\label{lem:Psmalllk}
Let $\vec{K} = (K,\vec{a},\vec{b}) \in \G(\ell,k)$ be a bi-labeled graph, and let $S \subseteq V(K)$ be the set of vertices appearing in $\vec{a}$ or $\vec{b}$. If $\K \in \P(\ell,k)$, then $K$ has a planar embedding in which every vertex of $S$ is incident to a common face. If $\ell + k \le 3$, then this condition is also sufficient for $\K \in \P(\ell,k)$.
\end{lemma}
\begin{proof}
Let $\K = (K,\vec{a},\vec{b}) \in \P(\ell,k)$ and consider $K^\odot$ (which must be planar), which has enveloping cycle $C = \alpha_1, \ldots, \alpha_\ell,\beta_k, \ldots, \beta_1$ where $\alpha_i \sim a_i$, $\beta_j \sim b_j$, and there is a vertex $z$ adjacent to every vertex of $C$. Contract all edges incident to $z$ and all edges of $C$ to obtain a new graph $K'$ and let the vertex corresponding to the nontrivial contracted component be $z'$. Since $K^\odot$ was planar, we have that $K'$ is planar, and it is easy to see that $K'$ is simply $K$ plus the vertex $z'$ that is adjacent to every vertex appearing in $\vec{a}$ or $\vec{b}$. Thus, by \Lem{commonface} $K$ has a planar embedding in which every vertex appearing in $\vec{a}$ or $\vec{b}$ is incident to a common face.



Now let $\ell+k \le 3$ and let $S$ be the set vertices in the input/output vectors. Suppose that $K$ has a planar embedding which has a face incident to every vertex in $S$. By \Lem{commonface}, the (multi)graph obtained from $K$ by adding a vertex $z$ adjacent to every element of $S$ is planar (if $s \in S$ appears more than once in $\vec{a}$ or $\vec{b}$, then we add a number of edges between $s$ and $z$ equal to the number of times $s$ appears). Furthermore, by \Lem{addcycle}, we can subdivide the edges incident to $z$ with vertices $\alpha_1, \ldots, \alpha_{\ell + k}$ and add a cycle through these vertices while remaining planar. Note that since $\ell + k \le 3$, we do not need to worry about the cyclic order of the edges incident to $z$. It is straightforward to see that the graph we have constructed is $K^\odot$. 


\end{proof}

\begin{remark}\label{rem:Psmalllk}
Note that for $\ell + k \le 1$, the above implies that $\K = (K,\vec{a},\vec{b}) \in G(\ell,k)$ is an element of $\P(\ell,k)$ if and only if $K$ is planar. In the case where $\ell + k = 2$, it is easy to see that $\K \in \P(\ell,k)$ is equivalent to the planarity of the graph obtained from $K$ by adding an edge between the two vertices occurring in $\vec{a}$, $\vec{b}$.
\end{remark}


To show that the second claim of \Lem{Psmalllk} does not hold for $\ell + k > 3$, we have the following example:

\begin{example}\label{ex:swapgraph}
Let $K$ be an empty graph on vertex set $\{a,b\}$, and define $\K = (K,(a,b),(b,a))$. Recall that this is the bi-labeled graph from \Ex{swap}, whose corresponding homomorphism matrix is the swap map. As $K$ has no edges, any planar embedding has only one face, and thus the vertices in the input/output vectors are always incident to a common face. However, it is easy to see that $K^\circ$ is isomorphic to a complete graph on four vertices in which two non-incident edges have been subdivided by the vertices $a$ and $b$. 
Therefore, $K^\odot$ is a subdivision of $K_5$ and thus non-planar. So we see that $\K \notin \P(2,2)$, which is desirable since the swap map is not contained in $C_q^G$ generically. We remark that this example also illustrates that it is not enough to simply require that $K^\circ$ be planar.
\end{example}

The purpose of introducing the enveloping cycle and graphs $K^\circ$, $K^\odot$ is that we need some conditions on the input/output vectors of the bi-labeled graphs in $\P$ so that when we compose them (or perform other operations on them), they remain planar. To illustrate such a condition is necessary, we give the following example:

\begin{example}\label{ex:notinP}
Let $K$ be 5-cycle on vertices $b_1, b_2, b_3, b_4, b_5$ in that order, and let $H$ be a 5-cycle on vertices $a_1, a_3, a_5, a_2, a_4$ in that order. Obviously, $K$ and $H$ are both planar. However, if $\vec{K} = (K,\varnothing,(b_1, b_2,b_3,b_4,b_5))$ and $\vec{H} = (H,(a_1,a_2,a_3,a_4,a_5),\varnothing)$, then it is not hard to see that the underlying graph of $\vec{K} \circ \vec{H}$ is the complete graph $K_5$ which is of course not planar.

One can check that $K^\circ$ is the 5-prism graph, with the enveloping cycle being one of the two 5-cycles, and thus $\vec{K} \in \P(0,5)$. On the other hand, $H^\circ$ is the Petersen graph, which is known to be non-planar. Thus $\vec{H} \notin \P$.
\end{example}


Before moving on, we give some examples of bi-labeled graphs that are contained in $\P$.

\begin{example}\label{ex:MinP}
Recall that $\vec{M}^{\ell,k} = (K,(v^\ell),(v^k))$ where $v$ is the only vertex of $K$ which has no loops. In this case the graph $K^\odot$ consists of a cycle of length $\ell + k$ and two vertices adjacent to every vertex on this cycle but not to each other. This is sometimes called the $\ell + k$ bi-wheel, and it is easy to see that it is planar. Thus we have that $\vec{M}^{\ell,k} \in \P(\ell,k)$. Similarly we have that $\mathring{\vec{M}}^{\ell,k} \in \P(\ell,k)$.
\end{example}

\begin{remark}\label{rem:xi}
Recall that $M^{1,1} = I$ and $M^{2,0} = \xi = \sum_{i} e_i \otimes e_i$ from the definition of tensor categories with duals. Thus by \Ex{M} and \Ex{MinP} we have that $\spn\{T^{\K \to G} : \K \in \P\}$ contains both $I$ and $\xi$.
\end{remark}

\begin{example}\label{ex:AinP}
Recall that $\vec{A} = (K,(u),(v))$ where $K$ is the complete graph on vertex set $\{u,v\}$. In this case one can obtain $K^\odot$ from a complete graph on three vertices by adding two additional edges between one pair of vertices (so that there are three edges in total between them), and then subdividing one of these edges twice. This is clearly planar and thus $\vec{A} \in \P(1,1)$.
\end{example}

\subsection{Closure properties}\label{sec:closureprops}

In this section we will show that the class $\P$ of planar bi-labeled graphs is closed under composition, tensor product, and transposition. We begin with composition, which is the most difficult case.

\begin{lemma}\label{lem:productclosed}
If $\mathbf{H}_1 \in \P(\ell,k)$ and $\mathbf{H}_2 \in \P(k,m)$, then $\mathbf{H}_1 \circ \mathbf{H}_2 \in \P(\ell,m)$.
\end{lemma}
\begin{proof}
Let $\H = (H,\vec{a'},\vec{d'})$ be the bi-labeled graph such that $\H = \H_1 \circ \H_2$ where $\H_1 = (H_1,\vec{a},\vec{b}) \in \P(\ell,k)$ and $\H_2 = (H_2, \vec{c},\vec{d})  \in \P(k,m)$. To show that $\H \in \P(\ell,m)$, we must show that $H^\odot$ is planar. We will show this by constructing $H^\odot$ from $H_1^\circ$ and $H_2^\circ$ in a careful manner. The idea of the construction is illustrated in \Fig{compproof}, though there we only show $H^\circ$ for aesthetic purposes.

First we can consider the case where $\ell = k = 0$. In this case, we may simply take the planar embedding of $H_2^\odot$ and embed $H_1^\odot = H_1$ inside one of its faces. Similarly the lemma holds if $k = m = 0$.


Now suppose that $\ell + k > 0$ and $k + m > 0$. Let $C_1 = \alpha_1, \ldots, \alpha_\ell, \beta_k, \ldots, \beta_1$ be the enveloping cycle of $H_1^\circ$ so that $a_i \sim \alpha_i$ and $b_j \sim \beta_j$ for all $i,j$. Similarly let $C_2 = \gamma_1,\ldots, \gamma_k, \delta_m, \ldots, \delta_1$ be the enveloping cycle of $H_2^\circ$ so that $c_i \sim \gamma_i$ and $d_j \sim \delta_j$. 

We begin by taking the disjoint union of $H_1^\circ$ and $H_2^\circ$. Now let us subdivide the edges $\alpha_1\beta_1$ and $\alpha _k\beta_\ell$, adding vertices $x_1$ and $y_1$ respectively. If $k=0$, then we instead subdivide the edge $\alpha_1\alpha_\ell$ twice such that $x_1 \sim \alpha_1$ and $y_1 \sim \alpha_\ell$, and we perform an analogous double subdivision if $\ell = 0$. If $\ell = k = 1$, then each edge between $\alpha_1$ and $\beta_1$ is subdivided. We similarly subdivide $\gamma_1\delta_1$ and $\gamma_k\delta_m$, adding vertices $x_2$ and $y_2$ respectively (with the same exceptional cases as above). We now add edges $\beta_i\gamma_i$ for $i = 1,\ldots,k$ as well as the edges $x_1x_2$ and $y_1y_2$. By \Lem{cyclejoin}, the graph obtained after adding these edges is still planar and has facial cycle $C' = x_1,x_2,\delta_1, \ldots, \delta_m, y_2,y_1,\alpha_\ell, \ldots, \alpha_1$. Thus by \Lem{facialcycle} we can add a vertex $z$ adjacent to every vertex of $C'$ other than $x_1,x_2,y_1,y_2$ while remaining planar. Next we delete the edges $\beta_i\beta_{i+1}$ and $\gamma_i\gamma_{i+1}$ for $i = 1,\ldots,k-1$ as well as $x_1\beta_1$, $\beta_k y_1$, $x_2\gamma_1$, and $\gamma_k y_2$. In the case where $k = 0$ we instead delete the edges $x_iy_i$ for $i = 1,2$. For each $i = 1,\ldots, k$, we now have paths $b_i,\beta_i,\gamma_i,c_i$ where the internal vertices (i.e., $\beta_i$ and $\gamma_i$) have degree two. Thus we can unsubdivide these paths, replacing them with edges $b_ic_i$ for each $i = 1,\ldots, k$. Now we can contract the edges $b_ic_i$ for each $i$. At this point, if we removed the cycle $C'$ and the vertex $z$, we would have the graph $H$. To finish, we note that $x_1,x_2$ and $y_1,y_2$ have degree two, and thus we can unsubdivide the respective paths they lie on, replacing them with edges $\alpha_1\delta_1$ and $\delta_m\alpha_\ell$. Our, still planar, graph now contains the cycle $C = \alpha_1, \ldots, \alpha_\ell, \delta_m, \ldots, \delta_1$ and a vertex $z$ adjacent to every vertex of $C$. Furthermore, the vertex $\alpha_i$, respectively $\delta_j$, is adjacent to the vertex of $H$ that $a_i$, respectively $d_j$, became after contracting the edges $b_ic_i$. Of course, this graph is exactly $H^\odot$, and thus we have that $\H = \H_1 \circ \H_2 \in \P(\ell,m)$.
\end{proof}

\begin{figure}[h!]
\begin{subfigure}{\textwidth}
  \centering
  \includegraphics[scale=1.2]{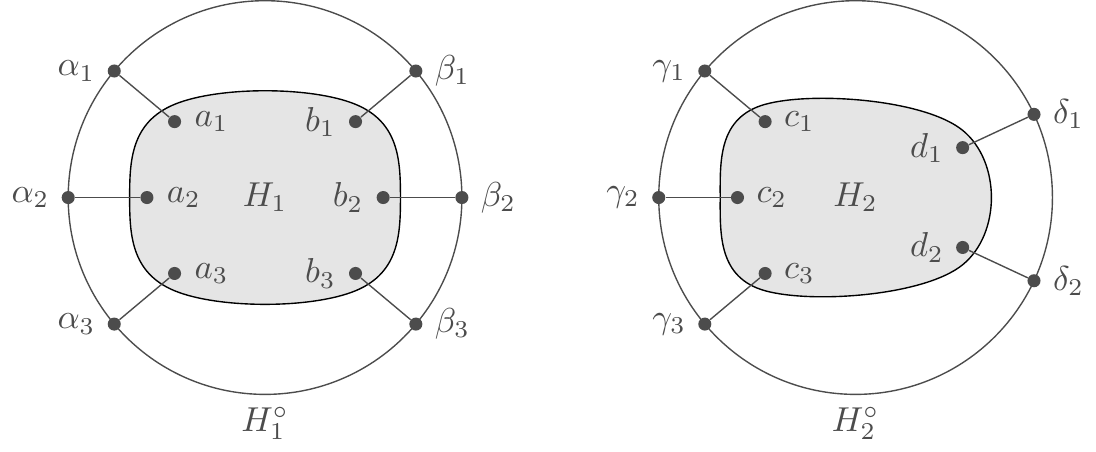}
  \caption{\phantom{a}$H_1^\circ \cup H_2^\circ$}
  \label{fig:compproof1}
\end{subfigure}

\vspace{.4in}

\begin{subfigure}{\textwidth}
  \centering
  \includegraphics[scale=1.2]{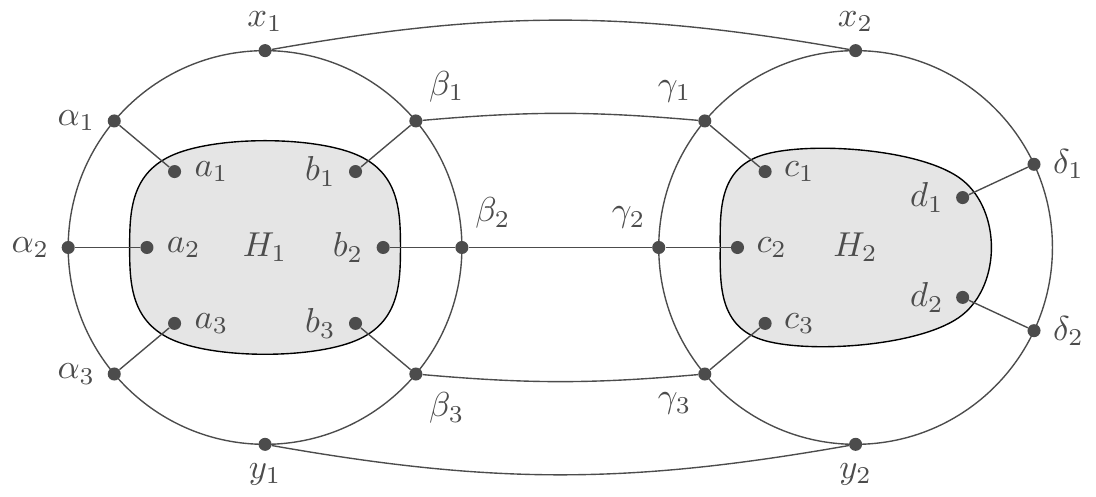}
  \caption{\phantom{a} Subdivide with $x_1,x_2,y_1,y_2$ and add edges $\beta_i\gamma_i$ for $i = 1,2,3$ and $x_1x_2$, $y_1,y_2$.}
  \label{fig:compproof2}
\end{subfigure}

\vspace{.4in}

\begin{subfigure}{\textwidth}
  \centering
  \includegraphics[scale=1.2]{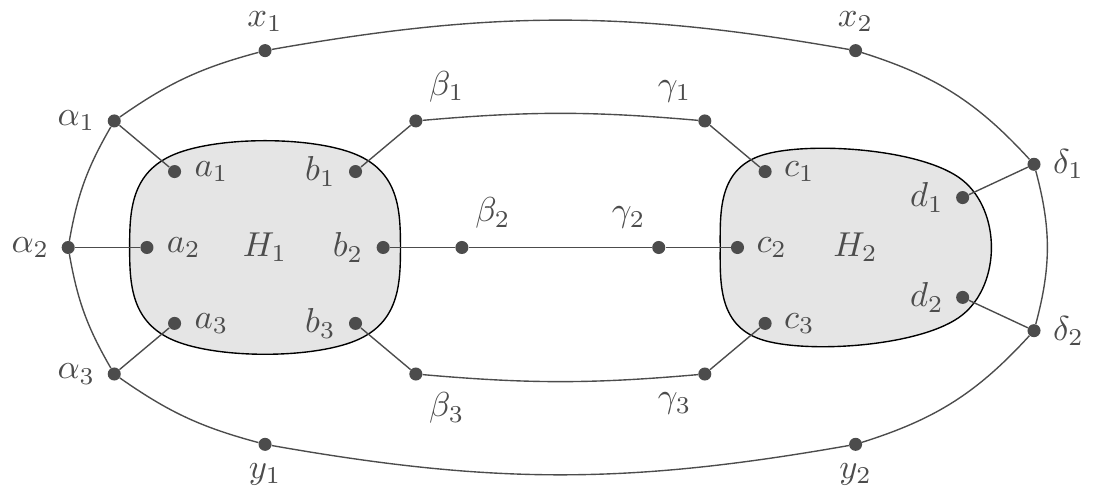}
  \caption{\phantom{a} Remove edges $x_1\beta_1$, $\beta_3y_1$, $x_2\gamma_1$, $\gamma_3y_2$, and $\beta_i\beta_{i+1}$ and $\gamma_i\gamma_{i+1}$ for $i = 1,2$.}
  \label{fig:compproof3}
\end{subfigure}
\end{figure}
\begin{figure}\ContinuedFloat
\begin{subfigure}{.5\textwidth}
  \centering
  \includegraphics[scale=1]{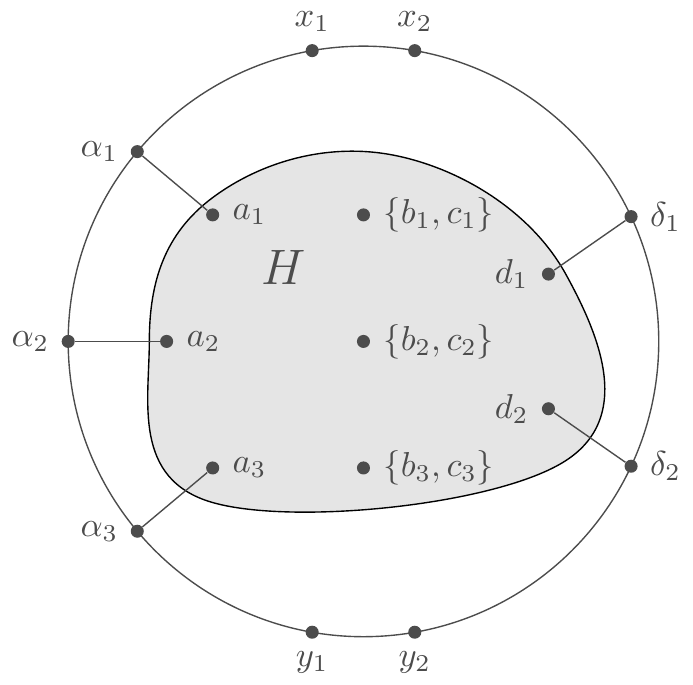}
  \caption{\phantom{a} Unsubdivide the paths $b_i,\beta_i,\gamma_i,c_i$ and then \\ \phantom{iasdf}  contract the edges $b_ic_i$ for $i = 1,2,3$.}
  \label{fig:compproof4}
\end{subfigure}
\begin{subfigure}{.5\textwidth}
\vspace{-.028in}
  \centering
  \includegraphics[scale=1]{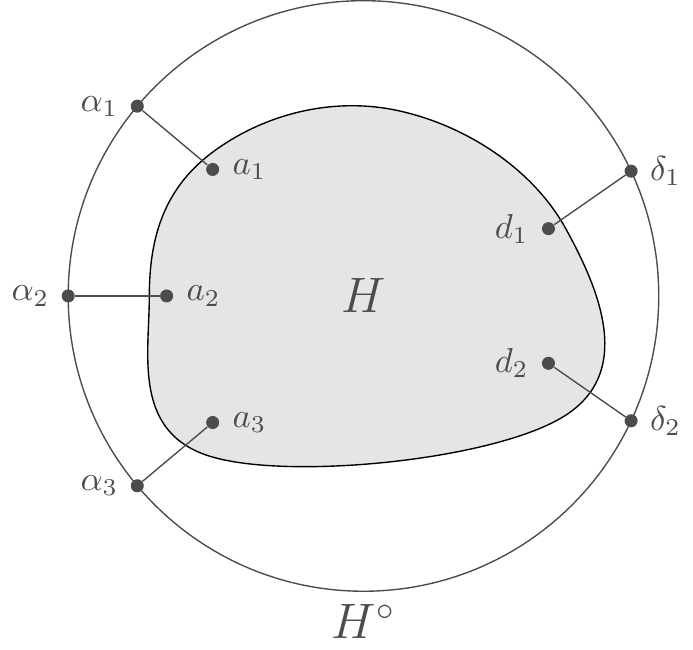}
  \caption{\phantom{a} Unsubdivide $\alpha_1,x_1,x_2,\delta_1$ and $\alpha_3,y_1,y_2,\delta_2$.}
  \label{fig:compproof5}
\end{subfigure}
\caption{Illustration of the proof that $\P$ is closed under composition.}\label{fig:compproof}
\end{figure}

Next we prove that $\P$ is closed under tensor products.

\begin{lemma}\label{lem:tensorclosed}
If $\mathbf{H}_1 \in \P(\ell,k)$ and $\mathbf{H}_2 \in \P(s,t)$, then $\mathbf{H}_1 \otimes \mathbf{H}_2 \in \P(k+s,\ell+t)$.
\end{lemma}
\begin{proof}
Let $\H_1 = (H_1,\vec{a},\vec{b}) \in \P(\ell,k)$ and $\H_2 = (H_2, \vec{c},\vec{d})  \in \P(s,t)$, and let $H = H_1 \cup H_2$ be the disjoint union of $H_1$ and $H_2$. Then $\H := \H_1 \otimes \H_2 = (H,\vec{a}\vec{c},\vec{b}\vec{d})$. We will show that $\H \in \P(\ell+s,k+t)$ by showing that $H^\circ$ has a planar embedding in which its enveloping cycle is the boundary of the outer face. This is illustrated in \Fig{tensorproof}.

In the exceptional case $\ell = k = 0$, we can take a planar embedding of $H_2^\odot$ and then simply embed $H_1^\odot = H_1$ in one of its faces. This shows that $H^\odot$ is planar in this case, and thus $\H \in \P$. Similarly, the lemma holds if $s = t = 0$.

Otherwise, let $C_1 = \alpha_1, \ldots, \alpha_\ell, \beta_k, \ldots, \beta_1$ be the enveloping cycle of $H_1^\circ$ so that $a_i \sim \alpha_i$ and $b_j \sim \beta_j$ for all $i,j$. Similarly let $C_2 = \gamma_1,\ldots, \gamma_s, \delta_t, \ldots, \delta_1$ be the enveloping cycle of $H_2^\circ$ so that $c_i \sim \gamma_i$ and $d_j \sim \delta_j$. Define
\[x_1 = \begin{cases}\alpha_\ell & \text{if } \ell \ne 0 \\ \beta_1 & \text{o.w.}\end{cases}, \quad y_1 = \begin{cases}\beta_k & \text{if } k \ne 0 \\ \alpha_1 & \text{o.w.}\end{cases}, \quad x_2 = \begin{cases}\gamma_1 & \text{if } s \ne 0 \\ \delta_t & \text{o.w.}\end{cases}, \quad y_2 = \begin{cases}\delta_1 & \text{if } t \ne 0 \\ \gamma_s & \text{o.w.}\end{cases}.\]
Let $e_i$ be the edge of $C_i$ between $x_i$ and $y_i$ (if there are two edges, pick either). Now subdivide $e_i$ with two vertices $x'_i$ and $y'_i$ such that $x'_i \sim x_i$ and $y'_i \sim y_i$. This creates new graphs $H'_1$ and $H'_2$ with facial cycles $C'_1 = \alpha_1, \ldots, \alpha_\ell,x'_1,y'_1,\beta_k,\ldots, \beta_1$ and $C'_2 = \gamma_1, \ldots, \gamma_s,\delta_t,\ldots, \delta_1,y'_2,x'_2$ respectively. Note that $C'_i$ has length at least three for $i = 1,2$. Let $H'$ be the graph obtained from $H'_1 \cup H'_2$ by adding edges $x'_1x'_2$ and $y'_1y'_2$. By \Lem{cyclejoin}, $H'$ is planar and the cycle $C' = \alpha_1, \ldots, \alpha_\ell, x'_1,x'_2, \gamma_1, \ldots, \gamma_s, \delta_t, \ldots, \delta_1, y'_2,y'_1,\beta_k, \ldots, \beta_1$ is facial. Now delete the edges $x'_1y'_1$ and $x'_2y'_2$, which results in these four vertices having degree two. We can then unsubdivide the paths $x_1x'_1x'_2x_2$ and $y_1y'_1y'_2y_2$, replacing them with edges $x_1x_2$ and $y_1y_2$ to obtain a, still planar, graph $H''$ with facial cycle $C'' = \alpha_1, \ldots, \alpha_\ell, \gamma_1, \ldots, \gamma_s, \delta_t,\ldots,\delta_1, \beta_k,\ldots, \beta_1$. It is easy to see that $H'' = H^\circ$ and $C''$ is its enveloping cycle. Thus we have proven that $\H \in \P(\ell+s,k+t)$.

\end{proof}

\begin{figure}
\begin{subfigure}{.5\textwidth}
  \centering
  \includegraphics[scale=1]{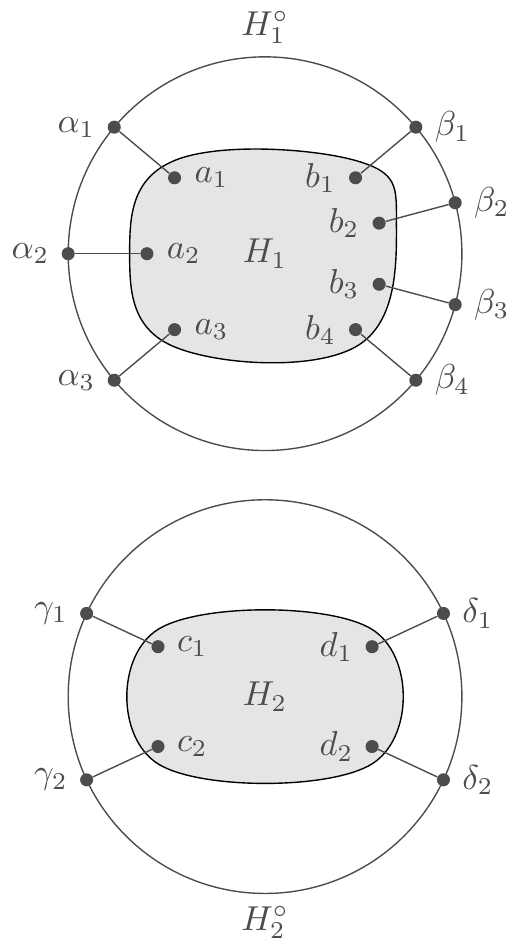}
  \caption{\phantom{a} $H_1^\circ \cup H_2^\circ$.}
  \label{fig:tensorproof1}
\end{subfigure}
\begin{subfigure}{.5\textwidth}
\vspace{-.028in}
  \centering
  \includegraphics[scale=1]{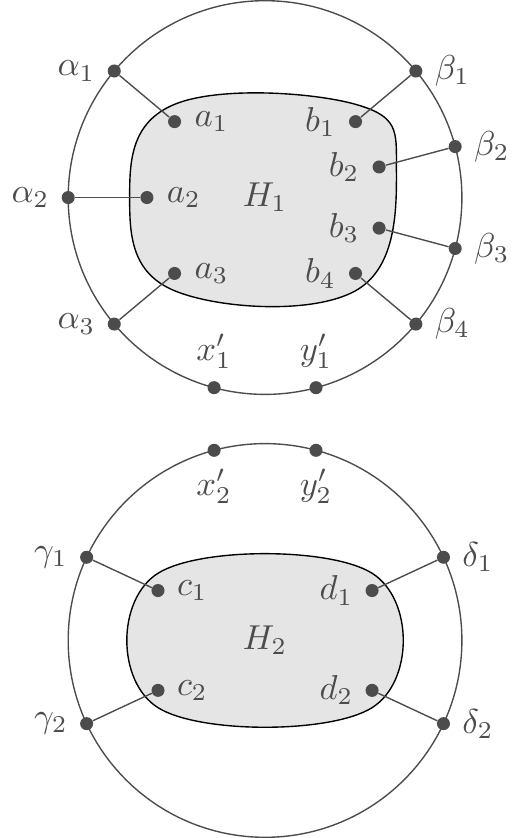}
  \caption{\phantom{a} Subdivide $\alpha_3\beta_4$ with $x'_1$, $y'_2$, and subdivide \\ \phantom{iasdi} $\gamma_1\delta_1$ with $x_2',y_2'$.}
  \label{fig:tensorproof12}
\end{subfigure}

\vspace{.3in}

\begin{subfigure}{.5\textwidth}
\vspace{-.028in}
  \centering
  \includegraphics[scale=1]{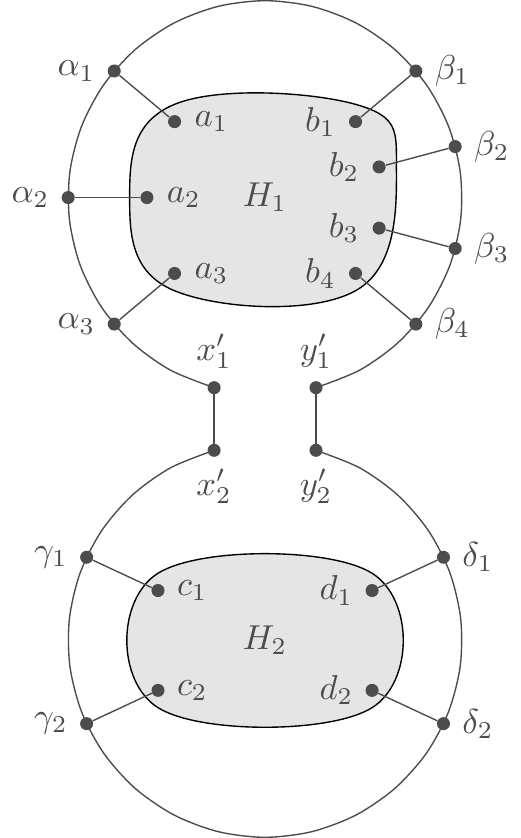}
  \caption{\phantom{a} Add edges $x'_1x'_2$ and $y'_1y'_2$, and delete edges \\ \phantom{iasdi} $x'_1y'_1$ and $x'_2y'_2$.}
  \label{fig:tensorproof2}
\end{subfigure}
\begin{subfigure}{.5\textwidth}
\vspace{-.028in}
  \centering
  \includegraphics[scale=1]{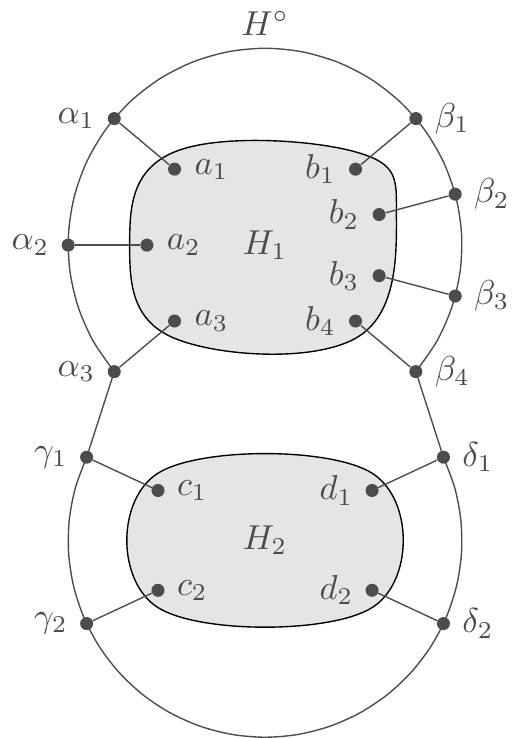}
  \caption{\phantom{a} Unsubdivide $\alpha_3,x'_1,x'_2,\gamma_1$ and $\beta_4,y'_1,y'_2,\delta_1$.}
  \label{fig:tensorproof3}
\end{subfigure}

\caption{Illustration of the proof that $\P$ is closed under tensor product.}\label{fig:tensorproof}
\end{figure}

Finally, we prove that $\P$ is closed under transposition, which is straightforward.

\begin{lemma}\label{lem:starclosed}
If $\H \in \P(\ell,k)$, then $\H^* \in \P(k,\ell)$.
\end{lemma}
\begin{proof}
Let $\H = (H,\vec{a},\vec{b})$. It is easy to see that $H^\odot(\vec{b},\vec{a})$ is isomorphic to $H^\odot(\vec{a},\vec{b})$.
\end{proof}

We define $C^G_\P(\ell,k) = \spn\{T^{\K \to G} : \K \in \P(\ell,k)\}$ and $C_\P^G = \bigcup_{\ell,k = 0}
^\infty C_\P^G(\ell,k)$. The three closure properties of $\P$ proven above along with the correspondence between matrix and bi-labeled graph operations proven in \Sec{correspondence}, and \Rem{xi} imply the following:

\begin{theorem}\label{thm:Ptensorcat}
For any graph $G$, the set $C_\P^G$ is a tensor category with duals.
\end{theorem}



\subsection{Building blocks}\label{sec:buildingblocks}

In \Sec{tale}, we will prove that $C_q^G = C_\P^G$ for any graph $G$ by showing that $\langle \vec{M}^{1,0},\vec{M}^{1,2},\vec{A} \rangle_{\circ,\otimes,*} = \P$. To prove the latter, we will need to show that we can build any element of $\P$ using some basic elements of $\langle \vec{M}^{1,0},\vec{M}^{1,2},\vec{A} \rangle_{\circ,\otimes,*}$. In this section we introduce these basic elements of $\langle \vec{M}^{1,0},\vec{M}^{1,2},\vec{A}\rangle_{\circ,\otimes,*}$.

\begin{definition}
For any integer $d \ge 0$, define $S^d$ to be the graph with vertex set $V(S^d) = \{v,v_1,\ldots,v_d\}$ and edges $\{vv_i :  i = 1, \ldots, d\}$. Thus $S^d$ is a star graph on $d+1$ vertices. Also, for any integers $m,d \ge 0$ define $\vec{S}^{m,d} = (S^d,(v^m),(v_1,\ldots, v_d))$, $\vec{S}_L^{m,d} = (S^d,(v^m),(v,v_1,\ldots, v_d))$, and $\vec{S}_R^{m,d} = (S^d,(v^m),(v_1,\ldots, v_d,v))$. We similarly define $\mathring{\vec{S}}^{m,d}$, $\mathring{\vec{S}}_L^{m,d}$, and $\mathring{\vec{S}}_R^{m,d}$ which are the same as above except that the vertex $v$ has a loop. 
\end{definition}

\begin{remark}
Note that $\vec{S}_L^{1,0} = \vec{S}_R^{1,0} = \vec{I}$, and $\vec{S}^{1,1} = \vec{A}$.

\end{remark}

\begin{lemma}\label{lem:starmaps}
For any integers $m,d \ge 0$, we have the following identities:
\begin{align*}
&\vec{S}^{m,d} = \vec{M}^{m,d} \circ \vec{A}^{\otimes d}, \quad
\vec{S}_R^{m,d} = \vec{M}^{m,d+1} \circ (\vec{A}^{\otimes d} \otimes \vec{I}), \quad
\vec{S}_L^{m,d} = \vec{M}^{m,d+1} \circ (\vec{I} \otimes \vec{A}^{\otimes d}), \\
&\mathring{\vec{S}}^{m,d} = \mathring{\vec{M}}^{m,d} \circ \vec{A}^{\otimes d}, \quad
\mathring{\vec{S}}_R^{m,d} = \mathring{\vec{M}}^{m,d+1} \circ (\vec{A}^{\otimes d} \otimes \vec{I}), \quad
\mathring{\vec{S}}_L^{m,d} = \mathring{\vec{M}}^{m,d+1} \circ (\vec{I} \otimes \vec{A}^{\otimes d}).
\end{align*}
In particular this implies that $\vec{S}^{m,d}, \vec{S}_L^{m,d}, \vec{S}_R^{m,d}, \mathring{\vec{S}}^{m,d}, \mathring{\vec{S}}_L^{m,d}, \mathring{\vec{S}}_R^{m,d} \in \langle \vec{M}^{1,0},\vec{M}^{1,2},\vec{A} \rangle_{\circ,\otimes,*}$.
\end{lemma}
\begin{proof}
We first show that $\vec{S}^{m,d} = \vec{M}^{m,d} \circ \vec{A}^{\otimes d}$. We have that $\vec{M}^{m,d} = (K_1,  (v^m),(v^d))$ and $\vec{A}^{\otimes d} = (dK_2, (u_1, \ldots, u_d),(v_1,\ldots,v_d))$ where $v$ is the only vertex of $K_1$ and $u_i,v_i$ are the vertices in the $i^\text{th}$ copy of $K_2$. Thus in $\vec{M}^{m,d} \circ \vec{A}^{\otimes d}$ all of the $u_i$'s are identified with the vertex $v$ and we obtain $(S^d,(v^m),(v_1, \ldots, v_d))$ as desired. 
Similarly, we can show that $\vec{S}_R^{m,d} = \vec{M}^{m,d+1} \circ (\vec{A}^{\otimes d} \otimes \vec{I})$ and $\vec{S}_L^{m,d} = \vec{M}^{m,d+1} \circ (\vec{I} \otimes \vec{A}^{\otimes d})$, and the analogous equalities for the looped varieties. Since $\vec{M}^{\ell,k}, \mathring{\vec{M}}^{\ell,k} \in \langle \vec{M}^{1,0},\vec{M}^{1,2}, \vec{A}\rangle_{\circ,\otimes,*}$ for all $\ell,k \ge 0$ by \Lem{Mgraphs} and \Lem{Mloopgraphs}, it follows that $\vec{S}^{m,d}, \vec{S}_L^{m,d}, \vec{S}_R^{m,d}, \mathring{\vec{S}}^{m,d}, \mathring{\vec{S}}_L^{m,d}, \mathring{\vec{S}}_R^{m,d} \in \langle \vec{M}^{1,0},\vec{M}^{1,2},\vec{A} \rangle_{\circ,\otimes,*}$ as desired.

\end{proof}




\section{A Tale of Two Tensor Categories}\label{sec:tale}

At the end of \Sec{closureprops}, we saw that $C_\P^G = \spn\{T^{\K \to G} : \K \in \P\}$ forms a tensor category with duals. In this section we will prove that $C_q^G = C_\P^G$ for any graph $G$. To do this, we will use \Thm{reformulation} which states that $C_q^G = \spn\{T^{\K \to G} : \K \in \langle \vec{M}^{1,0},\vec{M}^{1,2},\vec{A}\rangle_{\circ,\otimes,*}\}$. Thus, in order to prove $C_q^G = C_\P^G$, we will prove that $\P = \langle \vec{M}^{1,0},\vec{M}^{1,2},\vec{A}\rangle_{\circ,\otimes,*}$. One direction is immediate from the results we have already proven:

\begin{lemma}\label{lem:qinP}
We have that $\langle \vec{M}^{1,0},\vec{M}^{1,2},\vec{A}\rangle_{\circ,\otimes,*} \subseteq \P$. As a consequence, we have that $C_q^G \subseteq C_\P^G$ for any graph $G$.
\end{lemma}
\begin{proof}
In \Ex{MinP}, we saw that $\vec{M}^{\ell,k} \in \P$ for all $\ell,k$, thus in particular $\vec{M}^{1,2},\vec{M}^{1,0} \in \P$. In \Ex{AinP} we saw that $\vec{A} \in \P$. Since, as proved in \Sec{closureprops}, $\P$ is closed under composition, tensor product, and transposition, the claim immediately follows.
\end{proof}

To prove the other containment, we will show that any bi-labeled graph $\vec{K}  = (K,\vec{a},\vec{b}) \in {\P}$ with $|V(K)| \ge 2$ can be written in terms of $\vec{I}$, $\vec{S}^{m,d}, \vec{S}_L^{m,d}, \vec{S}_R^{m,d}, \mathring{\vec{S}}^{m,d}, \mathring{\vec{S}}_L^{m,d}, \mathring{\vec{S}}_R^{m,d}$, and a $\vec{K'} = (K',\vec{a'},\vec{b'}) \in \P$ with $|V(K')| = |V(K)| - 1$. This will allow us to use induction to prove that $\P \subseteq \langle \vec{M}^{1,0},\vec{M}^{1,2},\vec{A}\rangle_{\circ,\otimes,*}$. First we will need to prove that the input and output vectors of a bi-labeled graph in $\P$ must satisfy certain conditions. For this, we introduce the partition associated to a bi-labeled graph:

\begin{definition}\label{def:graphpartition}
Given a set $V$ and a tuple $\vec{c} = (c_1, \ldots, c_n) \in V^n$ of elements of $V$, we define $\mathbb{P}_\vec{c}$ as the partition of $[n]$ with parts $P_v = \{i \in [n] : c_i = v\}$ for each $v \in V$. Note that we allow empty parts in our partitions, which occurs precisely when there is some element of $V$ not occurring in $\vec{c}$.

Let $\vec{K}  = (K,\vec{a},\vec{b}) \in \G(\ell,k)$ be a bi-labeled graph, and let $\vec{c} = (c_1, \ldots, c_{\ell + k}) = (a_1, \ldots, a_\ell,b_k,\ldots, b_1)$. We define the \emph{partition associated to $\vec{K}$}, denoted $\mathbb{P}_{\vec{K}}$, as the partition $\mathbb{P}_\vec{c}$ (with underlying set $V(K)$).
\end{definition}

We remark that the partition $\mathbb{P}_\vec{c}$ depends also on the set $V$, not just on $\vec{c}$. But $V$ only affects the number of empty parts of $\mathbb{P}_\vec{c}$, and for the most part it will be clear what $V$ is from context.


Recall that a partition $\mathbb{P}$ of $[n]$ is \emph{non-crossing} if whenever $a < b < c < d$, and $a,c$ are in the same part  and $b,d$ are in the same part, then the two parts coincide.

\begin{lemma}\label{lem:noncross}
For any bi-labeled graph $\K \in \P$, the partition $\mathbb{P}_\K$ is non-crossing.
\end{lemma}
\begin{proof}
Let $\K = (K,\vec{a},\vec{b}) \in \P(\ell,k)$ and let $(c_1, \ldots, c_{\ell+k}) = (a_1, \ldots, a_\ell,b_k,\ldots,b_1)$. Suppose that the partition $\mathbb{P}_\K$ is not non-crossing. Then there exist $i_1 < i_2 < i_3 < i_4 \in [\ell+k]$ and $u \ne v \in V(K)$ such that $c_{i_1} = c_{i_3} = u$ and $c_{i_2} = c_{i_4} = v$. Let $C = \alpha_1, \ldots, \alpha_\ell, \beta_k, \ldots, \beta_1$ be the enveloping cycle of $K^\circ$, and let $(\gamma_1, \ldots, \gamma_{\ell + k}) = (\alpha_1, \ldots, \alpha_\ell, \beta_k, \ldots, \beta_1)$. The graph $K^\circ$ contains the paths
\begin{align*}
P_1 &= \gamma_{i_1}, \gamma_{i_1 + 1}, \ldots, \gamma_{i_2} \\
P_2 &= \gamma_{i_2}, \gamma_{i_2 + 1}, \ldots, \gamma_{i_3} \\
P_3 &= \gamma_{i_3}, \gamma_{i_3 + 1}, \ldots, \gamma_{i_4} \\
P_4 &= \gamma_{i_4}, \gamma_{i_4 + 1}, \ldots, \gamma_{i_1} \\
P_5 &= \gamma_{i_1}, u,\gamma_{i_3} \\
P_6 &= \gamma_{i_2}, v,\gamma_{i_4}
\end{align*}
which are all internally disjoint (i.e., the only vertices they share are their end vertices). Thus $K^\circ$ contains a $K_4$ subdivision where $\gamma_{i_1}, \gamma_{i_2}, \gamma_{i_3}$, and $\gamma_{i_4}$ correspond to the four vertices of the $K_4$. Since all four of these vertices are on the enveloping cycle $C$, the graph $K^\odot$ contains a $K_5$ subdivision and is therefore not planar, a contradiction.
\end{proof}

Using the above we can show that for any $\K \in \P$ there is always a vertex of $K$ whose neighbors on the enveloping cycle form a consecutive sequence. Formally, we say that an element $v$ \emph{occurs consecutively} in a tuple $\vec{c} = (c_1, \ldots, c_n)$ if there exists $i \in [n]$ and $r \in \{0, \ldots, n-1\}$ such that $c_{i+j} = v$ for all $j \in \{0,\ldots,r\}$ and $c_{i+j} \ne v$ for $j \in \{r+1, \ldots, n-1\}$, where all indices are taken modulo $n$. We will actually need the following stronger lemma, which implies that there is alway some vertex whose neighbors on the enveloping cycle are consecutive. In the lemma and corollary below the indices are always to be taken modulo $n$. 

\begin{lemma}\label{lem:consecutive}

Let $V$ be a set and let $\vec{c} = (c_1, \ldots, c_n) \in V^n$ be such that $\mathbb{P}_\vec{c}$ is non-crossing. If there exists $v \in V$ and some $1 < r \le n$ such that $c_i = v = c_{i+r}$ and $c_{i+s} \ne v$ for all $s = 1, \ldots, r-1$, then there exists $u \in V$ that occurs consecutively in $(c_1, \ldots, c_{n})$ and only occurs among $c_{i+1}, \ldots, c_{i+r-1}$.
\end{lemma}
\begin{proof}
Let $\vec{c}$, $v$, $i$, and $r$ be as in the lemma statement. Define $U = \{c_{i+s} : 0 < s < r\}$. Since $\mathbb{P}_\vec{c}$ is non-crossing, $c_j \not\in U$ for $j \not\in \{i+1, \ldots, i+r-1\}$. Pick $u \in U$ such that $\max\{s : c_{i+s} = u\} - \min\{s : c_{i+s} = u\}$ is minimized. Suppose that $u$ does not occur consecutively in $\vec{c}$. Then there exists $u' \in U$ such that $c_{i+t} = u'$ where $\min\{s : c_{i+s} = u\} < t < \max\{s : c_{i+s} = u\}$. However, again by the non-crossing property of $\mathbb{P}_\vec{c}$, it follows that $\max\{s : c_{i+s} = u'\} < \max\{s : c_{i+s} = u\}$ and $\min\{s : c_{i+s} = u'\} > \min\{s : c_{i+s} = u\}$. Therefore,
\[\max\{s : c_{i+s} = u'\} - \min\{s : c_{i+s} = u'\} < \max\{s : c_{i+s} = u\} - \min\{s : c_{i+s} = u\},\]
a contradiction to our choice of $u$.
\end{proof}


\begin{cor}\label{cor:consecutive}
Suppose that $\vec{c} \in V^n$ for $n \ge 1$ is such that $\mathbb{P}_\vec{c}$ is a non-crossing partition. Then there exists an element $v \in V$ that occurs consecutively in $\vec{c}$. In particular, if $\mathbb{P}_\vec{c} = \mathbb{P}_\vec{K}$ for some $\vec{K} = (K,\vec{a},\vec{b}) \in \P(\ell,k)$ for $\ell + k > 0$, then there exists a vertex $v \in V(K)$ that occurs consecutively in $(a_1,\ldots,a_\ell,b_k,\ldots,b_1)$. 
\end{cor}
\begin{proof}
Consider any element $v \in V$ that occurs in $\vec{c}$. Either every entry of $\vec{c}$ is equal to $v$ and then $v$ occurs consecutively and we are done, or there are $i,r \in [\ell+k]$ with $1 < r \le \ell + k$ such that $c_i = c_{i+r} = v$  and $c_{i+s} \ne v$ for all $s = 1, \ldots, r-1$. In this case we can apply \Lem{consecutive} and we are done.

\end{proof}

We are now able to prove the main lemma which serves as the inductive step of our proof that $\P \subseteq \langle\vec{M}^{1,0}, \vec{M}^{1,2}, \vec{A}\rangle_{\circ,\otimes,*}$.  
The idea is to take a bi-labeled graph $\vec{K} = (K,\vec{a},\vec{b}) \in \P$ and carefully choose a vertex $v \in V(K)$ that we can ``pluck" out of $\vec{K}$ in a way that lets us write $\K$ in terms of our basic building blocks from \Lem{starmaps} and a smaller $\vec{K'} \in \P$. As it may be of future use outside the proof of \Thm{onegraphcat}, we prove it as its own lemma.

\begin{lemma}\label{lem:plucking}
Suppose that $\vec{K} = (K,\vec{a},\vec{b}) \in \P$ with $|V(K)| \ge 2$. Then there exists $\vec{K'} = (K',\vec{a'},\vec{b'}) \in \P$ with $|V(K')| = |V(K)|-1$, and $m,d,r,t \ge 0$ such that one of the following hold:
\begin{align}
\vec{K} &= (\vec{I}^{\otimes r} \otimes \vec{S}^{m,d} \otimes \vec{I}^{\otimes t}) \circ \vec{K'} \label{decomp1}\\
\vec{K} &= (\vec{I}^{\otimes r} \otimes \mathring{\vec{S}}^{m,d} \otimes \vec{I}^{\otimes t}) \circ \vec{K'} \label{decomp1loop}\\
\vec{K} &= \vec{K'} \circ (\vec{I}^{\otimes r} \otimes \vec{S}^{m,d} \otimes \vec{I}^{\otimes t})^* \label{decomp2}\\
\vec{K} &= \vec{K'} \circ (\vec{I}^{\otimes r} \otimes \mathring{\vec{S}}^{m,d} \otimes \vec{I}^{\otimes t})^* \label{decomp2loop}\\
\vec{K} &= (\vec{S}_L^{m,d} \otimes \vec{I}^{\otimes t}) \circ (\vec{M}^{1,r} \otimes \vec{K'}) \label{decomp3}\\
\vec{K} &= (\mathring{\vec{S}}_L^{m,d} \otimes \vec{I}^{\otimes t}) \circ (\vec{M}^{1,r} \otimes \vec{K'}) \label{decomp3loop}\\
\vec{K} &= (\vec{I}^{\otimes r} \otimes \vec{S}_R^{m,d}) \circ (\vec{K'} \otimes \vec{M}^{1,r}) \label{decomp4}\\
\vec{K} &= (\vec{I}^{\otimes r} \otimes \mathring{\vec{S}}_R^{m,d}) \circ (\vec{K'} \otimes \vec{M}^{1,r}) \label{decomp4loop}
\end{align}
\end{lemma}
\begin{proof}
The bi-labeled graph $\vec{K'}$ will be constructed from $\vec{K}$ by removing a vertex $v$ from the underlying graph $K$ of $\vec{K}$. Depending on whether or not the vertex $v$ has a loop, we will either be among cases \eqref{decomp1loop}, \eqref{decomp2loop}, \eqref{decomp3loop}, \eqref{decomp4loop}, or among cases \eqref{decomp1}, \eqref{decomp2}, \eqref{decomp3}, \eqref{decomp4}. There is essentially no difference in the proofs of the looped versus loopless cases, and so we will just prove the loopless cases.

We have that there exist $\ell,k \ge 0$ such that $\vec{K} \in \P(\ell,k)$ and thus $\vec{a} = (a_1,\ldots,a_\ell)$, $\vec{b} = (b_1, \ldots, b_k)$. We first consider the case $\ell +k \ge 1$ (the $\ell + k = 0$ case will be easier). By \Cor{consecutive}, we have that there exists $v \in V(K)$ that occurs consecutively in $(a_1,\ldots,a_\ell, b_k, \ldots, b_1)$.

\vspace{.15cm}

\noindent\textbf{Case 1:} \textit{$v$ occurs only among $\vec{a}$}. 

If $k = 0$, then $v$ may occur consecutively by appearing in precisely the entries $a_1, \ldots, a_p$ and $a_q, \ldots, a_k$ for some $q > p +1$. However, if this is the case we may apply \Lem{consecutive} to obtain a $v' \in V(K)$ that occurs consecutively in $(a_1, \ldots, a_\ell)$ and only occurs among $a_{p+1}, \ldots, a_{q-1}$.

So we may assume that there are $1 \le p \le q \le \ell$ such that $a_p, a_{p+1}, \ldots, a_{q}$ are precisely the occurrences of $v$ among $\vec{a}$. Let $m = q-p+1$ (i.e., $m$ is the number of occurences of $v$) and let $v_1, \ldots, v_d$ be the neighbors of $v$ in $K$. Define $K'$ as the graph obtained from $K$ by deleting $v$, let $\vec{a'} = (a_1, \ldots, a_{p-1},v_1,\ldots,v_d,a_{q+1},\ldots,a_\ell)$, and define $\vec{K'} = (K',\vec{a'},\vec{b})$. Note that the order that $v_1, \ldots, v_d$ appear in $\vec{a'}$ matters for whether or not $\vec{K'} \in \P$. But as the indices have been assigned arbitrarily, we have not yet actually specified an order. The necessary ordering will emerge from our proof that $\vec{K'} \in \P$. First we will show that \eqref{decomp1} holds, which will be independent of this ordering.

Recall that $\vec{S}^{m,d} = (S^d,(v^m),(v_1, \ldots,v_d))$, where we are purposely using the names $v$ and $v_1,\ldots,v_d$ to indicate that these correspond to the same vertices from $K$. Let $r = p-1$ and $t = \ell - q$. The bi-labeled graph $\vec{I}^{\otimes r} \otimes \vec{S}^{m,d} \otimes \vec{I}^{\otimes t}$ has vertices $u_1, \ldots, u_{p-1},v,v_1,\ldots,v_d,u_{q+1},\ldots,u_\ell$ and edges $vv_1, \ldots, vv_d$. Its output vector is $(u_1, \ldots, u_{p-1},v^m,u_{q+1}, \ldots, u_\ell)$ and its input vector is $(u_1, \ldots, u_{p-1},v_1,\ldots,v_d,u_{q+1},\ldots,u_\ell)$. Multiplying $\vec{I}^{\otimes r} \otimes \vec{S}^{m,d} \otimes \vec{I}^{\otimes t}$ with $\vec{K'}$ identifies $u_i$ with $a_i$ for $i = 1,\ldots, p-1, q+1, \ldots,\ell$, and identifies $v_i$ from $S^d$ with $v_i$ from $\vec{K'}$. After this identification, the neighborhood of $a_i$ for $i= 1,\ldots, p-1, q+1, \ldots,\ell$ is the same as its neighborhood in $K'$ (unless $a_i$ is among $v_1, \ldots, v_d$) as the $u_i$ were all isolated vertices. However, the vertices $v_1, \ldots, v_d$ each gain $v$ as an additional neighbor. In other words, the underlying graph obtained from this multiplication is $K$. It is straightforward to see that the output and input vectors after multiplication are equal to $\vec{a}$ and $\vec{b}$ respectively, and thus we have proven the equation in~\eqref{decomp1}.

Now we must show that $\vec{K'} \in \P$. Consider the planar graph $K^\odot$, letting $z$ be the vertex adjacent to all of the vertices on the enveloping cycle $\alpha_1,\ldots,\alpha_\ell, \beta_k, \ldots, \beta_1$. We will construct $K'^\odot$ from this graph through operations that preserve planarity. We illustrate this in \Fig{indproofplanar}, though there we show $K'^\circ$ for aesthetic reasons. First, contract the edges $v\alpha_i$ for $i = p,p+1, \ldots,q$ and the edges $\alpha_j \alpha_{j+1}$ for $j = p,p+1,\ldots,q-1$ to obtain a new graph $K''$. Let the vertex resulting from this contraction be called $v$, and note that $N(v) = \{v_1, \ldots, v_d\} \cup \{z,\alpha_{p-1},\alpha_{q+1}\}$, where we let $\alpha_0 = \beta_1$ and $\alpha_{\ell+1} = \beta_k$. Note that $\alpha_{p-1}$ and $\alpha_{q+1}$ may not exist as we have defined them, for instance if $p = 1$, $q = \ell$, and $k = 0$. However, in this case we can subdivide the two edges of the enveloping cycle incident to $\alpha_p$ and $\alpha_q$ that we did not contract so that the new vertex $v$ has two neighbors along this cycle which we refer to as $\alpha_{p-1}$ and $\alpha_{q+1}$. This will not affect the proof of $\vec{K'} \in \P$, we only need to remember to remove these two vertices by unsubdividing at the end. As edge contraction preserves planarity, the resulting graph $K''$ is planar. Note that $K''$ contains the cycle $C = \alpha_1,\ldots,\alpha_{p-1},v,\alpha_{q+1},\ldots, \alpha_\ell,\beta_k, \ldots, \beta_1$, all of whose vertices are adjacent to $z$. Thus by \Lem{facialcycle}, $K'' \setminus z$ has a planar embedding in which $C$ is the boundary of the outer face. Since $C$ bounds a face in this embedding, the vertices $\alpha_{p-1}$ and $\alpha_{q+1}$ appear consecutively in the cyclic ordering of the neighbors of $v$. Thus, by possibly reindexing, the neighbors of $v$ in $K'' \setminus z$ appear in cyclic order $\alpha_{p-1},v_1, \ldots, v_d,\alpha_{q+1}$ in this embedding (this is what determines the order that $v_1,\ldots,v_d$ appear in $\vec{a'}$). Again using \Lem{facialcycle}, we can obtain a planar embedding of $K''$ by adding $z$ to the outer face of our embedding of $K'' \setminus z$ and adding edges between $z$ and each vertex of $C$. In this embedding of $K''$ the neighbors of $v$ appear in cyclic order $\alpha_{p-1},v_1, \ldots, v_d, \alpha_{q+1},z$. Let us now subdivide the edges incident to $v$. For the edge $vv_i$, the subdividing vertex will be called $u_i$. The vertices subdividing $v\alpha_{p-1}$, $v\alpha_{q+1}$, and $vz$ will be called $w_1$, $w_2$, and $y$ respectively. By \Lem{addcycle}, we may add the cycle $C' = w_1,u_1, \ldots, u_d,w_2,y$ to $K''$ and obtain a graph that is still planar. Now let us contract the edges $vy$ and $yz$, referring to the resulting vertex as $z$. The vertices $w_1$ and $w_2$ now each have three neighbors: $\{\alpha_{p-1}, z, u_1\}$, and $\{\alpha_{q+1},z,u_d\}$ respectively (unless $d = 0$ in which case $u_1$ and $u_d$ are replaced by $w_2$ and $w_1$ respectively). Remove the edges $zw_1, zw_2$, and unsubdivide the paths $\alpha_{p-1},w_1,u_1$ and $u_d,w_2,\alpha_{q+1}$ (or the path $\alpha_{p-1},w_1,w_2,\alpha_{q+1}$ if $d = 0$). The resulting planar graph has cycle $C'' = \alpha_1, \ldots, \alpha_{p-1}, u_1, \ldots, u_d, \alpha_{q+1}, \ldots, \alpha_{\ell},\beta_k, \ldots, \beta_1$, all of whose vertices are adjacent to $z$. Moreover, all vertices of $C''$ have degree four, $\alpha_i \sim a_i$ for $i = 1, \ldots, p-1,q+1,\ldots, \ell$, $\beta_i \sim b_i$ for $i = 1, \ldots, k$, and $u_i \sim v_i$ for $i = 1, \ldots, d$. In other words, the resulting graph is $K'^\odot$, and it is planar by construction. Thus we have proven that $\vec{K'} \in \P$.


\vspace{.15cm}

\noindent\textbf{Case 2:} \textit{$v$ occurs only among $\vec{b}$}.

In this case $\vec{K}^*$ falls into Case 1, and thus we can write $\vec{K}^* = (\vec{I}^{\otimes r} \otimes \vec{S}^{m,d} \otimes \vec{I}^{\otimes t}) \circ \vec{K'}$, for some $\vec{K'} \in \P$ and $|V(K')| = |V(K)|-1$. Taking $^*$ of both sides, we obtain $\vec{K} = \vec{K'}^* \circ (\vec{I}^{\otimes r} \otimes \vec{S}^{m,d} \otimes \vec{I}^{\otimes t})^*$, i.e., we have proven that \eqref{decomp2} holds in this case. Moreover, we have that $\vec{K'}^* \in \P$ by \Lem{starclosed}. 

\vspace{.15cm}

\noindent\textbf{Case 3:} \textit{$a_1 = b_1 = v$}.

In this case $v$ occurs in both $\vec{a}$ and $\vec{b}$, so we cannot merely multiply by some $\vec{K'}$ as such a multiplication cannot result in $v$ occuring in both the input and output vectors. However, we can accommodate this with only a slight modification. First we must handle an exceptional subcase.

Suppose that $a_\ell = b_k = v$, but that not every entry of both $\vec{a}$ and $\vec{b}$ is equal to $v$. Then there is either an $i \in [\ell]$ such that $a_i \ne v$ or an $j \in [k]$ such that $b_j \ne v$. In the first case, we can apply \Lem{consecutive} on $(a_1, \ldots, a_\ell,b_k, \ldots, b_1)$ to obtain that there must be some vertex $u \in V(K)$ that occurs consecutively in $(a_1, \ldots, a_\ell, b_k,\ldots, b_1)$ and only occurs among $\vec{a}$. Thus we are back to Case 1. Similarly, if there is $b_j \ne v$ then we can use Case 2.

Now we may assume that there are $m \in [\ell]$ and $r \in [k]$ such that the entries of $\vec{a}$ and $\vec{b}$ equal to $v$ are precisely $a_1, \ldots, a_m$ and $b_1, \ldots, b_r$ respectively. 
Let $t = \ell - m \ge 0$, and let $v_1, \ldots, v_d$ be the neighbors of $v$ in $K$. We will show that~\eqref{decomp3} holds for an appropriate $\vec{K'}$. Let $K'$ be the graph obtained from $K$ by deleting $v$, and let $\vec{a'} = (v_1, \ldots, v_d,a_{m+1}, \ldots, a_\ell)$ and $\vec{b'} = (b_{r+1}, \ldots, b_k)$. Finally, define $\vec{K'} = (K',\vec{a'},\vec{b'})$. Denote the vertex of degree $d$ in $\vec{S}^{m,d}_L$ by $v$ and its neighbors by $v_1, \ldots, v_d$, and let the single vertex of $\vec{M}^{1,r}$ be denoted $w$. The underlying graph of $\vec{M}^{1,r} \otimes \vec{K'}$ is just $K'$ plus the isolated vertex $w$, and its input and output vectors are $(w^r,b_{r+1}, \ldots, b_k)$ and $(w,v_1, \ldots, v_d,a_{m+1}, \ldots, a_\ell)$ respectively. The underlying graph of $\vec{S}^{m,d}_L \otimes \vec{I}^{\otimes t}$ is $S^d$ plus $t$ isolated vertices $u_1, \ldots, u_t$, and its input and output vectors are $(v,v_1, \ldots, v_d, u_1, \ldots, u_t)$ and $(v^m,u_1, \ldots, u_t)$ respectively. Thus the multiplication $(\vec{S}_L^{m,d} \otimes \vec{I}^{\otimes t}) \circ (\vec{M}^{1,r} \otimes \vec{K'})$ identifies $w$ with $v$ (let us call the resulting vertex $v$), identifies $v_i$ from $S^d$ with $v_i$ from $K'$, and identifies the isolated vertex $u_i$ with $a_{m+i}$ (we will call the resulting vertex $a_{m+i}$). Thus the vertices of this product coming from $\vec{K'}$ have the same neighborhoods as they did in $K'$, with the exception of $v_1, \ldots, v_d$ which have all gained $v$ as a neighbor. In other words, the underlying graph of this product is precisely $K$. Moreover, the input and output vectors are $(v^r,b_{r+1},\ldots, b_k) = (b_1, \ldots, b_k)$ and $(v^m,a_{m+1}.\ldots,a_\ell) = (a_1, \ldots, a_\ell)$. Thus we have proven~\eqref{decomp3}.

To prove that $\vec{K'} \in \P$, we will reduce to Case 1 above. Consider $\hat{\vec{K}} = \vec{K} \circ (\vec{M}^{r,0} \otimes \vec{I}^{\otimes k-r}) \in \P$. This multiplication simply removes the occurences of $v$ from the input vector, and thus $\hat{\vec{K}}$ falls into Case 1. Moreover, the $\hat{\vec{K}}'$ constructed for $\hat{\vec{K}}$ in Case 1 is precisely the $\vec{K'}$ constructed here for $\vec{K}$. It follows that $\vec{K'} \in \P$.

\vspace{.15cm}

\noindent\textbf{Case 4:} \textit{$a_k = b_\ell = v$}.

This case is symmetric to Case 3, it is simply the mirror image.

\vspace{.15cm}

\noindent\textbf{The $\ell + k = 0$ case}.

In this case there is no enveloping cycle. However, we can simply choose any vertex $v$ of $K$, letting $v_1, \ldots, v_d$ be its neighbors. Let $K'$ be the graph obtained from $K$ by deleting $v$, and let $\vec{K'} = (K', (v_1, \ldots, v_d),\varnothing)$. It is easy to see that $\vec{K} = \vec{S}^{0,d} \circ \vec{K'}$, i.e., that~\eqref{decomp1} holds. To show that $\vec{K'} \in \P$, fix a planar embedding of $K$. Assuming the neighbors of $v$ appear in cyclic order $v_1, \ldots, v_d$, apply \Lem{addcycle} to $v$ subdividing $vv_i$ with a vertex $\alpha_i$ and creating a cycle $C = \alpha_1, \ldots, \alpha_d$. This graph is planar and it is easy to see that it is precisely $K'^\odot$, where $v$ is the vertex adjacent to every vertex of the enveloping cycle $C$.\end{proof}

\begin{figure}[h!]
\begin{subfigure}{.5\textwidth}
  \centering
  \includegraphics[scale=1.2]{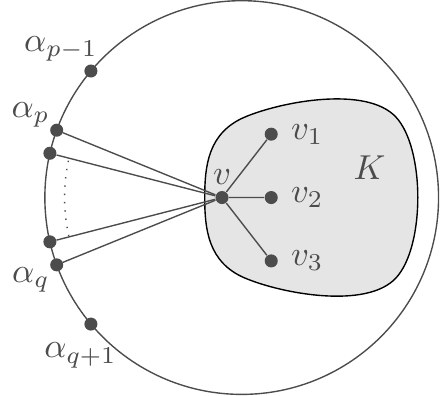}
  \caption{\phantom{a} The vertex $v$ and its neighbors in $K^\circ$. \\ \phantom{a}}
  \label{fig:indproof1}
\end{subfigure}
\begin{subfigure}{.5\textwidth}
  \centering
  \includegraphics[scale=1.2]{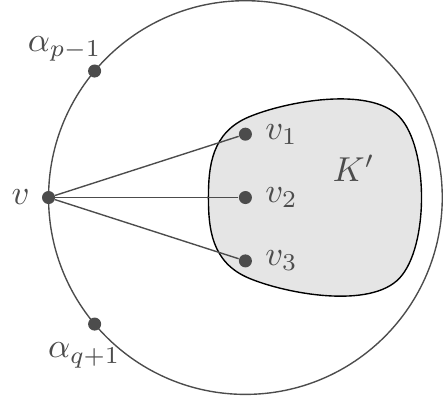}
  \vspace{.053in}
  \caption{\phantom{a} Contract edges $v\alpha_i$ for $i = p, \ldots, q$, and edges $\alpha_j \alpha_{j+1}$ for $j = p, \ldots, q-1$, to form new vertex $v$.}
  \label{fig:indproof2}
\end{subfigure}

\vspace{.4in}

\begin{subfigure}{.5\textwidth}
  \centering
  \includegraphics[scale=1.2]{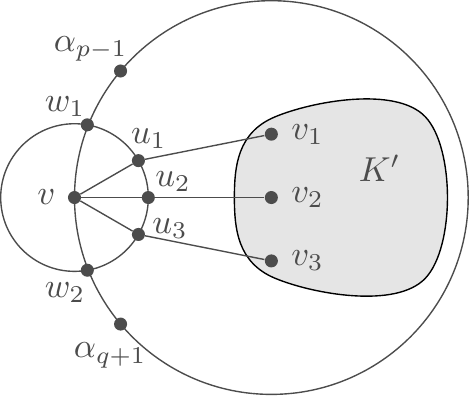}
  \vspace{.055in}
  \caption{\phantom{a} Subdivide edges incident to $v$ and add \\ \phantom{asdf} cycle through new vertices.}
  \label{fig:indproof3}
\end{subfigure}
\begin{subfigure}{.5\textwidth}
  \centering
  \includegraphics[scale=1.2]{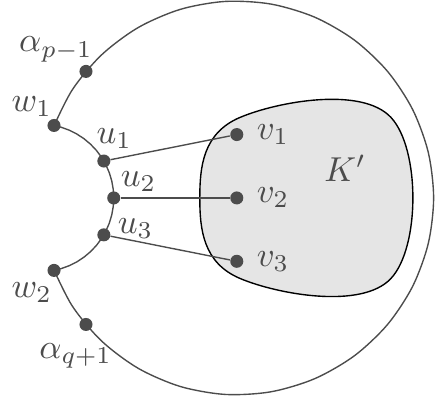}
  \caption{\phantom{a} Remove vertex $v$ and edge $w_1w_2$. \\ \phantom{iasdf}}
  \label{fig:indproof4}
\end{subfigure}

\vspace{.3in}

\begin{center}
\begin{subfigure}{.6\textwidth}
\vspace{-.028in}
  \centering
  \includegraphics[scale=1.2]{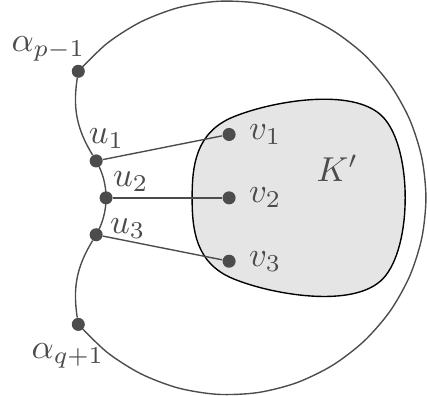}
  \caption{\phantom{a} Unsubdivide $\alpha_{p-1},w_1,u_1$ and $u_3,w_2,\alpha_{q+1}$ to obtain $K'^\circ$.}
  \label{fig:indproof5}
\end{subfigure}
\end{center}
\caption{Illustration of the fact that $\K' \in \P$, from the proof of \Lem{plucking}.}\label{fig:indproofplanar}
\end{figure}

Now we are able to prove the main results of this section:

\begin{theorem}\label{thm:onegraphcat}
$\P = \langle\vec{M}^{1,0}, \vec{M}^{1,2}, \vec{A}\rangle_{\circ,\otimes,*}$.
\end{theorem}
\begin{proof}
We have already shown that $\langle\vec{M}^{1,0}, \vec{M}^{1,2}, \vec{A}\rangle_{\circ,\otimes,*} \subseteq \P$ in \Lem{qinP}, so we will prove the reverse containment. 
We proceed by induction on the number of vertices of the underlying graph of $\vec{K}$.

If $\vec{K} = (K,\vec{a},\vec{b})$ where $|V(K)| = 1$, then $\vec{K} = (K,(v^\ell),(v^k))$ where $v$ is the lone vertex of $K$. Thus $\vec{K} = \vec{M}^{\ell,k}$ or $\vec{K} = \mathring{\vec{M}}^{\ell,k}$ for some $\ell,k \ge 0$ depending on whether $v$ has a loop or not. In either case, $\vec{K} \in \langle\vec{M}^{1,0}, \vec{M}^{1,2}, \vec{A}\rangle_{\circ,\otimes,*}$ by \Lem{Mgraphs} and \Lem{Mloopgraphs}.

Now suppose that $\vec{K} = (K,\vec{a},\vec{b}) \in \P$ where $|V(K)| \ge 2$. By \Lem{plucking}, one of Equations~\eqref{decomp1}--\eqref{decomp4loop} hold for some $\vec{K'} = (K', \vec{a'},\vec{b'}) \in \P$ with $|V(K')| = |V(K)| - 1$. Let us suppose that~\eqref{decomp1} holds. By induction, we have that $\vec{K'} \in \langle\vec{M}^{1,0}, \vec{M}^{1,2}, \vec{A}\rangle_{\circ,\otimes,*}$, and by \Lem{Mgraphs} and \Lem{starmaps}, we have that $\vec{I}, \vec{S}^{m,d} \in \langle\vec{M}^{1,0}, \vec{M}^{1,2}, \vec{A}\rangle_{\circ,\otimes,*}$. It thus follows by definition that
\[\vec{K} = (\vec{I}^{\otimes r} \otimes \vec{S}^{m,d} \otimes \vec{I} ^{\otimes t}) \circ \vec{K'} \in \langle\vec{M}^{1,0}, \vec{M}^{1,2}, \vec{A}\rangle_{\circ,\otimes,*}.\]
The proofs of the cases where~\eqref{decomp1loop}--\eqref{decomp4loop} hold are similar.
\end{proof}

We can now prove our characterization of the intertwiners of the quantum automorphism group of a graph:

\begin{theorem}\label{thm:onecat}
For any graph $G$, we have that $C_q^G = C_\P^G$.
\end{theorem}
\begin{proof}
This follows immediately from the definition of $C_\P^G$, and \Thm{reformulation} and \Thm{onegraphcat}.
\end{proof}

\begin{remark}\label{rem:onecat}
We remark that the above is equivalent to $C_q^G(\ell,k) = C_\P^G(\ell,k)$ for all $\ell,k \ge 0$.
\end{remark}

The above two results are the culmination of the work presented in this article up to this point. \Thm{onecat} gives our advertised combinatorial characetrization of the intertwiners of the quantum automorphism group of a graph: they are the (span of the) homomorphism matrices of planar bi-labeled graphs. Thus we are able to describe $C_q^G$ on the level of spanning sets for the vector spaces $C_q^G(\ell,k)$, as opposed to in terms of a generating set for the tensor category as given by Chassaniol's result (see \Thm{chassaniol}). This is analogous to the difference between describing the intertwiners of $S_n^+$ as being generated by $M^{1,0}$ and $M^{1,2}$, and describing them as the span of the maps associated to the non-crossing partitions.

\subsection{Further results}\label{sec:further}

In this section we prove a few more results about $\P$ and $C_q^G$ that we will need later on.

\begin{lemma}\label{lem:schurclosed}
For $\ell + k \le 2$ we have that $\P(\ell,k)$ is closed under Schur product.
\end{lemma}
\begin{proof}
This follows immediately from \Lem{buildingschur}, the fact that $\vec{M}^{\ell,k} \in \P$ for all $\ell,k \ge 0$, and the fact that $\P$ is closed under composition, tensor products, and transposition.
\end{proof}

Combining \Lem{schurclosed} and \Thm{onecat} shows that $C_q^G(\ell,k)$ is closed under entrywise product for $\ell + k \le 2$. However, this was already known~\cite{qperms}.

The next lemma shows that $C_q^G(1,0)$ is spanned by the $G$-homomorphism matrices of \emph{connected} bi-labeled graphs in $\P(1,0)$. We will use $\hom(K,G)$ to denote the number of homomorphisms from $K$ to $G$. Moreover, we will use $\hom((K,a_1,\ldots, a_\ell),(G,u_1, \ldots, u_\ell))$ to denoted the number of homomorphisms $\varphi$ from $K$ to $G$ such that $\varphi(a_i) = u_i$ for all $i$.

\begin{lemma}\label{lem:P01connected}
For any graph $G$,
\[C_q^G(1,0) = \spn\left\{T^{\vec{K} \to G} : \vec{K} = (K,(a),\varnothing) \in \P(1,0), \ K \text{ connected}\right\}.\]
\end{lemma}
\begin{proof}
Suppose that $\vec{K} = (K,(a),\varnothing) \in \P(1,0)$ such that $K$ is not connected. Let $K'$ be the connected component containing $a$ and let $K''$ be the graph obtained from $K$ by removing $K'$. Also let $\vec{K'} = (K',(a),\varnothing)$. It is straightforward to see that
\[(T^{\vec{K} \to G}) = \hom(K'',G)T^{\vec{K'} \to G}.\]
\end{proof}

\begin{cor}\label{cor:01qorbits}
Let $G$ be a graph. Vertices $u,v \in V(G)$ are in the same orbit of $\qut(G)$ if and only if
\[\hom((K,a),(G,u)) = \hom((K,a),(G,v))\]
for all connected planar graphs $K$ and $a \in V(K)$.
\end{cor}
\begin{proof}
Recall from \Rem{orbitsasintertwiners} that $u$ and $v$ are in the same orbit of $\qut(G)$ if and only if $T_u = T_v$ for all $T \in C_q^G(1,0)$. By \Lem{P01connected}, this means that $u$ and $v$ are in the same orbit of $\qut(G)$ if and only if $\hom((K,a),(G,u)) = \hom((K,a),(G,v))$ for all $(K,(a),\varnothing) \in \P(1,0)$ such that $K$ is connected. Recalling from \Rem{Psmalllk} that $(K,(a),\varnothing) \in \P(1,0)$ if and only if $K$ is planar yields the corollary statement.
\end{proof}

In \Def{graphpartition} we saw how to associate a partition to any bi-labeled graph. Now we will see how to associate a bi-labeled graph to any partition.

\begin{definition}\label{def:partitiongraph}
Let $\mathbb{P} = \{P_1, \ldots, P_m\}$ be a partition of $[n]$ (allowing empty parts), and suppose $\ell,k \ge 0$ are integers such that $\ell + k = n$. Define $\vec{c}_\mathbb{P} = (c_1, \ldots, c_n)$ to be the tuple such that $c_i = j$ where $i \in P_j$. We further define the bi-labeled graph $\vec{K}_\mathbb{P}^{\ell,k} = (K, \vec{a},\vec{b}) \in \G(\ell,k)$ where $K$ is the edgeless graph on vertex set $[m]$ and $(a_1, \ldots, a_\ell, b_k \ldots, b_1) = (c_1, \ldots, c_n)$.
\end{definition}

\begin{remark}\label{rem:partitioncorrespondence}
It is straightforward to see that $\mathbb{P}_{\vec{c}_\mathbb{P}} = \mathbb{P}$ and $\mathbb{P}_{\vec{K}_\mathbb{P}^{\ell,k}} = \mathbb{P}$ for any partition $\mathbb{P}$. Also note that $(\vec{K}_{\mathbb{P}}^{\ell,k})^\circ$ and $(\vec{K}_{\mathbb{P}}^{\ell,k})^\odot$ do not depend on $\ell,k$.
\end{remark}

%

In \Lem{noncross}, we showed that the partition associated to a bi-labeled graph in $\P$ is non-crossing. Here we will show somewhat of a converse:

\begin{lemma}\label{lem:noncross2}
Let $\mathbb{P} = \{P_1, \ldots, P_m\}$ be a partition of $[n]$ (allowing empty parts), and suppose $\ell,k \ge 0$ are integers such that $\ell + k = n$. Then $\mathbb{P}$ is non-crossing if and only if $\vec{K}_\mathbb{P}^{\ell,k} \in \P$.
\end{lemma}
\begin{proof}
For notational simplicity let $\K = (K,\vec{a},\vec{b}) = \vec{K}_\mathbb{P}^{\ell,k}$. Recall that $K$ is the edgeless graph on $m$ vertices. If we have $\K \in \P$, then $\mathbb{P}$ is non-crossing by \Rem{partitioncorrespondence} and \Lem{noncross}. Conversely, suppose that $\mathbb{P}$ is non-crossing and let $\vec{c} = \vec{c}_\mathbb{P} \in [m]^n$. We proceed by induction on $m$.

If $m = 1$, then $\vec{K} = (K,(v^\ell),(v^k))$ where $v$ is the only vertex of $K$. Thus $\vec{K} = \vec{M}^{\ell,k} \in \P$. Now suppose that $m \ge 2$. Since $\mathbb{P}$ is non-crossing, by \Cor{consecutive} we have that there is an $v \in [m]$ that occurs consecutively in $\vec{c} = (a_1, \ldots, a_\ell,b_k, \ldots, b_a)$. Consider $K^\circ$ with enveloping cycle $C = \alpha_1, \ldots, \alpha_\ell, \beta_k, \ldots, \beta_1$, and let $(\gamma_1, \ldots, \gamma_n) = (\alpha_1, \ldots, \alpha_\ell, \beta_k, \ldots, \beta_1)$. The neighbors of $v$ on the enveloping cycle are $\gamma_i, \gamma_{i+1}, \ldots, \gamma_{i+r-1}$ for $r = |P_v|$. After deleting $v$, its neighbors on the enveloping cycle have degree two, thus let $K'$ be the graph obtained from $K^\circ$ by deleting $v$ and unsubdividing the path $\gamma_{i-1}, \gamma_i, \ldots, \gamma_{i+r}$. It is not difficult to see that $K'$ is equal to $H^\circ$ where $\vec{H} = \vec{K}_{\mathbb{P}'}^{\ell',k'}$ and $\mathbb{P}'$ is the partition of $[n] \setminus P_v$ obtained from $\mathbb{P}$ by removing $P_v$. It is straightforward to see that removing a part of a non-crossing partition cannot result in a partition that is not non-crossing. Thus by induction we have that $\vec{H} \in \P$. Consider the planar embedding of $H^\circ$ such that its enveloping cycle is bounding the outer face. Recall from the construction of $K'$ that $\gamma_{i-1}$ and $\gamma_{i+r}$ are adjacent in the enveloping cycle of $H^\circ$. We can subdivide the edge between $\gamma_{i-1}$ and $\gamma_{i+r}$ with vertices $\gamma_{i}, \ldots, \gamma_{i+r-1}$, adding these to our embedding. The vertices $\gamma_{i}, \ldots, \gamma_{i+r-1}$ will thus all be incident to two common faces: the outer and another face we will call $F$. By \Lem{commonface}, we have that we can embed the new vertex $v$ in $F$ and connect it to $\gamma_{i}, \ldots, \gamma_{i+r-1}$ while remaining planar. Thus we have constructed a planar embedding of $K^\circ$ in which its enveloping cycle is bounding the outer face and so $\vec{K} \in \P$.
\end{proof}


\section{Characterization of Quantum Isomorphism}\label{sec:qiso}

In this section we will use the characterization of $C_q^G$ from \Thm{onecat} to prove that two graphs are quantum isomorphic if and only if they admit the same number of homomorphisms from any planar graph. We being by giving a name to the latter condition:

\begin{definition}\label{def:planariso}
We say that two graphs $G$ and $H$ are \emph{planar isomorphic}, denoted $G \cong_\P H$, if $\hom(K,G) = \hom(K,H)$ for all planar graphs $K$.
\end{definition}

In general, if $\mathcal{F}$ is a family of graphs, we will say that $G$ and $H$ are $\mathcal{F}$-isomorphic, denoted $G \cong_\mathcal{F} H$, if $\hom(K,G) \cong \hom(K,H)$ for all $K \in \mathcal{F}$. We remark that there are some well-known relations on graphs that can be phrased in this manner. For example, graphs $G$ and $H$ are \emph{cospectral} (i.e., their adjacency matrices have the same multiset of eigenvalues) if and only if $G \cong_\mathcal{F} H$ for $\mathcal{F}$ being the class of all cycles~\cite{VH03}. As mentioned in the abstract, Lov\'{a}sz~\cite{lovasz1967operations} showed that isomorphism is equivalent to $\mathcal{F}$-isomorphism for $\mathcal{F}$ being the class of all graphs. Also Dell, Grohe, and Rattan~\cite{DGR} proved that for $\mathcal{F}$ being the class of graphs of treewidth at most $k$, the relation $G \cong_{\mathcal{F}} H$ is equivalent to the graphs $G$ and $H$ not being distinguished by the $k$-dimensional Weisfeiler-Leman algorithm.

In the process of proving that quantum isomorphism and planar isomorphism are equivalent, we will show that both of these are equivalent to a third relation based on the intertwiner spaces of the quantum automorphism groups of two graphs:

\begin{definition}
Given graphs $G$ and $H$, we define a \emph{weak isomorphism} of $C_q^G$ and $C_q^H$ to be a bijection $\Phi\colon C_q^G \to G_q^H$ such that
\begin{enumerate}
\item for all $\ell,k \ge 0$, the restriction of $\Phi$ to $C_q^G(\ell,k)$ is a bijective linear map to $C_q^H(\ell,k)$;
\item if $T,T' \in C_q^G$, then $\Phi(T \otimes T') = \Phi(T) \otimes \Phi(T')$;
\item if $T \in C_q(\ell,k)$ and $T' \in C_q^G(k,m)$, then $\Phi(TT') = \Phi(T)\Phi(T')$;
\item if $T \in C_q^G$, then $\Phi(T^*) = \Phi(T)^*$;
\item $\Phi(M^{1,2}) = M^{1,2}$ and $\Phi(M^{1,0}) = M^{1,0}$.\footnote{Here we should specify that the $M^{1,2}$ and $M^{1,0}$ appearing in the argument of $\Phi$ are those living in $C_q^G$ whereas the $M^{1,2}$ and $M^{1,0}$  on the righthand sides are those living in $C _q^H$. However we will see that the existence of a weak isomorphism implies that $|V(G)| = |V(H)|$ and thus these maps can be identified with each other.}
\end{enumerate}
\end{definition}

\begin{remark}
Our definition of weak isomorphism of $C_q^G$ and $C_q^H$ above seems to be closely related to the \emph{monoidal equivalence} of $\qut(G)$ and $\qut(H)$ considered in~\cite{bigalois}. However, the approach taken, and language used, in that work is of a much more category theoretic flavor than ours.
\end{remark}

The main result of this section will be that the following three statements are equivalent:
\begin{enumerate}
\item[(1)] Graphs $G$ and $H$ are quantum isomorphic.
\item[(2)] There exists a weak isomorphism $\Phi \colon C_q^G \to C_q^H$ such that $\Phi(A_G) = A_H$.
\item[(3)] Graphs $G$ and $H$ are planar isomorphic.
\end{enumerate}

Note the importance of the condition $\Phi(A_G) = A_H$. Indeed, for any graph $G$ there is a weak isomorphism between $C_q^G$ and $C_q^{\overline{G}}$ (as they are equal), but $G$ and $\overline{G}$ are not typically quantum isomorphic.

The proof of the equivalence of these statements will follow the structure $(1) \Rightarrow (2) \Rightarrow (3) \Rightarrow (1)$, each implication receiving its own subsection.

\subsection{Quantum isomorphism implies weak isomorphism}

Recall from \Thm{qisoMU} that graphs $G$ and $H$ are quantum isomorphic if and only if there exists a magic unitary $\mathcal{U}$ such that $\mathcal{U}A_G = A_H\mathcal{U}$. We will use this magic unitary to construct the weak isomorphism from $C_q^G$ to $C_q^H$. The proof will make use of the tensor powers $\mathcal{U}^{\otimes \ell}$ of $\mathcal{U}$. Recall that this is a $V(H)^\ell \times V(G)^\ell$ operator-valued matrix whose $i_1\ldots i_\ell,j_1\ldots j_\ell$-entry is $u_{i_1j_1}\ldots u_{i_\ell j_\ell}$. We will use the fact that this matrix is unitary in the sense that $(\mathcal{U}^{\otimes \ell})^*\mathcal{U}^{\otimes \ell} = \mathcal{U}^{\otimes \ell}(\mathcal{U}^{\otimes \ell})^* = I^{\otimes \ell} \otimes \vec{1}$, where $I$ is the usual $n \times n$ identity (for $n = |V(G)| = |V(G)|$) and $\vec{1}$ is the identity in the $C^*$-algebra where the entries of $\mathcal{U}$ live. 
In the following proof we will make extensive use of this unitarity and the operations on operator-valued matrices discussed in \Sec{outline}.

\begin{lemma}\label{lem:qiso2weakiso}
Let $G$ and $H$ be graphs. If $G \cong_{qc} H$ then there exists a weak isomorphism $\Phi\colon C_q^G \to C_q^H$ such that $\Phi(A_G) = A_H$. 
\end{lemma}
\begin{proof}
By assumption there exists a magic unitary $\mathcal{U}$ such that $\mathcal{U}A_G = A_H\mathcal{U}$. We will show that the map $\Phi :C_q^G \to C_q^H$ given by $\mathcal{U}^{\otimes \ell}T = \Phi(T)\mathcal{U}^{\otimes k}$ for $T \in C_q^G(\ell,k)$ for all $\ell,k \ge 0$ is a weak isomorphism.

First we must show that this is well-defined, i.e., that $T\mathcal{U}^{\otimes k} = T'\mathcal{U}^{\otimes k}$ implies that $T = T'$. As $\mathcal{U}^{\otimes k}$ is unitary, we have that $T\mathcal{U}^{\otimes k} = T'\mathcal{U}^{\otimes k}$ implies that $T(I^{\otimes k} \otimes \vec{1}) = T\mathcal{U}^{\otimes k}(\mathcal{U}^{\otimes k})^* = T'\mathcal{U}^{\otimes k}(\mathcal{U}^{\otimes k})^* = T'(I^{\otimes k} \otimes \vec{1})$. Of course the entries of $T(I^{\otimes k} \otimes \vec{1})$ are just the entries of $T$ times $\vec{1}$, and similarly for $T'(I^{\otimes k} \otimes \vec{1})$. Thus we have that $T = T'$ as desired.

We still need to show that $\Phi$ is defined on every element of $C_q^G$, since not every $n^\ell \times n^k$ matrix $T$ satisfies $\mathcal{U}^{\otimes \ell}T = T'\mathcal{U}^{\otimes k}$ for some matrix $T'$. By assumption we have that $\mathcal{U}A_G = A_H\mathcal{U}$, and it follows from the properties of magic unitaries that $\mathcal{U}M^{1,0} = M^{1,0}\mathcal{U}^{\otimes 0}$ and $\mathcal{U}M^{1,2} = M^{1,2}\mathcal{U}^{\otimes 2}$. Since $C_q^G = \langle M^{1,0}, M^{1,2}, A_G\rangle_{+, \circ, \otimes, *}$, to show that $\Phi$ is defined everywhere on $C_q^G$, it suffices to show that $\Phi$ commutes with the necessary operations.

Suppose that $\mathcal{U}^{\otimes \ell}T_i = T'_i\mathcal{U}^{\otimes k}$ and $\alpha_i \in \mathbb{C}$ for $i =1, \ldots, m$. Then it is straightforward to see that
\[\mathcal{U}^{\otimes \ell}\left(\sum_{i=1}^m \alpha_i T_i\right) = \left(\sum_{i=1}^m \alpha_i T'_i\right)\mathcal{U}^{\otimes k}.\]

Next suppose that $T_1,T'_1 \in C_n(\ell,k)$, $T_2,T'_2 \in C_n(k,m)$, and $\mathcal{U}^{\otimes \ell}T_1 = T'_1\mathcal{U}^{\otimes k}$ and $\mathcal{U}^{\otimes k}T_2 = T'_2\mathcal{U}^{\otimes m}$. Then $\mathcal{U}^{\otimes \ell}(T_1T_2) = T'_1\mathcal{U}^{\otimes k}T_2 = (T'_1T'_2)\mathcal{U}^{\otimes m}$. Thus $\Phi$ commutes with matrix multiplication.

Suppose that $T_1,T'_1 \in C_n(\ell,k)$, $T_2,T'_2 \in C_n(r,s)$, and $\mathcal{U}^{\otimes \ell}T_1 = T'_1\mathcal{U}^{\otimes k}$ and $\mathcal{U}^{\otimes r}T_2 = T'_2\mathcal{U}^{\otimes s}$. Then
\[\mathcal{U}^{\otimes (\ell + r)}(T_1 \otimes T_2) = (\mathcal{U}^{\otimes \ell}\otimes \mathcal{U}^{\otimes r})(T_1 \otimes T_2) = (\mathcal{U}^{\otimes \ell}T_1) \otimes (\mathcal{U}^{\otimes r}T_2) = (T'_1\mathcal{U}^{\otimes k}) \otimes (T'_2\mathcal{U}^{\otimes s}) = (T'_1 \otimes T'_2)\mathcal{U}^{\otimes (k+s)}.\]

Finally, suppose that $T,T' \in C(\ell,k)$ and $\mathcal{U}^{\otimes \ell}T = T'\mathcal{U}^{\otimes k}$. Then
\[\mathcal{U}^{\otimes k}T^* = \mathcal{U}^{\otimes k}[T^*(\mathcal{U}^{\otimes \ell})^*]\mathcal{U}^{\otimes \ell} = \mathcal{U}^{\otimes k}[\mathcal{U}^{\otimes \ell}T]^*\mathcal{U}^{\otimes \ell} = \mathcal{U}^{\otimes k}[T'\mathcal{U}^{\otimes k}]^*\mathcal{U}^{\otimes \ell} = \mathcal{U}^{\otimes k}(\mathcal{U}^{\otimes k})^*T'^*\mathcal{U}^{\otimes \ell} = T'^*\mathcal{U}^{\otimes \ell}.\]
Thus we have shown that $\Phi$ commutes with all necessary operations which, along with the fact that $\Phi$ is defined on $M^{1,0}, M^{1,2}$, and $A_G$, implies that $\Phi$ is defined everywhere on $C_q^G$. Since $C_q^H = \langle M^{1,0}, M^{1,2}, A_H\rangle_{+, \circ, \otimes, *}$ and we have already seen that $M^{1,0},M^{1,2}$,and $A_H$ are in the image of $\Phi$, it follows that $\Phi$ is surjective. The injectivity of $\Phi$ follows from the unitarity of $\mathcal{U}^{\otimes \ell}$ via an argument identical to the one showing that $\Phi$ is well-defined. This is the last of the properties required of a weak isomorphism and thus we have completed the proof.
\end{proof}

\subsection{Weak isomorphism implies planar isomorphism}

In order to prove that the existence of a weak isomorphism $\Phi \colon C_q^G \to C_q^H$ such that $\Phi(A_G) = A_H$ implies that $G$ and $H$ are planar isomorphic, we will need to prove two things: that any such weak isomorphism maps $T^{\K \to G}$ to $T^{\K \to H}$ for any $\K \in \P$, and that any such $\Phi$ is sum-preserving. We begin with the former:

\begin{lemma}\label{lem:weakiso}
Let $G$ and $H$ be graphs which admit a weak isomorphism $\Phi\colon C_q^G \to C_q^H$ such that $\Phi(A_G) = A_H$. Then $\Phi(T^{\vec{K} \to G}) = T^{\vec{K} \to H}$ for all $\vec{K} \in \P$.
\end{lemma}
\begin{proof}

Since $\Phi(M^{1,0}) = M^{1,0}$, $\Phi(M^{1,2}) = M^{1,2}$, and $\Phi(A_G) = A_H$, we have that $\Phi(T^{\K \to G}) = T^{\K \to H}$ for $\K = \vec{M}^{1,0}, \vec{M}^{1,2}, \vec{A}$. Since by \Thm{onegraphcat} $\P = \langle \vec{M}^{1,0}, \vec{M}^{1,2}, \vec{A}\rangle_{\circ,\otimes,*}$, any $\K \in \P$ can be written as an expression in $\vec{M}^{1,0}, \vec{M}^{1,2}$, and $\vec{A}$ using the operations of composition, tensor product, and transposition. Let $\mathrm{op}(\K)$ denote the minimum number of such operations required to express $\K$. We will induct on $\mathrm{op}(\K)$. Our bases cases are $\vec{M}^{1,0}, \vec{M}^{1,2}$, and $\vec{A}$ where $\mathrm{op}(\K) = 0$ and which we have already dealt with above. Now suppose that $\mathrm{op}(\K) > 0$. Consider an expression for $\K$ using $\mathrm{op}(\K)$ operations and consider the last operation performed which is one of composition, tensor product, or transposition. If the last operation was composition, then $\K = \K' \circ \K''$ where $\mathrm{op}(\K'), \mathrm{op}(\K'') < \mathrm{op}(\K)$. Thus by induction we have that $\Phi(T^{\K' \to G}) = T^{\K' \to H}$ and $\Phi(T^{\K'' \to G}) = T^{\K'' \to H}$. By \Lem{compcorr} and the fact that $\Phi$ commutes with composition, we have that
\begin{align*}
\Phi(T^{\K \to G}) &= \Phi(T^{\K' \circ \K'' \to G}) = \Phi(T^{\K' \to G}T^{\K'' \to G}) = \Phi(T^{\K' \to G})\Phi(T^{\K'' \to G}) \\
&= T^{\K' \to H}T^{\K'' \to H} = T^{\K' \circ \K'' \to H} = T^{\K \to H}.
\end{align*}
The cases where the last operation used to express $\K$ were tensor product or transposition are essentially identical. Thus by induction we have proven the lemma.
\end{proof}

Next we will need to show that a weak isomorphism $\Phi\colon C_q^G \to C_q^H$ is \emph{sum-preserving}, meaning that the sum of the entries of $T$ and of $\Phi(T)$ are the same. We denote the sum of the entries of a matrix $T$ as $\text{sum}(T)$. First we prove the following lemma, where we temporarily revert to using the notation $U$ for the unit map instead of $M^{1,0}$ to avoid too many superscripts.

\begin{lemma}\label{lem:sumfromunit}
Let $n \ge 1$ and $\ell,k \ge 0$ be integers, and let $U\colon \mathbb{C} \to \mathbb{C}^n$ be the unit map. If $T \in C_n(\ell,k)$, then $(U^{\otimes \ell})^*TU^{\otimes k}$ is the $1 \times 1$ matrix with entry $\text{sum}(T)$.
\end{lemma}
\begin{proof}
As $U(1) = \sum_{i=1}^n e_i$, the matrix representation of $U$ is the $n \times 1$ all ones matrix. Thus $U^{\otimes k}$ is the $n^k \times 1$ all ones matrix and $(U^{\otimes \ell})^*$ is the $1 \times n^\ell$ all ones matrix. Simple matrix multiplication shows that $(U^{\otimes \ell})^*TU^{\otimes k}$ is then the $1 \times 1$ matrix whose only entry is $\text{sum}(T)$.
\end{proof}

\begin{lemma}\label{lem:identity}
Let $\Phi\colon C_q^G \to C_q^H$ be a weak isomorphism. Then $\Phi(I^{\otimes r}) = I^{\otimes r}$ for all integer $r \ge 0$. Furthermore, $|V(G)| = |V(H)|$.
\end{lemma}
\begin{proof}
By definition, we have that $\Phi(M^{1,2}) = M^{1,2}$ and $\Phi(M^{2,1}) = \Phi((M^{1,2})^*) = \Phi((M^{1,2}))^* = (M^{1,2})^* = M^{2,1}$. Since $M^{1,2}M^{2,1} = I$, it follows that
\[\Phi(I) = \Phi(M^{1,2}M^{2,1}) = \Phi(M^{1,2})\Phi(M^{2,1}) = M^{1,2}M^{2,1} = I.\]
From this and the fact that $\Phi$ commutes with tensor product, we have that $\Phi(I^{\otimes r}) = I^{\otimes r}$ for all $r \ge 1$. For $r = 0$, we will use the fact that the unit map $U \in C_q^G$ satisfies $U^*U = |V(G)|I^{\otimes 0}$ and similarly for the unit map of $C_q^H$. But here is the one place where we must be careful to distinguish these two unit maps. So we use $U_G, I_G \in C_q^G$ and $U_H, I_H \in C_q^H$ to denote the respective unit and identity maps. Letting $n = |V(G)|$ and $n' = |V(H)|$, we have that
\[n\Phi(I_G^{\otimes 0}) = \Phi(nI_G^{\otimes 0}) = \Phi(U_G^*U_G) = \Phi(U_G)^*\Phi(U_G) = U_H^*U_H = n'I_H^{\otimes 0}.\]
Thus $\Phi(I_G^{\otimes 0}) = \frac{n'}{n}I_H^{\otimes 0}$. However, we then have that
\[\Phi(I_G^{\otimes 0}) = \Phi(I_G^{\otimes 0} I_G^{\otimes 0}) = \Phi(I_G^{\otimes 0})\Phi(I_G^{\otimes 0}) = \frac{n'^2}{n^2}I_H^{\otimes 0}.\]
It follows that $\frac{n'}{n} = 1$ and so $\Phi(I_G^{\otimes 0}) = I_H^{\otimes 0}$ and $|V(G)| = |V(H)|$ as desired.
\end{proof}

We remark that the above implies that the restriction of $\Phi$ to $C_q^G(0,0)$ is the identity map.

\begin{lemma}\label{lem:equalsums}
Suppose that $\Phi\colon C_q^G \to C_q^H$ is a weak isomorphism. Then $\text{sum}(\Phi(T)) = \text{sum}(T)$ for all $T \in C_q^G$.
\end{lemma}
\begin{proof}
Let $T \in C_q^G(\ell,k)$ for some $\ell,k$. By \Lem{identity}, $\Phi(I^{\otimes r}) = I^{\otimes r}$ for all $r \in \mathbb{N}$ where $I$ is the $n \times n$ identity and $n = |V(G)| = |V(H)|$. In particular, $\Phi(I^{\otimes 0}) = I^{\otimes 0}$. Let $m = \text{sum}(T)$ and $m' = \text{sum}(\Phi(T))$. Then, by \Lem{sumfromunit} we have that
\[\Phi(mI^{\otimes 0}) = \Phi((U^{\otimes \ell})^*TU^{\otimes k}) = (U^{\otimes \ell})^*\Phi(T)U^{\otimes k} = m' I^{\otimes 0}.\]
As $\Phi$ is linear, this implies that $m = m'$.
\end{proof}

Note that if $\vec{K} = (K,\vec{a},\vec{b}) \in \G(k,\ell)$, then $\text{sum}(T^{\vec{K} \to G}) = \hom(K,G)$. From this observation we obtain the second implication required for our main result:

\begin{lemma}\label{lem:weakiso2piso}
Let $G$ and $H$ be graphs. If there is a weak isomorphism $\Phi\colon C_q^G \to C_q^H$ such that $\Phi(A_G) = A_H$, then $G \cong_\P H$.
\end{lemma}
\begin{proof}
Let $K$ be a planar graph. Then $\vec{K} = (K,\varnothing,\varnothing) \in \P(0,0)$. Let $T = T^{\vec{K} \to G}$ and $T' = T^{\vec{K} \to H}$. By \Thm{onecat}, we have that $T \in C_q^G$ and $T' \in C_q^H$. By \Lem{weakiso} we have that $\Phi(T) = T'$, and by \Lem{equalsums} we have that $\hom(K,G) = \text{sum}(T) = \text{sum}(T') = \hom(K,H)$.
\end{proof}

\subsection{Planar isomorphism implies quantum isomorphism}

The proof that planar isomorphism implies quantum isomorphism requires the most work of the three implications comprising our main result. To prove it we will need several lemmas that establish some properties of planar isomorphism. We first simply remark that if $G$ and $H$ are planar isomorphic, then they must have the same number of vertices since $|V(G)| = \hom(K_1,G) = \hom(K_1,H) = |V(H)|$. Only slightly more work is required to prove that it suffices to check homomorphisms from \emph{connected} planar graphs to determine planar isomorphism:

\begin{lemma}\label{lem:connectedinputs}
Graphs $G$ and $H$ are planar isomorphic if and only if
\[\hom(K,G) = \hom(K,H) \text{ for all connected planar graphs } K.\]
\end{lemma}
\begin{proof}
The forward direction is obvious. Suppose that $\hom(K,G) = \hom(K,H)$ for all connected planar graphs $K$. Let $K'$ be an arbitrary planar graph with connected components $K_1, \ldots, K_m$. Then we have that $\hom(K_i,G) = \hom(K_i,H)$ for all $i = 1, \ldots, m$ and it is easy to see that $\hom(K',G) = \prod_{i=1}^m \hom(K_i,G)$ and similarly for $H$. Thus
\[
\hom(K',G) = \prod_{i=1}^m \hom(K_i,G) = \prod_{i=1}^m \hom(K_i,G) = \hom(K',H).
\]
\end{proof}

Of course the same proof works to show that $G \cong_\mathcal{F} H$ is determined by the connected graphs in the family $\mathcal{F}$. Next we show that both quantum and planar isomorphism preserve the number of connected components.

\begin{lemma}\label{lem:connectpreserve}
If $G \cong_\P H$ then $G$ and $H$ have the same number of connected components. The same is true for $G \cong_{qc} H$.
\end{lemma}
\begin{proof}
In~\cite{DGR}, they showed that graphs $G$ and $H$ are not distinguished by the 2-dimensional Weisfeiler-Leman algorithm if and only if $\hom(K,G) = \hom(K,H)$ for all graphs $K$ with treewidth at most $2$. The graphs with treewidth at most 2 are exactly those not having $K_4$ as a minor. Thus any such graph is planar. Therefore if $G \cong_\P H$, then $G$ and $H$ are not distinguished by 2-dimensional Weisfeiler-Leman. Similarly, it was shown in~\cite{qperms} that quantum isomorphic graphs are not distinguished by 2-dimensional Weisfeiler-Leman. Finally, it is well-known that if $G$ and $H$ are not distinguished by 2-dimensional Weisfeiler-Leman, then they have the same number of connected components.
\end{proof}

The next lemma shows that graphs are planar isomorphic if and only if their (full) complements are. Recall from \Rem{complements} that the same holds for quantum isomorphic graphs.

\begin{lemma}\label{lem:complements}
For graphs $G$ and $H$, we have that $G \cong_\P H$ if and only if $\overline{G} \cong_\P \overline{H}$ if and only if $\overline{\overline{G}} \cong_\P \overline{\overline{H}}$.
\end{lemma}
\begin{proof}
We prove the first equivalence first. Suppose that $G \cong_\P H$ and let $K$ be a planar graph. If $K$ is empty then it is clear that $\hom(K,\overline{G}) = \hom(K,\overline{H})$ since both numbers are equal to $n^{|V(K)|}$ where $n = |V(G)| = |V(H)|$. Otherwise let $e \in E(K)$ be an edge between vertices $u$ and $v$. Consider the homomorphisms from $K\setminus e$ to an arbitrary graph $X$. Partition these homomorphisms as follows:
\begin{align*}
S_1 &= \{\varphi\colon K\setminus e \to X \ | \ \varphi(u) = \varphi(v) \text{ and this vertex has a loop}\}; \\
S_2 &= \{\varphi\colon K\setminus e \to X \ | \ \varphi(u) = \varphi(v) \text{ and this vertex does not have a loop}\}; \\
S_3 &= \{\varphi\colon K\setminus e \to X \ | \ \varphi(u) \sim \varphi(v) \ \& \ \varphi(u) \ne \varphi(v)\}; \\
S_4 &= \{\varphi\colon K\setminus e \to X \ | \ \varphi(u) \ \& \ \varphi(v) \text{ are distinct non-adjacent vertices}\}.
\end{align*}
Let $K'$ be the graph obtained from $K/e$ by adding a loop to the vertex $u$ and $v$ become after contracting $e$. Then we have that
\begin{align*}
&\hom(K,X) = |S_1| + |S_3|, \\
&\hom(K/e) = |S_1|+|S_2|, \\
&\hom(K',X) = |S_1|.
\end{align*}
Since $\hom(K\setminus e,X) = |S_1| + |S_2| + |S_3| + |S_4|$, it follows that
\begin{equation}\label{eq:planarcomb}
|S_1| + |S_4| = \hom(K\setminus e,X) - \hom(K,X) - \hom(K/e,X) + 2\hom(K',X).
\end{equation}
Now $|S_1| + |S_4|$ is the number of functions $\psi\colon V(K) \to V(X)$ that preserve adjacency for all edges of $K$ except $e$ and that maps $u$ and $v$ either to the same vertex which has a loop or to two distinct non-adjacent vertices. In other words, it counts $\psi\colon V(K) \to V(X)$ that act as a homomorphism from $K$ to $X$ on $E(K) \setminus \{e\}$ an as a homomorphism from $K$ to $\overline{X}$ on $e$. Since all of $K, K\setminus e, K/e$, and $K'$ are planar, it follows from \Eq{planarcomb} that $|S_1| + |S_4|$ is determined by the number of homomorphisms from planar graphs to $X$. By recursively applying this argument to the remaining edges in $E(K)$, it follows that the number of homomorphisms from $K$ to $\overline{X}$ is determined by the number of homomorphisms from planar graphs to $X$. Thus if $G \cong_\P H$, then $\overline{G} \cong_\P \overline{H}$, and the converse holds since taking complements is an involution.

The second equivalence is proven similarly using the fact that $|S_2| + |S_4| = \hom(K\setminus e,X) - \hom(K,X)$.

\end{proof}

\begin{remark}
The above proof works for $\cong_\mathcal{F}$ for any class of graphs $\mathcal{F}$ closed under minors and adding loops. The equivalence of $G \cong_\mathcal{F} H$ and $\overline{\overline{G}} \cong_\mathcal{F} \overline{\overline{H}}$ holds for any minor-closed class $\mathcal{F}$ (even if it is not closed under adding loops), as does the equivalence of $G \cong_\mathcal{F} H$ and $\overline{G} \cong_\mathcal{F} \overline{H}$ if restricted to loopless graphs $G$ and $H$.
\end{remark}

Since at least one of a graph or its complement is connected, \Lem{connectpreserve} and \Lem{complements} allow us to reduce the proof of the equivalence of $G \cong_{qc} H$ with $G \cong_\P H$ to the case where both $G$ and $H$ are connected.

\begin{lemma}\label{lem:01201}
Let $G$ and $H$ be graphs such that $G \cong_\P H$, and suppose $\vec{K}_i = (K_i,(a_i),\varnothing) \in \P(1,0)$ for $i = 1, \ldots, r$. Let $T_i = T^{\vec{K}_i \to G}$ and $T'_i = T^{\vec{K}_i \to H}$. If $T = \sum_{i=1}^r \alpha_iT_i$ and $T' = \sum_{i=1}^r \alpha_i T'_i$ for some $\alpha_i \in \mathbb{C}$,
then $T$ and $T'$ have the same multiset of entries.
\end{lemma}
\begin{proof}
As each $K_i$ is planar, we have that $\text{sum}(T_i) = \hom(K_i,G) = \hom(K_i,H) = \text{sum}(T'_i)$, and thus $\text{sum}(T) = \text{sum}(T')$. Now consider $k^\text{th}$ schur powers $T^{\schur k}$ and $T'^{\schur k}$. Letting $T_{i_1\ldots i_k} = T_{i_1} \schur \ldots \schur T_{i_k}$, we have that
\[
T^{\schur k} = \sum_{i_1,\ldots,i_k = 1}^r \left(\prod_{t=1}^k \alpha_{i_t}\right) T_{i_1 \ldots i_k},
\]
and similarly for $T'^{\schur k}$. Let $\K_{i_1 \ldots i_k} = \K_{i_1} \schur \ldots \schur \K_{i_k}$ and let $K_{i_1 \ldots i_k}$ be the underlying graph of this bi-labeled graph. By \Lem{schurclosed} we have that $\K_{i_1 \ldots i_k} \in \P(1,0)$ and thus $K_{i_1 \ldots i_k}$ is planar. We also have that $T_{i_1 \ldots i_k} = T^{\K_{i_1 \ldots i_k} \to G}$ and $T'_{i_1 \ldots i_k} = T^{\K_{i_1 \ldots i_k} \to H}$. Since $K_{i_1 \ldots i_k}$ is planar, we have that $\text{sum}(T_{i_1 \ldots i_k}) = \hom(K_{i_1 \ldots i_k},G) = \hom(K_{i_1 \ldots i_k},H) = \text{sum}(T'_{i_1 \ldots i_k})$ and thus $\text{sum}(T^{\schur k}) = \text{sum}(T'^{\schur k})$. If $x_1, \ldots, x_{|V(G)|}$ and $y_1, \ldots, y_{|V(H)|}$ are the entries of $T$ and $T'$ respectively, then the latter equation is equivalent to
\[\sum_{i=1}^{|V(G)|} x_i^k = \sum_{i=1}^{|V(H)|} y_i^k.\]
By Newton's relations this implies that $T$ and $T'$ must have the same multiset of \emph{nonzero} entries. But since $G \cong_\P H$, we have that $|V(G)| = |V(H)|$ and thus $T$ and $T'$ must also have the same number of zero entries.
\end{proof}

The last lemma we will need shows that there is a correspondence between the orbits of the quantum automorphism groups of planar isomorphic graphs.

\begin{lemma}\label{lem:qorbits}
Let $G$ and $H$ be graphs such that $G \cong_\P H$. Suppose that $T_1,\ldots,T_r \in C_q^G(1,0)$ are the characteristic vectors of the orbits of $\qut(G)$ and that $T_i = \sum_{\ell=1}^{m_i} \alpha_{i,\ell} T^{\vec{K}_{i,\ell} \to G}$ for some $\K_{i,\ell} \in \P(1,0)$ for $\ell \in [m_i]$, for all $i \in [r]$. Then $T'_i = \sum_{\ell=1}^{m_i} \alpha_{i,\ell} T^{\vec{K}_{i,\ell} \to H} \in C_q^H(1,0)$ for $i=1,\ldots,r$ are the characteristic vectors of the orbits of $\qut(H)$.
\end{lemma}
\begin{proof}
By \Lem{01201}, we have that $T_i$ and $T'_i$ have the same multiset of entries. By definition, $T_i$ is a 01-vector and thus $T'_i$ is a 01-vector with the same number of 1's as $T_i$. We will first show that $T'_i \schur T'_j = \delta_{ij} T'_i$. For $i=j$ this holds since $T'_i$ is a 01-vector. We have that 
\begin{align*}
T_i \schur T_j &= \sum_{\ell,k} \alpha_{i,\ell} \alpha_{j,k} T^{\vec{K}_{i,\ell} \schur \vec{K}_{j,k} \to G} \\
T'_i \schur T'_j &= \sum_{\ell,k} \alpha_{i,\ell} \alpha_{j,k} T^{\vec{K}_{i,\ell} \schur \vec{K}_{j,k} \to H}.
\end{align*}
Since $\vec{K}_{i,\ell} \schur \vec{K}_{j,k} \in \P(1,0)$ by \Lem{schurclosed}, we have that $T'_i \schur T'_j$ has the same multiset of entries as $T_i \schur T_j$ by \Lem{01201}. For $i \ne j$, this implies that $T'_i \schur T'_j$ is the zero vector. Thus we have shown that $T'_i \schur T'_j = \delta_{ij} T'_i$. Furthermore, $\sum_i T'_i$ must have the same multiset of entries as $\sum_i T_i$ which is the all ones vector, and thus the former must also be the all ones vector. 

By \Lem{orbitsasintertwiners}, $C_q^H(1,0)$ is the span of the characteristic vectors of the orbits of $\qut(H)$. Thus the above implies that each $T'_i$ is the sum of one or more characteristic vectors of orbits of $\qut(H)$. Moreover, each characteristic vector of the orbits of $\qut(H)$ appears in exactly one of these sums (since $\sum_i T'_i$ is the all ones vector). It follows that $\qut(H)$ has at least as many orbits as $\qut(G)$, but by symmetry they must then have the same number of orbits. It follows that the $T'_i$ are the characteristic vectors of the orbits of $\qut(H)$, and we are done.
\end{proof}

We now prove the main result of this section.

\begin{theorem}\label{thm:main}
Let $G$ and $H$ be graphs. Then the following are equivalent:
\begin{enumerate}
\item $G \cong_{qc} H$, i.e., there is a perfect quantum strategy for the $(G,H)$-isomorphism game;
\item there exists a weak isomorphism $\Phi\colon C_q^G \to C_q^H$ such that $\Phi(A_G) = A_H$;
\item $G \cong_\P H$, i.e., $G$ and $H$ admit the same number of homomorphisms from any planar graph.
\end{enumerate}
\end{theorem}
\begin{proof}
\Lem{qiso2weakiso} and \Lem{weakiso2piso} have already shown that $(1) \Rightarrow (2) \Rightarrow (3)$, so it only remains to prove that $(3)$ implies $(1)$. By \Lem{connectpreserve}, \Lem{complements}, and the fact that either a graph or its complement is connected, we may assume that both $G$ and $H$ are connected. Suppose that $R = \sum_{\ell=1}^{m} \alpha_\ell T^{\vec{K_\ell} \to G} \in C_q^G(1,0)$ is the characteristic vector of an orbit of $\qut(G)$. By \Lem{P01connected}, we may assume that $\vec{K_\ell} = (K_\ell,(a_\ell),\varnothing)$ where $K_\ell$ is connected for all $\ell$. By \Lem{qorbits}, we have that $R' = \sum_{\ell=1}^{m} \alpha_\ell T^{\vec{K_\ell} \to H} \in C_q^H(1,0)$ is the characteristic vector of an orbit of $\qut(H)$. 

Suppose that $u \in V(G)$ and $u' \in V(H)$ are such that $R_u = 1 = R'_{u'}$, i.e., $u$ is in the orbit of $\qut(G)$ whose characteristic vector is $R$ and $u'$ is in the orbit of $\qut(H)$ whose characteristic vector is $R'$. Let $X = G \cup H$ be the disjoint union of $G$ and $H$. We will show that
\[\hom((K,a),(X,u)) =\hom((K,a),(X,u'))\]
for all connected planar graphs $K$ and $a \in V(K)$, which will prove that $u$ and $u'$ are in the same orbit of $\qut(X)$ by \Cor{01qorbits}.

Let $K$ be a connected planar graph, $a \in V(K)$, and $\vec{K} = (K,(a),\varnothing) \in \P(1,0)$. Define $T = T^{\vec{K} \to G}$ and $T' = T^{\vec{K} \to H}$. Note that since $K$ is connected, if $\varphi$ is a homomorphism from $K$ to $X$ such that $\varphi(a) = u$, then the image of $\varphi$ is completely contained in $V(G)$. Similarly, if $\varphi(a) = u'$, then its image is completely contained in $V(H)$. Thus we have that $\hom((K,a),(X,u)) = \hom((K,a),(G,u)) = T_u$ and $\hom((K,a),(X,u')) = \hom((K,a),(H,u')) = T'_{u'}$. To determine $T_u$ and $T'_{u'}$ we consider
\begin{align*}
R \schur T &= \sum_{\ell=1}^{m} \alpha_\ell T^{\vec{K_\ell}\schur\vec{K} \to G} \\
R' \schur T' &= \sum_{\ell=1}^{m} \alpha_\ell T^{\vec{K_\ell}\schur\vec{K} \to H}
\end{align*}
Note that since $T \in C_q^G(1,0)$, it is constant on the orbits of $\qut(G)$, and similarly for $T'$. Thus $R \schur T = \alpha R$ and $R' \schur T' = \alpha'R'$ where $\alpha = T_u$ and $\alpha' = T'_{u'}$
By \Lem{schurclosed}, we have that $\K_\ell \schur \K \in \P(1,0)$ for all $\ell$, and thus we may apply \Lem{01201} to conclude that both $\text{sum}(R) = \text{sum}(R')$ and $\text{sum}(\alpha R) = \text{sum}(\alpha' R')$. Since $\text{sum}(R) \ne 0$, we have that $\alpha = \alpha'$, i.e., $T_u = T'_{u'}$. This implies that $\hom((K,a),(X,u)) =\hom((K,a),(X,u'))$ for all connected planar graphs $K$ and $a \in V(K)$, and thus $u$ and $u'$ are in the same orbit of $\qut(X)$. By \Thm{qisoorbits} this implies that $G \cong_{qc} H$ and we are done.
\end{proof}

It was shown in~\cite{qiso1} that the problem of determining whether two given graphs $G$ and $H$ are quantum isomorphic is undecidable, thus we have the following:

\begin{cor}\label{cor:undecidable}
The problem of determining whether for two input graphs $G$ and $H$ there exists a planar graph $K$ such that $\hom(K,G) \ne \hom(K,H)$ is undecidable.
\end{cor}

Given quantum isomorphic graphs $G$ and $H$, and a planar graph $K$, if $K$ has a homomorphism to $G$, then \Thm{main} tells us that $K$ must have a homomorphism to $H$. However, it does not tell us how we might construct the latter homomorphism from the former.

\vspace{.1in}

\noindent \textbf{Question:} Given graphs $G$ and $H$, a quantum permutation matrix $\mathcal{U}$ such that $\mathcal{U}A_G = A_H\mathcal{U}$, and a homomorphism $\varphi \colon K \to G$ for a planar graph $K$, is there some (natural) way of constructing a homomorphism $\varphi' \colon K \to H$? If we are given all homomorphisms from $K$ to $G$, can we construct all (or even one) homomorphisms from $K$ to $H$?

\section{Applications, Consequences, and Further Directions}\label{sec:discuss}

In this section we will introduce a new class of quantum groups inspired by our characterization of $C_q^G$, and discuss some further implications of our main results.

\subsection{Graph-theoretic quantum groups}\label{sec:graphqgroups}

In \Sec{closureprops} we showed that $C_\P^G$ is a tensor category with duals for any graph $G$. By Woronowicz' version of Tannaka-Krein duality given in \Thm{tannakakrein}, this already implies that $C_\P^G$ is the space of intertwiners for \emph{some} quantum group (and \Thm{onecat} shows that this quantum group is $\qut(G)$). Showing that $C_\P^G$ is a tensor category with duals only required us to prove that $\P$ is closed under the bi-labeled graph operations of composition, tensor product, and transposition, and that it contains $\vec{I}$ and $\vec{M}^{2,0}$ (the bi-labeled graphs whose homomorphism matrices are $I$ and $\xi$). This inspires the following definition:

\begin{definition}
Let $\mathcal{F} \subseteq \G$ be a family of bi-labeled graphs and let $\mathcal{F}(\ell,k) = \mathcal{F} \cap \G(\ell,k)$. We say that $\mathcal{F}$ is a \emph{graph category} if the following hold:
\begin{itemize}
\item $\vec{I},\vec{M}^{2,0} \in \mathcal{F}$;
\item $\vec{K} \in \mathcal{F}(\ell,k)$, $\vec{K}' \in \mathcal{F}(k,m)$ $\Rightarrow$ $\vec{K} \circ \vec{K}' \in \mathcal{F}(\ell,m)$;
\item $\K,\K' \in \mathcal{F}$ $\Rightarrow$ $\K \otimes \K' \in \mathcal{F}$;
\item $\K \in \mathcal{F}$ $\Rightarrow$ $\K^* \in \mathcal{F}$.
\end{itemize}
\end{definition}

For any graph category $\mathcal{F}$ and graph $G$, we define $C_\mathcal{F}^G(\ell,k) := \spn \{T^{\K \to G} : \K \in \mathcal{F}(\ell,k)\}$ and $C_\mathcal{F}^G := \cup_{\ell,k = 0}^\infty C_\mathcal{F}^G(\ell,k)$. From the correspondence between homomorphism matrices and bi-labeled graphs, we immediately have the following:

\begin{theorem}\label{thm:graphcats}
Let $\mathcal{F}$ be a graph category and $G$ a graph with $n$ vertices. Then $C_\mathcal{F}^G$ is a tensor category with duals. It follows that there exists a compact matrix quantum group $Q \subseteq O_n^+$ such that $C_Q = C_\mathcal{F}^G$.
\end{theorem}

The above theorem inspires the definition of a new class of compact quantum groups: We say that a compact matrix quantum group $Q \subseteq O_n^+$ is a \emph{graph-theoretic quantum group} if $C_Q = C_\mathcal{F}^G$ for some graph category $\mathcal{F}$ and graph $G$. This definition is in the spirit of the definition of easy quantum groups (recall that these are the quantum groups whose intertwiner spaces are spanned by maps indexed by partitions from a partition category). Note however that a partition category gives rise to a quantum group for every positive integer $n$, e.g.~the non-crossing partitions and non-crossing pairings give rise to $S_n^+$ and $O_n^+$, but a graph category gives rise to a quantum group for every graph $G$, e.g.~the graph category $\P$ gives rise to $\qut(G)$ for any $G$. This is one reason to expect that graph-theoretic quantum groups will be a significantly richer and more diverse class of quantum groups than easy quantum groups. The other reason is that we expect graph categories themselves to be more complex and numerous than partition categories. In support of this expectation we will now explain how to view partition categories and special cases of graph categories.

\paragraph{Partitions as bi-labeled graphs.} Recall from \Def{graphpartition} that we can associate a partition $\mathbb{P}_\K$ to any bi-labeled graph $\K = (K,\vec{a},\vec{b}) \in \G(\ell,k)$. For the purposes of that section, it was notationally convenient for $\mathbb{P}_\K$ to be a partition of $[\ell + k]$. Let us now consider $\mathbb{P}_\K$ as a partition of the totally ordered set $1_L < \ldots < \ell_L < k_U < \ldots < 1_U$. Thus $i_L, j_L$ are in the same part of $\mathbb{P}_\K$ if $a_i = a_j$ (and similarly for $i_U, j_U$ and $b_i,b_j$), and $i_L,j_U$ are in the same part if $a_i = b_j$. Also recall that $\mathbb{P}_\K$ has an empty part for each vertex of $K$ not appearing in either $\vec{a}$ or $\vec{b}$. Note that changing the underlying set of the partition does not change the validity of any statements such as \Lem{noncross}, since we have merely renamed the set elements. However, the new names do allow us to keep track of the values of $\ell$ and $k$, which is information that was loss before.

It is clear that the above is not a one-to-one correspondence between bi-labeled graphs from $\G(\ell,k)$ and partitions from $\mathbb{P}(\ell,k)$, since there is no dependence on the edges of $K$. However, if we restrict to edgeless bi-labeled graphs, then this correspondence is one-to-one. This is straightforward to see, since the inverse of this correspondence is essentially the construction given in \Def{partitiongraph}, though here we consider partitions of $\{1_L, \ldots, \ell_L, 1_U, \ldots, k_U\}$ instead of $[\ell + k]$. Thus a partition $\mathbb{P} = \{P_1, \ldots, P_r\} \in \mathbb{P}(\ell,k)$ is mapped to an (edgeless) bi-labeled graph $\K = (K,\vec{a},\vec{b}) \in \G(\ell,k)$ where $V(K) = [r]$, $a_i = s$ such that $i_L \in P_s$, and $b_j = t$ such that $j_U \in P_t$. Not only is this correspondence a bijection, it is routine to check that it commutes with the operations of composition, tensor product, and transposition of bi-labeled graphs and partitions respectively, i.e., that $\mathbb{P}_\K \circ \mathbb{P}_{\K'} = \mathbb{P}_{\K \circ \K'}$, $\mathbb{P}_\K \otimes \mathbb{P}_{\K'} = \mathbb{P}_{\K \otimes \K'}$, and $\mathbb{P}_\K^* = \mathbb{P}_{\K^*}$. Furthermore, the partitions $\{\{1_L,1_U\}\}$ and $\{\{1_L,2_L\}\}$ correspond to the bi-labeled graphs $\vec{I}$ and $\vec{M}^{2,0}$. Thus any partition category is isomorphic to a graph category containing only edgeless bi-labeled graphs, and vice versa. 

Now consider an edgeless bi-labeled graph $\K = (K,\vec{a},\vec{b}) \in \G(\ell,k)$, and let $G$ be any graph on $n$ vertices. Then, letting $r$ be the number of vertices of $K$ not appearing in $\vec{a}$ or $\vec{b}$,
\[(T^{\K \to G})_{u_1\ldots u_\ell, v_1\ldots v_k} = \begin{cases} n^r & \text{if } a_i = a_j \Rightarrow u_i = u_j \ \& \ b_i = b_j \Rightarrow v_i = v_j \ \& \ a_i = b_j \Rightarrow u_i = v_j \\ 0 & \text{o.w.}
\end{cases}\]
But this is precisely the map $T_\mathbb{P_\K}$ defined in \Sec{qautogroups}. Therefore, not only can partition categories be viewed as graph categories of edgeless bi-labeled graphs, but the corresponding tensor categories with duals, and thus corresponding quantum groups, are the same. Summarizing, we have the following:

\begin{theorem}
The construction $\K \mapsto \mathbb{P}_\K$ is a bijection between edgeless bi-labeled graphs and partitions that commutes with composition, tensor product, and transposition, and maps $\vec{I}$, $\vec{M}^{2,0}$ to $\{\{1_L,1_U\}\}$, $\{\{1_L,2_L\}\}$ respectively. Thus this map induces a bijection between graph categories of edgeless bi-labeled graphs and partition categories. Moreover, $T^{\K \to G} = T_{\mathbb{P}_\K}$\footnote{Here we should specify that this is the $T_{\mathbb{P}_\K}$ for $n = |V(G)|$.} for any edgeless bi-labeled graph $\K$ and graph $G$, and thus the graph-theoretic quantum groups arising from graph categories of edgeless bi-labeled graphs are precisely the easy quantum groups.
\end{theorem}

As one would expect that graph categories are a much richer class than graph categories consisting only of edgeless bi-labeled graphs, the final statement in the above theorem strongly supports our expectation that graph-theoretic quantum groups are a significantly more diverse class than easy quantum groups. However, like easy quantum groups, graph-theoretic quantum groups still have an underlying combinatorial structure that can be exploited for their study.

As an illustrative example, we will show how to recover Banica and Speicher's characterization of the intertwiners of $S_n^+$~\cite{noncross}. Recall that $S_n^+$ is defined as the compact matrix quantum group with fundamental representation $\mathcal{U}$ with the only constraint being that $\mathcal{U}$ is a quantum permutation matrix, and $\qut(G)$ is defined by adding the condition that $\mathcal{U}A_G = A_G\mathcal{U}$. However, if $G = \overline{K_n}$ is the empty graph then this additional constraint is trivial, and thus $\qut(\overline{K_n}) = S_n^+$. By \Thm{onecat}, we know that the intertwiner spaces of $S_n^+$ are spanned by the maps $T^{\K \to \overline{K_n}}$ for $\K \in \P$. However, as $\overline{K_n}$ has no edges, this map is zero unless $\K$ is an edgeless bi-labeled graph. Thus the intertwiner spaces of $S_n^+$ are spanned by the maps $T^{\K \to \overline{K_n}}$ for $\K \in \P$ such that $\K$ is edgeless. Finally, by \Lem{noncross} and \Lem{noncross2}, we have that if $\K$ is edgeless, then $\K \in \P$ if and only if $\mathbb{P}_\K$ is a non-crossing partition. Thus we obtain Banica and Speicher's result   showing that the intertwiners of $S_n^+$ are spanned by the maps $T_\mathbb{P}$ where $\mathbb{P}$ is a non-crossing partition.

\paragraph{The category of all bi-labeled graphs.} In this work we have mainly focused on the graph category $\P$, whose corresponding family of quantum groups is $\qut(G)$ for graphs $G$. Of course, the class $\G$ of all bi-labeled graphs is also a graph category, and thus corresponds to a family of quantum groups. We will see that it in fact corresponds to the family of classical automorphism groups of graphs $\aut(G)$. To do this we will show that any bi-labeled graph can be generated by the bi-labeled graphs $\vec{M}^{1,0},\vec{M}^{1,2}, \vec{S}$, and $\vec{A}$. We just provide a sketch of the proof since it is relatively straightforward in comparison to \Thm{onegraphcat}.

\begin{theorem}\label{thm:allgraphcat}
$\G = \langle \vec{M}^{1,0},\vec{M}^{1,2}, \vec{S}, \vec{A}\rangle_{\circ, \otimes, *}$
\end{theorem}
\begin{proof}
Obviously, we have that $\langle \vec{M}^{1,0},\vec{M}^{1,2}, \vec{S}, \vec{A}\rangle_{\circ, \otimes, *} \subseteq \G$. For the other containment, we will describe how to construct  an arbitrary bi-labeled graph $\K = (K,\vec{a},\vec{b}) \in \G$. We begin with the bi-labeled graph $\K' = (K',\vec{a}',\vec{b}')$ where $K'$ is an empty graph on the same vertex set as $K$, and $\vec{a}' = \vec{b}' = (v_1, \ldots, v_n)$ is any ordering of $V(K)$ (i.e., each vertex of $K$ appears precisely once). Note that $\K'$ is isomorphic to $\vec{I}^{\otimes n}$, and thus $\K' \in \langle \vec{M}^{1,0},\vec{M}^{1,2}, \vec{S}, \vec{A}\rangle_{\circ, \otimes, *}$. Let $\hat{\vec{A}} = (K_2,(a,b),(a,b))$ where $a$ and $b$ are the two vertices of the complete graph $K_2$. It is easy to see that $\hat{\vec{A}} = (\vec{I} \otimes \vec{M}^{1,2}) \circ (\vec{I} \otimes \vec{A} \otimes \vec{I}) \circ (\vec{M}^{2,1} \otimes \vec{I})$. By multiplying by $\vec{I}^{\otimes r} \otimes \hat{\vec{A}} \otimes \vec{I}^{n-r-2}$ we can add an edge between $v_{r+1}$ and $v_{r+2}$ for $r = 0, \ldots, n-2$. Moreover, we can permute the elements of $\vec{a}'$ or $\vec{b}'$ arbitrarily through compositions with $\vec{I}^{\otimes r} \otimes \vec{S} \otimes \vec{I}^{n-r-2}$. Combining these two, we can add edges between arbitrary pairs of vertices, thus constructing the graph $K$. Following this, we can ensure that each vertex of $K$ appears in the input and output vectors the correct number of times by using $\vec{M}^{1,2}$, $\vec{M}^{1,0}$, and their transposes. Finally, using $\vec{S}$ again we can permute the input and output vectors so that they are equal to $\vec{b}$ and $\vec{a}$ respectively. Thus we have constructed $\K$ and we are done.
\end{proof}

Recall from \Thm{chassaniol}, that Chassaniol showed that the intertwiners of the classical automorphism group of $G$ are given by $C^G = \langle M^{1,0},M^{1,2},A_G, S\rangle_{+,\circ, \otimes, *}$. Thus, by the above theorem we have the following:

\begin{theorem}
For any graph $G$,
\[C^G = \spn\{T^{\K \to G} : \K \in \langle \vec{M}^{1,0},\vec{M}^{1,2}, \vec{S}, \vec{A}\rangle_{\circ, \otimes, *}\} = \spn\{T^{\K \to G} : \K \in \G\}.\]
\end{theorem}

\paragraph{Discovering new graph categories.} Though we have just introduced the notion of graph categories and graph-theoretic quantum groups, it is clear that one of the end goals should be to give a complete classification, analogous to the recent classification of easy quantum groups~\cite{Raum2016}. The difficulty of establishing such a classification is quite unclear at the moment, perhaps it is even impossible. As a hint towards its difficulty, it is even unclear what graph categories can arise as $\langle \vec{M}^{1,0},\vec{M}^{1,2}, \vec{K}\rangle_{\circ, \otimes, *}$ or $\langle \vec{M}^{1,0},\vec{M}^{1,2}, \vec{S}, \vec{K}\rangle_{\circ, \otimes, *}$ for some bi-labeled graph $\K$, or even if only finitely many arise in this way (this is not even clear for  $\K = \vec{A}^r$). For now, it will likely be enlightening to find, and give good descriptions of, new graph categories. Perhaps the most natural place to start is to consider a partition category as a graph category, add the bi-labeled graph $\vec{A}$, and determine what graph category is generated. In terms of the corresponding compact matrix quantum groups, this corresponds to adding the requirement that the fundamental representation commutes with the adjacency matrix of a graph, i.e., $\mathcal{U}A_G = A_G \mathcal{U}$. In this work we have done this for the class of non-crossing partitions and the class of all partitions in \Thm{onegraphcat} and \Thm{allgraphcat} respectively.

Somewhat related to the search for new graph categories is the question of what compact matrix quantum groups arise as graph-theoretic quantum groups. Is it possible that any CMQG $Q \subseteq O_n^+$ can be obtained in this way (up to isomorphism)? This seems unlikely, but perhaps should not be ruled out immediately. In the classical case there is Frucht's Theorem~\cite{frucht}: that any finite group is isomorphic to the automorphism group of some graph. Finite groups are exactly the subgroups of the symmetric groups $S_n$, and thus a potential quantum analog of this could be that any quantum subgroup of $S_n^+$ is isomorphic to the quantum automorphism group of some graph. If this is the case then our \Thm{onecat} characterizes the intertwiners of any quantum subgroup of $S_n^+$\footnote{We thank Moritz Weber for pointing this out.}.

\subsection{Multiedges, directed edges, and colored edges}\label{sec:colorededges}

In this work we mainly restricted our attention to undirected graphs and bi-labeled graphs without any multiple edges. One reason for this is that most of the previous work on quantum automorphism groups of graphs (including our own) focuses on the case of undirected simple graphs. But the main reason is merely to simplify the presentation of our results. In fact, allowing multiple and/or directed edges does not really make a fundamental difference, and the lemmas and theorems we have proven can be rather straightforwardly adapted to the more general case. Here we will discuss some of the minor changes one would need to make to accommodate these more general settings.

If we allow multiple edges in our bi-labeled graphs, then we should no longer perform the simplification step when constructing the composition or Schur products of two bi-labeled graphs. Note that, as mentioned in \Rem{bilabeledmultigraphs}, if we are considering $\qut(G)$ where $G$ has no multiple edges, then multiple edges in a bi-labeled graph does not change the corresponding $G$-homomorphism matrix. Note that if $G$ has multiple edges, then the $uv$-entry of $T^{\vec{A} \to G}$ is the number of edges between $u$ and $v$, which fits the usual definition of the adjacency matrix of a multigraph. The proofs of the correspondence between bi-labeled graph operations and matrix operations go through essentially unchanged when multiple edges are allowed. Furthermore, the proof that $\P$ (now allowing multiple edges) is closed under composition, tensor product, and transposition does not change. The main difference in proving $\langle \vec{M}^{1,0},\vec{M}^{1,2}, \vec{A}\rangle_{\circ, \otimes, *} = \P$ in the multiedge case is that one must allow for more general versions of the bi-labeled graphs $\vec{S}^{m,d}, \vec{S}_R^{m,d}, \vec{S}_L^{m,d}$, etc. from \Lem{starmaps}. In particular, now one must specify a multiplicity for each edge in the underlying graphs. But this mostly just makes the presentation more cumbersome, rather than actually making the proof more difficult.

As far as we are aware, there is no established definition of the quantum automorphism group of a multigraph. One possibility is to define it as for simple graphs, by taking the quantum symmetric group and adding the condition that $\mathcal{U}A_G = A_G\mathcal{U}$, where $A_G$ is the adjacency matrix of the multigraph $G$ as mentioned above. This definition yields a quantum group whose intertwiners are the $G$-homomorphism matrices of the bi-labeled graphs from $\P$ where multiple edges are allowed. However, such a quantum automorphism group only considers action on the vertices, not the edges. In analogy, in the classical case one may define automorphisms of multigraphs as permutations of the vertices preserving multiplicities of edges between pairs, or they may also include a permutation on edges (that must agree with the vertex permutation). Characterizing the intertwiners of the latter type of automorphism group (classical or quantum) may be possible by also allowing edges to appear in the input/output vectors of bi-labeled graphs. But we do not speculate further about how precisely to do this.

In the case of directed graphs, the main difference is that we would need to consider directed bi-labeled graphs. The class $\P$ would be defined the same except edges would be replaced by directed edges (also known as \emph{arcs}), and we would permit edges going in both directions between two vertices. The definition of the bi-labeled graph $\vec{A}$ would be changed so that its underlying graph is now the directed graph on two vertices with a single arc between them (either from input to output or vice versa depending on your convention). Again, very little changes in the proof of \Thm{onegraphcat} and thus \Thm{onecat}. As with multigraphs, the main difference is that the bi-labeled graphs from \Lem{starmaps} would need to allow for the edges of the underlying graph to have different directions, thus increasing the number of possibilities and making the presentation a bit more cumbersome. Otherwise, everything goes through essentially the same as in the undirected case, including the characterization of quantum isomorphism in terms of counting homomorphisms from planar graphs (though now it will be directed planar graphs). The definition of quantum isomorphism for directed graphs is precisely what one would expect.

Finally, we may want to allow our graphs to have colored/labeled arcs, typically labeled by positive integers. In this case, for any ordered pair $(u,v)$ of vertices, we would allow there to be at most one arc from $u$ to $v$. Note that, if we are only interested in (quantum) group actions on vertices, there is no need for multiple arcs in this setting, since we can replace multiple arcs with a single arc labeled according to the multiset of labels occurring on the multiple arcs. We may also allow for colored vertices, but this is functionally the same as having colored loops. In this setting the underlying graphs of our bi-labeled graphs would also have colored arcs, and the entries of the homomorphism matrices would count homomorphisms that map arcs of color $i$ to arcs of color $i$. The class $\P$ would now be defined as before, but edges replaced by colored arcs. The bi-labeled graph $\vec{A}$ would need to be replaced by a set of bi-labeled (directed) graphs $\vec{A}_i$ for each arc color $i$. Here $\vec{A}_i$ is the same as the directed version of $\vec{A}$ we described above, but the arc in the underlying graph is colored $i$. It may be convenient to allow our bi-labeled graphs to have arcs colored by any positive integer (thus there would be an infinite number of $\vec{A}_i$'s). This removes a parameter (the number of arc colors) from the description of any arc-colored graph category, and allows us to deal simultaneously with the intertwiners of $\qut(G)$ for $G$ with any number of colors on its arcs. If an arc-colored bi-labeled graph contains an arc of a color not appearing among the arcs of $G$, then its corresponding $G$-homomorphism matrix will be zero.

The quantum automorphism group of an arc-colored graph $G$ can be defined similarly to the quantum automorphism group of a graph. The difference is that the condition $\mathcal{U}A_G = A_G\mathcal{U}$ now becomes $\mathcal{U}A_i = A_i\mathcal{U}$ for all $i$, where $A_i$ is the adjacency matrix of the directed graph consisting of the arcs of color $i$. This condition is saying that we require each $A_i$ to be a $(1,1)$-intertwiner of our quantum automorphism group. Here we begin to see that allowing colored arcs is not just a generalization that allows us to consider a larger class of objects, it is useful even when studying $\qut(G)$ for simple graphs $G$. By this we are referring to the fact that for uncolored graphs $G$, the $(1,1)$-intertwiners form a \emph{coherent algebra}, as proven in~\cite{qperms}. This simply means that the $(1,1)$-intertwiners form a unital matrix algebra containing the all ones matrix that is additionally closed under entrywise product and conjugate transposition. It follows that there is a unique linear basis of $C_q^G(1,1)$ consisting of $01$-matrices $A_1, \ldots, A_d$. Considering $A_i$ as the adjacency matrix of a directed graph $G_i$, we see by Tannaka-Krein duality (\Thm{tannakakrein}) that $\qut(G)$ is also the quantum automorphism group of the arc-colored directed graph $\hat{G}$ formed by coloring the arcs of $G_i$ by color $i$ for $i = 1,\ldots,d$. The arc sets of the $G_i$ are the orbitals of $\qut(G)$, defined in \Def{qorbits}. These form a \emph{coherent configuration}, as show in~\cite{qperms} (in fact coherent configurations are in one-to-one correspondence with coherent algebras). The point here is that the quantum automorphism group of a graph \emph{is} the quantum automorphism group of a coherent configuration, though the converse is not true in general.

This is useful for several reasons. First, there appear to be remarkably fewer coherent configurations than graphs. For instance, on 15 vertices there are approximately $3\cdot 10^{19}$ isomorphism classes of graphs, but only about $10^4$ coherent configurations~\cite{nonschurian}. Thus if one wants to classify all quantum automorphism groups of graphs up to $n$ vertices, a great deal of time can be saved by simply considering quantum automorphism groups of coherent configurations. Another reason this is useful is that it provides a more efficient description of the intertwiners of $\qut(G)$. Here we are referring to the fact that the intertwiner $A_i$ is likely not the $G$-homomorphism matrix of a single planar bi-labeled graph, but rather a linear combination of several such $G$-homomorphism matrices. However, in the arc-colored graph formalism, the intertwiner $A_i$ is simply the $\hat{G}$-homomorphism matrix of the bi-labeled graph $\vec{A}_i$. One obstacle to applying this in practice is that determining the orbitals of $\qut(G)$ (i.e., determining the $A_i$) is undecidable in general~\cite{qperms}. However, it is useful even to consider a coarse-graining of the orbitals, i.e., a partition of $V(G) \times V(G)$ in which each part is a union of orbitals of $\qut(G)$. One way to do this in polynomial time is to use the Weisfeiler-Leman algorithm to find the coarsest coherent configuration ``compatible" with the graph $G$. 

\subsection{Quantum isomorphism and the Four Color Theorem}

The Four Color Theorem is perhaps the most well-known theorem in graph theory. Originally proven by Appel and Haken~\cite{4color}, it states that the vertices of any planar graph (without loops) can be colored with four colors so that no two adjacent vertices are colored the same. It has many equivalent reformulations, and here we will add one more to the list.

A trivial restatement of the Four Color Theorem is that any loopless planar graph $G$ has a homomorphism to $K_4$, since an $n$-coloring is precisely a homomorphism to $K_n$. We will show that the four color theorem is equivalent to a statement which {\it{a priori}} appears to be weaker. The statement will have the form ``every planar graph has a homomorphism to $H$", where $H$ is a specific graph which is in general easier to have a homomorphism to than $K_4$. By this we mean that \emph{any} graph (not only planar graphs) that has a homomorphism to $K_4$ has a homomorphism to $H$, but there are graphs which have homomorphisms to $H$ but not to $K_4$. As an example, if we could prove this for $H = K_5$, then we would have a ``human" proof of the Four Color Theorem, as proving the five color theorem is straightforward.

The graph $H$ for which we prove this reformulation is the Cayley graph for the symmetric group $S_4$ whose connection set is the set of all elements of order two (these are the transpositions and the ``double transpositions": the permutations $(ab)(cd)$ for $\{a,b,c,d\} = \{1,2,3,4\}$). This means that $V(H) = S_4$, and two vertices $u,v \in S_4$ are adjacent if $uv^{-1}$ is in the connection set. This graph is quantum isomorphic to a graph $G$ which is the Cayley graph for $S_4$ with connection set consisting of all double transpositions and all 4-cycles. Though they were not described as Cayley graphs, it is easy to verify that these are precisely the two graphs on 24 vertices shown to be quantum isomorphic in~\cite{qiso1} via a reduction from the Mermin magic square game.

It follows from our characterization of quantum isomorphism (\Thm{main}) that a planar graph $K$ has a homomorphism to $G$ (i.e., $\hom(K,G) > 0$) if and only if it has a homomorphism to $H$ (i.e., $\hom(K,H) > 0$). On the other hand, the graph $G$ is \emph{homomorphically equivalent} to $K_4$, i.e, there are homomorphisms in both directions between $G$ and $K_4$\footnote{A homomorphism from $K_4$ to $G$ is given by mapping the four vertices of the former to the identity and the three double transpositions. A 4-coloring of $G$ is given by the cosets of the subgroup generated by $\{(123), (23)\}$.}. By composition of homomorphisms, it follows that \emph{any} graph has a homomorphism to $G$ if and only if it has a homomorphism to $K_4$. Combining this with the previous equivalence, we see that any \emph{planar} graph has a homomorphism to $K_4$ if and only if it has a homomorphism to $H$. We formally state this as a theorem below:

\begin{theorem}\label{thm:4color}
A planar graph has a 4-coloring if and only if it has a homomorphism to the Cayley graph on the symmetric group $S_4$ whose connection set consists of all elements of order two.
\end{theorem}

Of course the above theorem is only interesting if the graph $H$ does not itself have a 4-coloring. This is indeed the case: the graph $H$ has chromatic number 5 (and fractional chromatic number $24/5$). Furthermore, $H$ does have a homomorphism from $K_4$. Thus any graph with a homomorphism to $K_4$ has a homomorphism to $H$, but the converse does not hold, e.g., $H$ has a homomorphism to itself but no homomorphism to $K_4$ (and there are infinitely many examples even up to homomorphic equivalence).

Thus \Thm{4color} achieves our goal: the Four Color Theorem is equivalent to every planar graph having a homomorphism to $H$, and it is in general easier to have a homomorphism to $H$ than to have a 4-coloring. It is not clear if this is of any practical use, as proving that all planar graphs have a homomorphism to the 24-vertex graph $H$ does not \emph{seem} any easier than proving they have a 4-coloring. 

By considering the requirements $H$ must satisfy in order to obtain such a reformulation of the Four Color Theorem, we obtain the following lemma:
\begin{lemma}\label{lem:4color}
Suppose that $G$ and $H$ are quantum isomorphic graphs such that $G$ is homomorphically equivalent to $K_4$. Then any planar graph $K$ has a 4-coloring if and only if it has a homomorphism to $H$.
\end{lemma}
As in the case of \Thm{4color}, a graph $H$ fitting the hypotheses of \Lem{4color} only provides an interesting reformulation of the Four Color Theorem if it does not have a 4-coloring (it will always have a homomorphism \emph{from} $K_4$ since $G$ does and $K_4$ is planar). To find examples of such graphs $H$ (other than the one presented above), one can make use of a reduction from linear system games presented in~\cite{qiso1}. Combining this reduction with a construction of Arkhipov~\cite{arkhipov} for building binary linear system games from graphs, one can construct a pair of (non-isomorphic) quantum isomorphic graphs from any connected non-planar graph. The example above comes from applying this construction to the complete bipartite graph $K_{3,3}$. If one uses a connected non-planar graph $K$ that is cubic (regular of degree three), then the resulting quantum isomorphic graphs $G_K$ and $H_K$ will always have a homomorphism from $K_4$. Moreover, one of them (say $H_K$) is guaranteed to not have a 4-coloring. Thus if $K$ is chosen so that $G_K$ has a 4-coloring, then we can apply \Lem{4color} to obtain a nontrivial reformulation of the Four Color Theorem. Unfortunately, it is not immediately clear what properties one needs to require of the graph $K$ in order to ensure that $G_K$ is 4-colorable. However, we expect that there should be many such $K$ that result in $G_K$ being 4-colorable. Other than $K = K_{3,3}$ which we described above, we have verified by computer that the graph $G_K$ obtained when $K$ is the Wagner graph is 4-colorable. In this case, the graph $H_K$ does not have a homomorphism to or from the graph $H_{K_{3,3}}$ described above, and thus we obtain another reformulation of the Four Color Theorem that is not implied by, nor implies, the reformulation given by \Thm{4color}. It seems possible that infinitely many such reformulations could be found that are all mutually incomparable.

\paragraph{Acknowledgements.} \textcolor{white}{a} \\
\textbf{LM:} Supported by the Villum Fonden via the QMATH Centre of Excellence (Grant No. 10059).

\noindent \textbf{DR:} I would like to acknowledge the helpful discussions with Arthur Chassaniol concerning intertwiners of quantum automorphism groups which took place at the 'Cohomology of quantum groups and quantum automorphism groups of finite graphs' Workshop at Saarland University in October of 2018. The research leading to these results has received funding from the European Union's Horizon 2020 research and innovation program under the Marie Sklodowska-Curie grant agreement no. 713683 (COFUNDfellowsDTU).

\bibliographystyle{plainurl}
\bibliography{Qplanar}
%

\end{document}